\documentclass[11pt,a4paper,oneside]{amsart}
\usepackage[a4paper,margin=2.5cm]{geometry}
\usepackage[dvipsnames]{xcolor}
\usepackage[hypertexnames=false,colorlinks=true,linkcolor=blue,%
citecolor=purple,filecolor=magenta,urlcolor=cyan]{hyperref}                      
\usepackage{amssymb,anyfontsize,enumerate,layout,mathrsfs,mdwlist,stmaryrd,thmtools,tikz}
\usepackage{appendix}
\usepackage{bbm} 
\usepackage{dsfont} 
\usepackage[new]{old-arrows}
 
\usepackage[textsize=footnotesize]{todonotes}
\presetkeys%
    {todonotes}%
    {color=Apricot}{}%
\usepackage{csquotes} 
\usepackage[T1]{fontenc}
\usepackage[scaled]{helvet}     
\usepackage{courier}            
\usepackage{eulervm}            
\usepackage{microtype} 
\usepackage[style=alphabetic,maxnames=99,maxalphanames=5]{biblatex}
\usepackage{mathtools}
\usepackage[english,noabbrev]{cleveref}
 
\allowdisplaybreaks 

\theoremstyle{plain}                          
\newtheorem{theorem}{Theorem}[section]
\newtheorem{proposition}[theorem]{Proposition}    
\newtheorem{lemma}[theorem]{Lemma}
\newtheorem{corollary}[theorem]{Corollary}
    
\theoremstyle{definition}
\newtheorem{definition}[theorem]{Definition}
\newtheorem{prop-defin}[theorem]{Proposition-definition} 
\newtheorem{example}[theorem]{Example}
\theoremstyle{remark}
\newtheorem{remark}[theorem]{Remark}

\numberwithin{equation}{section}

\renewcommand{\theta}{\vartheta}
\renewcommand{\phi}{\varphi}
\renewcommand{\epsilon}{\varepsilon}
\newcommand{\normord}[1]{\vcentcolon\mathrel{\mspace{2mu}#1\mspace{2mu}}\vcentcolon}
\DeclarePairedDelimiter\floor{\lfloor}{\rfloor}
\DeclarePairedDelimiter\ceil{\lceil}{\rceil}

\newcommand{\mb}[1]{\mathbbm{#1}} 

\newcommand{\mc}[1]{\mathcal{#1}}

\newcommand{\su}[1]{\mathbbm{#1}}

\newcommand{\R}{\mb{R}} 
\newcommand{\N}{\mb{N}} 
\newcommand{\C}{\mb{C}} 
\newcommand{\Z}{\mb{Z}} 

\renewcommand{\P}{\mb{P}}

\newcommand{\<}{\langle}
\renewcommand{\>}{\rangle}
\newcommand{\del}{\partial}

\DeclareMathOperator{\im}{Im}
\DeclareMathOperator{\Ker}{Ker}

\newcommand{\Id}{\mathord{\mathrm{Id}}}
\DeclareMathOperator*{\Res}{Res}
\DeclareMathOperator{\Spec}{Spec}

\DeclareMathOperator{\Ber}{Ber}
\DeclareMathOperator{\Jac}{Jac}
\DeclareMathOperator{\ad}{ad}

\begingroup\expandafter\expandafter\expandafter\endgroup
\expandafter\ifx\csname pdfsuppresswarningpagegroup\endcsname\relax
\else
\fi
{}

\begin{document}

\title{Superconformal topological recursion}

\author[N.~Aghaei]{Nezhla Aghaei}
\address[N.~A.]{Section of Mathematics, University of Geneva, Rue du Conseil-General 7-9, 1205 Geneva, Switzerland}
\email[N.~A.]{nezhla.aghaei@unige.ch}

\author[R.~Kramer]{Reinier Kramer}
\address[R.~K.]{Universit\`{a} di Milano-Bicocca, Dipartimento di Matematica e Applicazioni, Via Roberto Cozzi, 55, Milan, 20125, Italy}

\address[R.~K.]{INFN sezione di Milano-Bicocca, Milano, Italy}
\email[R.~K.]{reinier.kramer@unimib.it}

\author[N.~Orantin]{Nicolas Orantin}
\address[N.~O.]{}
\email[N.~O.]{nicolas.orantin@gmail.com}

\author[K.~Osuga]{Kento Osuga}
\address[K.~O.]{Nagoya University, Kobayashi--Maskawa Institute for the Origin of Particles and the Universe \& Graduate School of Mathematics, Furocho, Nagoya, Aichi, 464-8602, Japan}
\email[K.~O.]{osuga@math.nagoya-u.ac.jp}

\thanks{}

\begin{abstract}

We investigate a supersymmetric generalisation of topological recursion from two perspectives: algebraic and geometric. The algebraic side concerns a recursive structure encoded in modules of a super Virasoro algebra, and the geometric counterpart is what we call superconformal topological recursion defined on a super Riemann surface. Superconformal topological recursion indicates that odd holomorphic one-forms on a super Riemann surface are related to zero modes of copies of the Clifford algebras, and it also provides a tool to study deformation of non-split super Riemann surfaces, e.g. certain families of super Riemann surfaces carrying odd parameters. On a super Riemann surface over a non-reduced base, the formalism is recursive not only in terms of pants-decomposition but also in terms of odd parameters in a suitable sense.
\end{abstract}

\maketitle

\tableofcontents


\section{Introduction}

The main aim of the present article is to introduce the formalism of \emph{superconformal topological recursion} and investigate its properties. The setting is fully superconformal in the sense that the initial data and the outputs are both geometric objects on a super Riemann surface, which by definition is a $ 1|1$-dimensional space with a superconformal structure. Superconformal topological recursion is closely related to modules of the $\mathcal{N}=1$ super Virasoro algebra, providing a connection between two notions of superconformality, geometric and algebraic.

\subsection{Motivation and backgrounds}

\emph{Topological recursion} \cite{EO07} is a universal recursive formalism which has applications in a variety of subjects in mathematics and physics. The set of initial data, called a \emph{spectral curve} $\mathcal{S}_{\Sigma_0}$, consists of a Riemann surface $\Sigma_0$ and a few more geometric objects. The output is a sequence of multidifferentials $(\omega_{g,n})$ labelled by a pair of two non-negative integers $(g,n)$, each of which is a symmetric meromorphic section of $K_{\Sigma_0}^{\boxtimes n}$. The defining formula of $\omega_{g,n}$ admits an interpretation in terms of pants-decomposition, and several generalisations have been considered, e.g. \cite{BE13,BS17}.

The range of applications is remarkably broad, well beyond the original scope in terms of matrix models \cite{CE05}, and it is yet increasing over time. To give a few, topological recursion is strongly linked to intersection theory on moduli spaces of stable curves and cohomological field theories \cite{NS11,DOSS14,BCGS25}, including a notable achievement known as \emph{remodeling theorem} conjectured by \cite{BKMP07,BKMP08} and proved in \cite{EO12,FLZ16}. It also describes a recursive structure of Hurwitz numbers and integrable hierarchies \cite{EMS09,ACEH18} recently recaptured in great generality by \cite{ABDKS24,ABDKS24a}. Furthermore, topological recursion has been used as a tool of quantisation \cite{BE17,Iwa19,EGMO24,BBKN25} creating a bridge with exact WKB analysis and associated Donaldson--Thomas invariants \cite{IK20}.

In 2017, the notion of \emph{Airy structures} \cite{KS18,ABCO24} was introduced. This can be thought of as an algebraic reformulation (and generalisation) of topological recursion. For an index set $I$ and a vector space $V_0$, the formalism starts with an \emph{Airy ideal}  $\mathcal{I}=(H_i)_{i\in I}$, which is a set of differential operators $H_i$ satisfying certain conditions, and it gives rise to a unique sequence of functions $F_{g,n}$ on $V_0^{\otimes n}$. It is a reformulation in the sense that given a spectral curve $\mathcal{S}_{\Sigma_0}$, one can construct the corresponding Airy ideal $\mathcal{I}_\mathcal{S}$, and for this Airy ideal, multidifferentials $\omega_{g,n}$ on $\Sigma_0^n$ are related to functions $F_{g,n}$ on $V_0^{\otimes n}$ by an explicit formula. We note that the converse is not always true, i.e. not all Airy ideals come from topological recursion.

A crucial point is that the Airy ideal $\mathcal{I}_\mathcal{S}$ related to a certain spectral curve $\mathcal{S}_{\Sigma_0}$ can be derived also in terms of so-called twisted modules of the Virasoro algebra. In particular, differential operators $H_i$ of the associated Airy ideal $\mathcal{I}_\mathcal{S}$ induce Virasoro constraints on the generating function of the functions $F_{g,n}$. As a consequence, in terms of Airy structures, one can understand a structural correspondence between topological recursion on a spectral curve $\mathcal{S}_{\Sigma_0}$ and twisted modules of the Virasoro algebra. The perspective of Airy structures has triggered a series of programmes (e.g. \cite{BBCCN18,BBCC21,BKS23}) connecting topological recursion with modules of $\mathcal{W}(\mathfrak{gl}_r)$-algebras in general --- it reduces to the Virasoro algebra when $r=2$. Furthermore, the technique of Airy structures has led to a different proof of the correspondence between topological recursion and certain Hurwitz numbers \cite{CDO24-1}.

Besides $\mathcal{W}(\mathfrak{gl}_r)$-algebras, the Virasoro algebra enjoys another algebraically natural extension, supersymmetric Virasoro algebras, which play a central role in superstring theory. Motivated by developments in topological recursion and Airy structures, one may ask the following questions:
\begin{itemize}
\item Can we incorporate the $\mathcal{N}=1$ super Virasoro algebra into topological recursion, defined on some super spectral curve?
\item Is such a supersymmetric recursion related to enumerative invariants of moduli spaces of super Riemann surfaces?
\end{itemize}
The first idea was put forward by a series of works \cite{BO18,BCHORS20,BO21,Osu21}. Building on their previous attempts, the present article brings major progress towards the first point, though the second point is to be investigated and we leave it to future work.

\subsection{Main results}

Our aim poses significant problems. Topological recursion in the `ordinary', i.e. non-super, setting relies on a lot of standard theory of Riemann surfaces, as developed over the last 160 or so years, from the original descriptions by Riemann as extensions of solutions of equations, Mittag--Leffler's reconstruction of functions/forms from principal parts and the theory of homology and (de Rham/Betti) cohomology to theta functions and prime forms of \cite{Fay73,Mum83,Mum84}. While many of these constructions have been adapted to super Riemann surfaces, in many cases this was done in a setting (namely purely Neveu--Schwarz super Riemann surfaces) that is too restrictive for our needs. The more general setting of super Riemann surfaces with Ramond divisor, which we require, has only gained more attention recently \cite{Wit15,DP15,Wit19,DO23}, and hence many tools have not yet been developed (or had not when we started this project).

Therefore, we take time and space to develop the theory of super Riemann surfaces we need. Most of these results were developed before, but this part of the paper does also contain new results. In the current phase of the research of super Riemann surfaces, giving a consistent definition itself can be considered as a part of our results. Among such, objects that deserve attentions are the super Givental pairing in \cref{SuperGiventalPairing} and the super quadratic Casimir operators defined in \cref{SuperQuadraticCasimir,TwistedSuperQuadraticCasimir}, both defined on a super Riemann surface $\Sigma$. The former is necessary to establish a clear relation to the free fermion algebra. The latter shows how to construct a super analogue of quadratic differential from a super analogue of a differential on $\Sigma$ -- it is too naive to take a square in the super setting. 

We then turn to the algebraic side and focus on combining two notions, super Airy structures due to \cite{BCHORS20} and partial Airy structures due to \cite{BCJ22}, into what we call \emph{partial super Airy structures}. The concept mostly remains similar to the ordinary setting, and there exists a sequence of associated functions $F_{g,n}$ (unique up to some choices) on a super vector space $V$. Major differences are that the associated vector space $V$ as well as differential operators are $\mathbb{Z}_2$-graded in the supersymmetric sense, and that one needs extra data for uniqueness for the functions $F_{g,n}$. Because it involves new notions and new definitions which can be stated without ambiguities only in the main text we shall not precisely state our main theorem in this introduction. Nonetheless, one of  our main theorems (\cref{thm:SASmain}) can be briefly summarised as below.

\begin{theorem}\label{thm:intro1}
    There exists certain partial super Airy ideals $\mathcal{I}_{\rm STR}^{\,r}$ whose generators are constructed by $r$ copies of twisted modules of the $\mathcal{N}=1$ super Virasoro algebra. Furthermore, the generating function of the functions $F_{g,n}$ obey super Virasoro constraints.
\end{theorem}

Our next task is to translate algebraic objects $\mathcal{I}_{\rm STR}^{\,r}$ and $F_{g,n}$ into geometric ones on a super Riemann surface. To be more concrete, we define a \emph{super spectral curve} $\mathcal{S}_\Sigma$ which consists of a super Riemann surface $\Sigma$ with a few more geometric objects carrying information about the partial super Airy structure $\mathcal{I}_{\rm STR}^{\,r}$. Then, we would like to define multidifferentials $\omega_{g,n}$ on $\Sigma$ from which one can recover the functions $F_{g,n}$ of $\mathcal{I}_{\rm STR}^{\,r}$ in a suitable manner. Again, a precise statement about our second main result (\cref{thm:SLE}) is hard to write here as it involves new notions, but it can be stated as below to address our main idea:
\begin{theorem}\label{thm:intro2}
Given a partial super Airy structure $\mathcal{I}_{\rm STR}^{\,r}$, there exists an associated super spectral curve $\mathcal{S}_\Sigma$ and multidifferentials $\omega_{g,n}$ on $\Sigma$ that encode information about the functions $F_{g,n}$ associated to $\mathcal{I}_{\rm STR}^{\,r}$. Furthermore, the differentials $\omega_{g,n}$ satisfy a set of constraints called \emph{super loop equations} which are geometric counterparts of super Virasoro constraints.
\end{theorem}

Although \cref{thm:intro1} and \cref{thm:intro2} guarantee existence of such functions $F_{g,n}$ or equivalently multidifferentials $\omega_{g,n}$, they do not say how to find them in practice. Indeed, topological recursion of \cite{EO07} can be thought of as a way of solving ordinary loop equations. Thus, our final task is to provide a way of solving super loop equations by using the geometry of a super spectral curve $\mathcal{S}_\Sigma$. To this end, it turns out that it is more convenient to introduce a new parameter $\gamma$ which intuitively measures the degree of functions in odd parameters. As long as there are only finitely many independent odd parameters (which is our setting), we can expand $\omega_{g,n}$ with respect to $\gamma$ as $\omega_{g,n}=\sum_k\gamma^k\omega_{g,n}^{(k)}$ -- this truncates at finite $k$. Then, our third main statement (\cref{thm:STR}) can be summarised as below:

\begin{theorem}
    There exists an explicit formula to solve super loop equations for $\omega_{g,n}^{(k)}$. It is a recursion not only in terms of $2g-2+n$, but also, for fixed $(g,n)$, it is a recursion with respect to $k$. We call such a formalism that recursively constructs $\omega_{g,n}$ the superconformal topological recursion.
\end{theorem}

Perhaps surprisingly, the recursion formula in \cref{thm:STR} exhibits a similar structure to blobbed topological recursion \cite{BS17} rather than the ordinary one. There is actually an intriguing reason why we encounter such a phenomenon. Algebraically, it originates from a generator $\Gamma_0$ (called a zero mode) of the Clifford (free fermion) algebra  satisfying $[\Gamma_0,\Gamma_0]=\hslash^2$ where $[\cdot,\cdot]$ refers to a super commutator. Geometrically, it is inherited from an odd element $d\theta_0$ in $H^1(\Sigma)$ satisfying $\langle d\theta_0,d\theta_0\rangle=1$, where $\langle \cdot,\cdot \rangle$ is the super Givental pairing. Such a self-conjugate element does not exist in the ordinary setting, and this causes us to take care of such holomorphic contributions upon solving super loop equations.

\subsection{Future directions}

Since topological recursion has relations to many different areas of mathematics and several of these also have super analogues, a general question is whether these super analogues are `the same' as ours, i.e. if the relations in the non-super setting also have analogues in the super setting. Let us give a few explicit cases of this as potential future directions.\footnote{There is a very extensive literature on different aspects on supersymmetry. We do not claim expertise most of these aspects, and as such may have missed relevant references.}

\begin{itemize}
	\item The moduli spaces of super Riemann surfaces have been studied in e.g. \cite{DW15,BR21,DO23a,DO23}, and their compactification by stable curves in \cite{FKP23}. Furthermore, there has been recent progress \cite{KSY22,KSY20,KSY23,SYZ25} on defining super Gromov--Witten theory in the symplectic framework. We would be interested to know if superconformal topological recursion captures invariants of this sort.
	\item The super Weil--Petersson volumes of moduli spaces of super Riemann surfaces were studied via 
	$N = 1$ Jackiw--Teitelboim gravity in \cite{SW20} (and \cite{Mertens:2022irh} for review on JT) and related to topological recursion in \cite{Nor23, AN24}. It is reasonable to speculate whether our superconformal topological recursion can recast this relation in a more natural form and that it can be extended to a larger class of integrals over these moduli spaces. In addition, the Weil-Petersson form on moduli space descends from the Weil-Petersson form on Teichmüller space. Towards this direction, the Super Teichmüller space introduced in \cite{Penner:2015xla, Aghaei:2015bqi, Aghaei:2020otq} could be relevant.
	\item Super analogues of integrable hierarchies were discussed in the 90's, see e.g. \cite{Mul91,ABBEM93,FS93}, but there has been very little progress since then. Related concepts are super analogues of matrix models \cite{CHMS16,CHJMS17,BO18,Osu19} which suggest a possibility of quantisation of super spectral curves. Towards this direction, the superopers on supercurves introduced in \cite{Zei15} may be relevant.
	\item We do not prove all expected properties of superconformal topological recursion in the present article as it requires more careful analysis and it is beyond our scope. Two important ones are super analogues of the \emph{dilaton equation} and the \emph{variational formula} \cite{EO07}. These will be important to define and study what would be called \emph{super free energy} $F_g$. We hope to return this point in the future.
\end{itemize}

\subsection{Outline}

In \cref{sec:super_riemann_surfaces}, we introduce super Riemann surfaces with Ramond divisor. Most of the material in this section is not new, but we expect it not to be known to most of our audience. We tailor our presentation with applications to super topological recursion in mind, which means that we also deal with certain less standard concepts, such as symmetric multidifferentials. This section does contain some original material, e.g. the super Givental pairing in \cref{SuperGiventalPairing} and the (twisted) super quadratic Casimirs in \cref{SuperQuadraticCasimir,TwistedSuperQuadraticCasimir} as mentioned above.

In \cref{sec:the_symmetric_bidifferential}, we investigate a particular object that is essential for (super) topological recursion: the fundamental bidifferential of the second type. On an ordinary Riemann surface, this object encodes all holomorphic one-forms and their periods, and also yields information on meromorphic forms. An analogous story holds in the super setting, but the extension is not straightforward. We will review such subtleties and study its properties.

\Cref{SuperAiry} recalls super Airy structures and generalises this theory. For this purpose, we will highlight differences from ordinary Airy structures to motivate the necessity of defining partial super Airy structures. We then show (\cref{thm:SASmain}) that partial super Airy structures of our interests can be obtained by twisted modules of the $\mathcal{N}=1$ super Virasoro algebra. We will also explore new phenomena that only appear in the super setting, e.g. relaxing constraints due to nilpotent parameters.

In \cref{sec:super loop equations}, we translate the algebraic theory of partial super Airy structures to the geometric counterpart, i.e. loop equations among multidifferentials on a super spectral curve. To be more concrete, we derive loop equations from the associated super Virasoro constraints in \cref{thm:SLE}. We find that compact super spectral curves yield extra global constraints which imply that the corresponding Airy ideal is no longer just generated by direct sums of the Airy ideals at their ramification loci. We overcome this issue by finding a set of extra algebraic constraints that correctly captures such a global geometric feature.

\Cref{sec:SuperTR} contains the actual formula of superconformal topological recursion, i.e. a recursive way to solve super loop equations. Due to the indeterminacy of the partial Airy structures we started with, we require some choice of holomorphic parts, unlike the ordinary topological recursion but rather analogous to blobbed topological recursion. We also find that both the naive and the elegant super analogue of the derivation of topological recursion do not give rise to a useful formula, so we introduce an \emph{ad hoc} formula instead to make it work with a secondary recursion.

\Cref{sec:Examples} contains two examples of super spectral curves, namely the super versions of the Airy and Weber curves. These examples are explained in some detail, and they are used to highlight interesting features of superconformal topological recursion. We present a piece of evidence that the variational formula also holds in the super setting, which by itself is an interesting observation in terms of deformation of super Riemann surfaces.

\subsubsection*{Notation}

\begin{itemize}
\item Since this paper is concerned with a super analogue of non-super structures, we will frequently use the word `ordinary' to refer to the non-super version. 
\item We use the standard notations for super algebra and super spaces. For an algebraic superobject $V$ (a super vector space, superring, {\&}c.), we write $ V^0$ for its even and $ V^1$ for its odd part. In order to avoid confusion with other meanings of even and odd (e.g. even and odd integers), we occasionally use another common terminology: \emph{bosonic} for even and \emph{fermionic} for odd.
\item Following \cite{DM99}, we will consider the super grading and the cohomological grading to be independent. That is, forms $\omega,\tau $ have a \emph{parity}, a $ \Z_2$-valued super grading $ | \omega |$, as well as a $\Z$-valued cohomological grading $ \deg \omega$, and if we take their wedge product,
\begin{equation}
	\omega \wedge \tau = (-1)^{\deg \omega \deg \tau + |\omega | |\tau |} \tau \wedge \omega \,.
\end{equation}
For instance, $ dz \cdot \theta = \theta \cdot dz$ while $ d\theta  \cdot \theta = - \theta \cdot d\theta $  for even $ z$ and odd $ \theta$, and $d\theta\wedge d\theta$ is symmetric and non-zero.

\item Derivatives act from the left, so $ \del (ab) = (\del a) b + (-1)^{|\del | |a|} a(\del b)$. The contraction between a form $da$ and a vector field $\partial$ is given via $ da (\del) = (-1)^{|\del| |a|} (\del a)$. This is different from \cite{Wit19}.
\item When considering forms on powers of a space, as usual in the topological recursion formalism, we consider the cohomological grading in each copy of the space to be independent, as well as independent of the single super grading. That is, forms on a space $ \Sigma^n$ have a grading in $ \Z^n \times \Z /2\Z$. See \cref{sec:multi differential} for the meaning of symmetric forms with these gradings. 
\item When we consider a Lie bracket $[\cdot,\cdot]$, it is always a super Lie bracket. That is, it is either commutator or anti-commutator, depending on the parities of arguments.
\item We will often denote superconformal coordinates, e.g. $ (x \, | \, \phi )$, by blackboard bold letters such as $ \su{x}$. The superconformal analogue of non-super objects may receive the same treatment, e.g. the superconformal fundamental bidifferential will be denoted $ \su{B} $ in analogy with the non-super $B$.
\item For a set $I=\{i_1,...,i_n\}$ and objects $X_{i_1},...,X_{i_n}$, we write their tuple as $X_I\coloneqq\{X_{i_1},...,X_{i_n}\}$. We also write $ [n] \coloneqq \{ 1, \dotsc, n\}$, and hence, $ \su{z}_{[n]} = \{ \su{z}_1, \dotsc , \su{z}_n \}$. 
\end{itemize}

\subsubsection*{Acknowledgments}

We would like to thank Vincent Bouchard, Enno Ke{\ss}ler, Katherine Maxwell, Paul Norbury, Nadia Ott, and Artan Sheshmani for interesting and relevant discussions.

We are grateful to the organisers of the conference `TRSalento2021', where this project started. The authors also thank Syddansk Universitetet, the University of Alberta, the Universit\'{e} de Gen\'{e}ve, and the University of Tokyo for providing great working conditions during the collaborations that led to this work.

R.K. acknowledges support from the National Science and Engineering Research Council of Canada and the Pacific Institute for the Mathematical Sciences. The research and findings may not reflect those of these institutions. The University of Alberta respectfully acknowledges that we are situated on Treaty 6 territory, traditional lands of First Nations and M\'{e}tis people.

R.K. is partially supported by funds of the Istituto Nazionale di Fisica Nucleare, by IS-CSN4 Mathematical Methods of Nonlinear Physics. R.K. is also thankful to GNFM (Gruppo Nazionale di Fisica Matematica) of INdAM for supporting activities that contributed to the research reported in this paper.

K. O. acknowledges the support by JSPS KAKENHI Grant-in-Aid for JSPS Fellows (22KJ0715) and for Early-Career Scientists (23K12968), and in part also by 24K00525. K.O. also acknowledges the support from the Kobayashi--Maskawa Institute (KMI) for the Origin of Particles and the Universe at Nagoya University.  

N.A. acknowledges the funding and support during the time of this project; the European Union’s Horizon 2020 research and innovation programme under the Marie Sk{\l}odowska-Curie grant agreement No. 89875 and ReNewQuantum ERC-2018-SyG No. 810573 hosted by Centre for Quantum Mathematics, Department of Mathematics and Computer Science (IMADA), Southern Denmark University(SDU).  N.A. acknowledges NCCR SwissMAP, funded by the Swiss National Science Foundation at the University of Geneva and also support of the SNF Grant No. 20002-192080 at University of Zürich.

\section{Super Riemann surfaces}
\label{sec:super_riemann_surfaces}

In this section, we review and introduce geometric objects in the super setting. Our primary interests are in differentials and multidifferentials that play fundamental roles in superconformal topological recursion. For a good introduction into super Riemann surfaces, which we used as a basis for this section, see \cite{Wit19}. 

\subsection{Super Riemann surfaces}

We work in the framework of supergeometry, i.e. locally ringed spaces whose local models are open subspaces of 
\begin{equation}
	\C^{p| q} \coloneqq ( \C^p, \bigwedge \mc{O}^q_{\C^p} )\,,
\end{equation}
with transition maps  preserving \emph{parity}, i.e. the $\Z_2$ grading induced by putting $ \mc{O}_{\C^p}^q$ in degree $1$. We may take either the analytic or the Zariski topology. We will not indicate which choice we make unless it affects the results given. For a thorough introduction into superalgebra and supergeometry, see \cite{Ber87}. For a recent introduction into supergeometry in the algebro-geometric setting, see \cite{BHP23}. For a modern introduction in the differential-geometric setting focused more on super Riemann surfaces, see \cite{Kes19}.

\begin{definition}
	A \emph{super Riemann surface} $ \Sigma$ over a base complex supermanifold $T$ is a  $(1|1)$-dimensional complex supermanifold over $T$, together with the additional structure of a completely non-integrable distribution $ \mc{D} \subset \mc{T}_{\Sigma/T}$, i.e. there is a short exact sequence

\begin{equation}
	0 \to \mc{D} \to \mc{T}_{\Sigma/T} \to \mc{D}^2 \to 0\,.
\end{equation}
This means that for any non-vanishing local section $ \del $ of $ \mc{D}$, $ [\del, \del ]$ is linearly independent of $ \del$ at any point. The distribution $\mc{D}$ is called the \emph{superconformal structure} of $\Sigma$
\end{definition}

Every super Riemann surface $\Sigma$ admits an underlying reduced Riemann surface, which we often denote by $\Sigma_{\textup{red}}$, and which is naturally a subspace of $ \Sigma$ (although of course not a sub-super Riemann surface). Note that not every $(1|1)$-dimensional supermanifold is a super Riemann surface, unlike the ordinary case. This makes the study of super Riemann surfaces and their moduli interesting and at the same time complicated.

To understand super Riemann surfaces, or super geometry in general, it is sometimes useful to first start from the reduced underlying space, and a natural question is how to extend. For this, we need the following definition.

\begin{definition}
    A super space $ S $ over another super base $ T$ is \emph{projected} if there is a projection map $ S \to S_{\textup{red}}$ which is left inverse to the inclusion $ S_{\textup{red}} \to S$.\par
    Such a super space $ S/ T$ is \emph{split} if it is projected, and $ S $ is isomorphic to the total space of a parity-reversed vector bundle over $ S_{\textup{red}}$.
\end{definition}

In the smooth category, every super manifold is equivalent to a split one \cite{Bat79}. However, this is not true in the algebraic category if the base $T$ is not equal to its reduced subspace, and this is in particular also the case for super Riemann surfaces: moduli spaces of super Riemann surfaces are not projected for $ g \geq 5 $ \cite{DW15} (and not for $ g \geq 5r + 1 \geq 6$ in presence of Ramond punctures, introduced below \cite{DO23a}). Therefore, it is important to explicitly develop the relative theory (i.e. theory over a base $T$), because certain relevant super phenomena only occur in this setting.  In practice, working over a base $T$ amounts to considering families of super Riemann surfaces. We shall always consider super Riemann surfaces over a base $T$. In order to shorten the notations, however, this will be implicit in the following of this article.

On such a super Riemann surface, there exist specific coordinates called superconformal coordinates. 
\begin{definition}
	A pair of coordinates $ \su{z} = (z | \theta) $ is said to be \emph{superconformal} if 
	\begin{equation}
		D_\theta = \frac{\del}{\del \theta} + \theta \frac{\del}{\del z}
	\end{equation}
	is a local generator of $ \mc{D}$.
\end{definition}
One can easily check that in superconformal coordinates $[D_\theta,D_\theta]=\partial_z$, which is linearly independent of $D_\theta$.

Dually, we may consider the short exact sequence

\begin{equation}\label{SESCotangent}
	0 \to \mc{D}^{-2} \to \Omega^1_{\Sigma} \to \mc{D}^{-1} \to 0\,.
\end{equation}
In superconformal coordinates, the subbundle of $ \Omega^1_\Sigma $ isomorphic to $\mc{D}^{-2}$ is generated by $ \varpi_\theta = dz - d \theta \cdot \theta $ where one can  check, with care on sign conventions, that $\varpi_\theta(D_\theta)=0$. Using this notation, we see that the exterior derivative on $\Sigma$ can be written as

\begin{equation}\label{Extd}
	d = dz \frac{\del}{\del z} + d \theta \frac{\del}{\del \theta} = \varpi_\theta \frac{\del}{\del z} + d \theta D_\theta = \varpi_\theta D_\theta^2 + d \theta D_\theta \,.
\end{equation}

\begin{example}\label{P1|1}
	The simplest compact super Riemann surface is $ \P^{1|1} $, which is the quotient of $ \C^{2|1} \setminus \{ (0,0 |\epsilon)\, , \; \epsilon \in \mathbb{C}^{0|1} \}$ under the $ \C^*$-action $ c \cdot ( u, v | \epsilon ) = ( c u, c v | c \epsilon )$. Outside the locus $ v = 0$, the coordinates $ ( z = \frac{u}{v}| \theta=\frac{\epsilon}{v} )$ are well-defined. Furthermore, there is a natural anti-symmetric bilinear form on $\C^{2|1}$ (unique up to $ \mathop{\textup{PGL}} (2|1) $-action):
	\begin{equation}
		\< (u_1, v_1| \epsilon_1 ), (u_2, v_2 | \epsilon_2) \> = u_1 v_2 - u_2 v_1 - \epsilon_1 \epsilon_2\,.
	\end{equation}
This bilinear form induces a one-form $\varpi_\theta$ by specialisation $(u_1,v_1|\epsilon_1)=(u_2,v_2|\epsilon_2)=(z,1|\theta)$:
\begin{equation}
		\varpi_\theta\coloneqq d_1\< (u_1, v_1| \epsilon_1 ), (u_2, v_2 | \epsilon_2) \>|= dz- d\theta\cdot\theta\,.
	\end{equation}
This defines a (dual) superconformal structure $\mc{D}^{-2}$ on $\mb{P}^{1|1}$ making the coordinates $\su{z}=(z|\theta)$ superconformal.
\end{example} 

It turns out that a superconformal structure is a severe constraint. E.g. a ramified cover of a super Riemann surface cannot be given a compatible superconformal structure in general. Thus, we need a more general setup, allowing for the superconformality to break down in a controlled manner. A reasonable idea is to introduce a(n even) divisor\footnote{An even divisor is of codimension $1|0$.} $\mc{R}\subset\Sigma$, called the Ramond divisor, along which the superconformal structure degenerates.

\begin{definition}
	A \emph{super Riemann surface $\Sigma$ with a Ramond divisor $\mc{R}$} is a $(1|1)$-dimensional complex supermanifold together with a subbundle $ \mc{D} \subset \mc{T}_{\Sigma}$ such that there is a short exact sequence
	\begin{equation}
		0 \to \mc{D} \to \mc{T}_{\Sigma} \to \mc{D}^2(\mc{R}) \to 0\,.
	\end{equation}
	Generally, $ \mc{R}$ is assumed to be reduced, i.e. all components have multiplicity one. Its components are called \emph{Ramond punctures} (R). Points outside the support of $ \mc{R}$ may be called \emph{Neveu--Schwarz points} (NS).
\end{definition}

The idea behind this definition is that if $ \del $ is a local generator of $ \mc{D}$, then $ [\del, \del ] $ is proportional to $ \del$ along $\mc{R} $, which implies the breakdown of the complete non-integrability along $ \mc{R}$. As a consequence, $ \mc{T}_{\Sigma} / \mc{D} $ is actually generated by $ [\del, \del]/z $, where $ z $ is a local generator of $ \mc{O}(-\mc{R})$, the ideal sheaf of $\mc{R}$.

Note that Neveu--Schwarz points are of dimension $0|0$, whereas Ramond punctures are indeed of dimension $0|1$, despite the terminology `puncture'. From now, we will always implicitly allow our super Riemann surfaces to have a Ramond divisor, which we always denote by $\mc{R}$. 

\begin{lemma}[\cite{Wit19}]
	The restriction $ \mc{D}^2 (\mc{R})\big|_{\Sigma_{\textup{red}}}$ is isomorphic to $ \mc{T}_{\Sigma_{\textup{red}}}$. If $ \Sigma$ is proper, then $ \deg \mc{R}$ is even.
\end{lemma}

Throughout the paper, we assume that the degree of $\mc{R}$ is finite, and we denote this by either $\deg\mc{R}=r$ or $\deg\mc{R}=2r$, depending on the context -- the latter notation is useful for compact super Riemann surfaces.

The dual short exact sequence is
\begin{equation}
	0 \to \mc{D}^{-2}(-\mc{R}) \to \Omega^1_{\Sigma} \to \mc{D}^{-1} \to 0\,.
\end{equation}

\begin{definition}
	We will denote the subbundle of $ \Omega^1_{\Sigma/T}$ isomorphic to $ \mc{D}^{-2}(-\mc{R})$ by $ \mc{P}$, and define $ \mc{Q} \coloneqq  \Omega^1_{\Sigma} / \mc{P} $ for the quotient bundle isomorphic to $ \mc{D}^{-1}$. We write $ [\omega ] $ for the image of $ \omega \in \Omega_{\Sigma}^1$ in $ \mc{Q}$.
\end{definition}

At the Ramond divisor, the notion of superconformal coordinates does not make sense (as the superconformal structure degenerates there), but there is a notion of Ramond coordinates. 

\begin{definition}\label{def:Ramond coordinates}
	A pair of coordinates $\su{z} = ( z \, | \, \theta )$ is called \emph{Ramond coordinates} if for some Ramond puncture $ p$,  $z(p)=0$ and
	\begin{equation}
		D^*_\theta 
        =
        \frac{\del}{\del \theta} +  \theta z \frac{\del}{\del z}
	\end{equation}
	is a local generator of $ \mc{D}$.
\end{definition}

We immediately see $ (D^*_\theta)^2 = z \frac{\del}{\del z}$. The dual generator of $ \mc{P}$ is $ \varpi_\theta^* = dz -  d\theta \cdot \theta z$, while the induced generator of $ \mc{P}(\mc{R}) \cong \mc{D}^{-2}$  is $\frac{\varpi^*_\theta}{z}$. In particular, we have
\begin{equation}\label{RExtd}
    d
    =
    \frac{\varpi^*_\theta}{z}(D_\theta^*)^2+d\theta D_\theta^*.
\end{equation}

\begin{example}\label{exa:P11|0}
	Let us consider a super version of $ \P^1 $ with two Ramond punctures: the weighted projective space $ \P (1,1| 0)$. It is the quotient of $ \C^{2|1} \setminus \{ (0,0 \, | \, \theta )\}$ under the $ \C^*$-action $ c \cdot (u, v \, | \, \theta ) = (cu, cv \, | \, \theta )$. A natural skew-symmetric bilinear form in this case is:
\begin{equation}
		\< (u_1, v_1 \, | \, \theta_1 ), (u_2, v_2 \, | \, \theta_2) \> = u_1 v_2 - u_2 v_1 - \frac{1}{2} ( u_1 v_2 + u_2 v_1) \theta_1 \theta_2
	\end{equation}
Outside the locus $v=0$, consider the specialisation $(u_1,v_1|\theta_1)=(u_2,v_2|\theta_2)=(z,1|\theta)$ on the following one-form
	\begin{equation}
    	\varpi^*_\theta= d_1\< (u_1, v_1 \, | \, \theta_1), (u, v | \theta) \>|=dz -  d\theta\cdot  \theta z\,.
	\end{equation}
This verifies that $ \P (1,1| 0)$ is a super Riemann surface with $\{(0,1\, | \, \theta)\}$ as a Ramond puncture and $\su{z}=(z|\theta)$ are Ramond coordinates. Another Ramond puncture is located at  $\{(1,0 \, | \, \theta)\}$.
\end{example}

\subsection{Superconformal maps}

In the non-super setting, a conformal map is a holomorphic, or complex-analytic map, i.e. a map which locally around any point $ p$ in the source coincides with its Taylor expansion (in complex variables). For any local coordinate $z$ near $p$, and any coordinate near $ f(p)$, one can write 
\begin{equation}
	f(z) = f(p) + \sum_{j \geq k} c_j z^j 
\end{equation}
with $ c_k \neq 0$, for some $k \geq 0$ independent of the chosen coordinate, but dependent on $f$ and $p$.\par
A superconformal map should be a (superholomorphic/superanalytic) map between the $(1|1)$-dimensional supermanifolds underlying super Riemann surfaces, which moreover preserves the superconformal structure in some way. The right definition is the following.

\begin{definition}\label{def:SCtrans}
	Let $\Sigma$, $\mc{C}$ be two super Riemann surfaces. A map $f\colon\Sigma\to\mc{C}$ is said to be \emph{superconformal} if it is a holomorphic map with an even ramification divisor $R = \sum_i e_i R_i$\footnote{A ramification divisors $R$ and a Ramond divisor $\mc{R}$ are different in general, and we denote them by different letters.} such that $ f^* \mc{D}_C \cong \mc{D}_\Sigma \big(\ceil{\frac{R}{2}} \big) $ where $ \ceil{\frac{R}{2}}\coloneqq \sum_i \ceil{ \frac{e_i}{2} } R_i$. We call $e_i$ the \emph{ramification index of $f$ at $R_i$} and $r_i=e_i+1$ the \emph{ramification order of $f$ at $R_i$}, and we impose that ramification indices of $f$ are equal to those of its bosonic reduction.
\end{definition}
 
\begin{remark}
	The definition of superconformal transformation given here seems to not have appeared in the literature before. It is essentially the same as the description of branched covers of super Riemann surfaces by \cite{DW15,DO23a}. The description there is given in two parts: first, given a super Riemann surface together with a branched cover of its reduced space and a divisor whose reduction is the branch divisor, a branched cover of $(1|1)$-dimensional supermanifolds is constructed by pulling back the sheaf of sections. However, the resulting cover is not a super Riemann surface. In order to fix this, the cover must be blown up along the divisor. Because of this blowup, the inclusion $ f^* \mc{D}_C \supseteq \mc{D}_\Sigma$ should be twisted along the ramification divisor.
\end{remark}

Let us generalise a local description of ramified maps into the super setting.

\begin{proposition}\label{prop:SCtrans}
Let $ f \colon \Sigma \to C$ be a superconformal map. In local superconformal or Ramond coordinates $ \su{z} = (z\, | \, \theta)$ and $ \su{x} = (x \, | \, \phi)$ on $ \Sigma $ and $ C$, respectively, we may write
	\begin{equation}
		\begin{split}
			x \circ f(\su{z}) & = a(z) + \theta \alpha (z)\,,
			\\
			\phi \circ f (\su{z})&= \beta (z) + \theta b(z)\,,
		\end{split}
	\end{equation}
	where $ a,b$ are bosonic and $ \alpha, \beta$ are fermionic holomorphic functions of $ z$. Then, depending on the types at $z=0$ and $x=0$ of ramification order $r$, the four holomorphic functions $a,b,\alpha,\beta$ are related as below:
	\begin{align}
		&\textup{NS to NS} & \alpha &=b\beta\,, & a' &= b^2+ \beta' \beta\,, & r=&2k+1, &\beta(z)  &=\beta_0+\mc{O}(z^{k+1});&\label{NStoNS}
				\\
		&\textup{R to NS} & \alpha &= b\beta, & za' &=b^2+z\beta' \beta\,,  & r= &2k  &\beta(z)&=\beta_0+\mc{O}(z^{k}),;& \label{RtoNS}
		\\
		&\textup{NS to R} &  &  &  & \text{only constant maps;}   &  & & &   & \label{NStoR}
		\\
		&\textup{R to R} & \alpha &= ab\beta , & za' &=a(b^2 +z \beta' \beta) \,, & r =& 1, & &\hspace{-5mm} \text{any }  \beta(z),\;\;\;\; a(0)=0 &  \label{RtoR}
	\end{align}
\end{proposition}

\begin{proof}
	Constant maps are always allowed, with a condition that $x=0$ when the target is Ramond by \cref{def:Ramond coordinates}. Thus, we assume that $f$ is not constant. One of the requirements of superconformality is that $ D_\theta $ be proportional to the pullback of $ D_\phi$ (c.q. $D^*_\theta$, $D^*_\phi$). Calculating by the chain rule, we obtain the following conditions.
	\begin{align}
		&\textup{NS to NS} & \alpha &=b\beta\,, & a' &= b^2+ \beta' \beta\,, & D_\theta = (b+\theta \beta')D_\phi\,;\label{NStoNS1}
		\\
		&\textup{R to NS} & \alpha &=  b \beta \,, & za' &=b^2+z\beta' \beta\,, & D^*_\theta = (b+\theta z \beta')D_\phi\,; \label{RtoNS1}
		\\
		&\textup{NS to R} & \alpha &=  ab\beta\,, & a'&= a(b^2+ \beta' \beta\,)\,, & D_\theta=(b+\theta\beta')D_\phi^*;\label{NStoR1}
		\\
		&\textup{R to R} & \alpha &=  ab\beta \,, & za' &=a(b^2 +z \beta' \beta) \,, & D^*_\theta = (b+\theta z \beta')D^*_\phi\, , \label{RtoR1}
	\end{align}
	where the proportionality between the superconformal structures can be also written as $D_\theta=(D_\theta\phi)D_\phi$ (c.q. $D^*_\theta$, $D^*_\phi$) and we dropped arguments of $a,b,\alpha,\beta$ and the pullback notation $f^*$ for brevity. Note that the above constrains need to be modified if the Ramond divisors in $\su{z}$- or $\su{x}$-coordinates are not along $z=0$ or $x=0$. Then, the twist condition implies that $b(z)=\mc{O}(z^{\floor{r}})$ and $\beta'(z)= \mc{O}(z^{\floor{\frac{r}{2}}-s})$ where $s=0$ when $z=0$ is NS and $s=1$ when $z=0$ is R, using that $\ceil{\frac{e}{2}} = \floor{\frac{r}{2}}$. The constant term $\beta(0)$ is not fixed by the twist condition.

	For the NS-NS case, the ramification order must be odd, because $b=\sqrt{a'-\beta'\beta}$ should make sense. For the R-NS case, the order of the argument of the square root is shifted by one, and hence the ramification order must be even.

	When the target is Ramond, $a(0)=0$ by the definition of Ramond coordinates, so let us denote $a(z)=u z^k+\mc{O}(z^{k+1})$ for $k\geq1$ where $u$ may be nilpotent. Then for the NS-R scenario,  $a'(z) = ku z^{k-1} + \mc{O}(z^{k})$ while $a(z)(b(z)^2 + \beta'(z) \beta(z)) = \mc{O}(z^{k})$, which is a contradiction. Therefore, only constant maps are allowed in this case. Similarly for R-R maps, one finds a contradiction, except when $r=1$, because of the required vanishing order of $ b$. Thus, every R-R superconformal map is unramified.
\end{proof}

\begin{remark}\label{rem:NS shift}
	When the target point is NS, only $a'$ appears in the constraints, hence any even constant can be freely added to $x$, similarly to the ordinary conformal setting. However, adding an odd constant is constrained. More concretely, if one considers a superconformal transformation such that the odd coordinate $\theta$ is shifted by an odd constant $\beta_0$, i.e., $\phi=\beta_0+\theta$, then the even coordinate $z$ is forced to become $x = z+\theta\beta_0$ by the case of $r=1$ NS-NS transformations of \cref{prop:SCtrans}. 
\end{remark}

We will only be interested in the case where the target has empty Ramond divisor, so we will focus on that and consider the local model $ \C^{1|1}$. In this situation, we can find local coordinates, as in the ordinary case, in which the superconformal map has a simple shape. The following two propositions (cf. \cite[Proposition~3.3]{DW15} and \cite[Section~3.1.1]{DO23a}) give this shape in the case of a Ramond and Neveu--Schwarz source point.

\begin{proposition}\label{prop:NS-NS SCmap}
	Let $ f \colon \Sigma \to C$ be a superconformal map, which at a given point $ p \in \Sigma $ is of the type NS to NS. Then for any local superconformal coordinates $ \su{x} = ( x \, | \, \phi )$ near $ f(p)$, we may find local superconformal coordinates $ \su{z}=( z \, | \, \theta )$ near $p$ such that $ f(z \, | \, \theta ) = ( z^{2k+1} \, | \, \sqrt{2k+1} \theta z^k)$, where $2k+1$ is the ramification index of $ f$.
\end{proposition}

\begin{proof} 
	Let us take some local superconformal coordinates $\tilde{\su{z}}= (\tilde{z} | \tilde{\theta} ) $ near $p$. We write $ f = \big( a(\tilde{z} ) + \tilde{\theta} \alpha( \tilde{z} ) \, \big| \, \beta ( \tilde{z} ) + \tilde{\theta} b(\tilde{z}) \big) $. 
	
	We want to find a change of NS coordinates on $\Sigma$ to $\su{z}= (z \, | \, \theta ) = \big( c (\tilde{z}) + \tilde{\theta} \gamma (\tilde{z}) \, \big| \, \delta (\tilde{z}) + \theta d (\tilde{z}) \big)$ (subject to \cref{NStoNS}) such that $ f = ( z^{2k+1} \, | \, \sqrt{2k+1} \theta z^k )$. Such a change of coordinates on $\Sigma$ would have to satisfy (as equations in $\tilde{\su{z}}$-coordinates)
	\begin{equation*}
		\begin{split}
			(a + \tilde{\theta} \alpha \, | \, \beta + \tilde{\theta} b)
			&=
			\big( (c+ \tilde{\theta} \gamma )^{2k+1} \, \big| \, \sqrt{2k+1} (\delta + \tilde{\theta} d) (c + \tilde{\theta} \gamma)^k \big)
			\\
			&=
			\big( c^{2k+1} + (2k+1) \tilde{\theta} \gamma c^{2k} \, \big| \, \sqrt{2k+1} (\delta + \tilde{\theta} d) (c^k + k \tilde{\theta} \gamma c^{k-1}) \big)
			\\
			&=
			\big( c^{2k+1} + (2k+1) \tilde{\theta} \gamma c^{2k} \, \big| \, \sqrt{2k+1} \delta c^k + \sqrt{2k+1} \tilde{\theta} ( d c^k - k \delta \gamma a c^{k-1}) \big)
		\end{split}
	\end{equation*}
	Requiring that $ \gamma = \sqrt{c'} \delta$ and $ d = \sqrt{ c' - \delta' \delta}$, we see that $ \delta \gamma = 0$, so we find
	\begin{equation*}
		\begin{split}
			(a + \tilde{\theta} \alpha \, | \, \beta + \tilde{\theta} b)
			&=
			\big( c^{2k+1} - (2k+1) \tilde{\theta} \sqrt{c'} \delta c^{2k} \, \big| \, \sqrt{2k+1} \delta c^k + \sqrt{2k+1} \tilde{\theta} \sqrt{ c' - \delta' \delta} c^k \big)
		\end{split}
	\end{equation*}
	This can be solved: we may find $ c$ such that $ c^{2k+1} = a = \mc{O}(z^{2k+1})$, hence $c= \mc{O}(z)$. Then, take $ \delta = \frac{\beta}{\sqrt{2k+1} c^k}$ which can be done because $ \beta = \mc{O}(z^{k+1}) = \mc{O}(c^{k+1})$; there is no constant term due to the assumption that $\su{x}$ are coordinates near $f(p)$. Moreover, as a consistency check, we do get
	\begin{align*}
		\alpha
		&=
		 \sqrt{a'} \beta 
		= 
		\sqrt{ (2k + 1) c^{2k} c'} \cdot \sqrt{2k+1} \delta c^k 
		= 
		 (2k + 1) \sqrt{c'} \delta c^{2k}\,,
		\\
		b
		&=
		\sqrt{ a' - \beta' \beta } = \sqrt{ (2k + 1) c^{2k} c' - \sqrt{2k + 1} \delta' c^k \sqrt{2k + 1} \delta c^k } = \sqrt{ (2k + 1) (c' -\delta' \delta )} c^k \,,
	\end{align*}
	which is again consistent with the previous calculation.
\end{proof}

\begin{proposition}\label{prop:NS-R SCmap}
	Let $ f \colon \Sigma \to C$ be a superconformal map, which at a given point $ p \in \Sigma $ is of the type R to NS. Then for any local superconformal coordinates $ \su{x} = ( x \, | \, \phi )$ near $ f(p)$, we may find local Ramond coordinates $ \su{z}=( z \, | \, \theta )$ near $p$ such that $ f(z \, | \, \theta ) = ( z^{2k} \, | \, \sqrt{2k} \theta z^k)$, where $ 2k$ is the ramification index of $ f$.
\end{proposition}

\begin{proof}
	A parallel proof to the previous proposition works.	
\end{proof}

In order to clarify potential confusions, it is perhaps helpful to present examples of non- super conformal maps that one may naively think of as superconformal.

\begin{example}
	For odd parameters $\tau_0,\tau_1$ and positive integers $k,s$ with $k<s$ and $k+3\neq s$, consider the following NS-NS map $f \colon \Sigma \to C$:
	\begin{equation}
		x
		=
		\frac{(k+3-s)\tau_0\tau_1}{k+s+1}z^{k+s+1}+\frac{z^{2s+1}}{2s+1}+\theta(\tau_0z^{k+s+1}+\tau_1z^{2s}), \quad 
		\phi
		=
		\tau_0 z^{k+1}+\tau_1 z^s+\theta(\tau_0\tau_1 z^k+z^s)
	\end{equation}
	By setting $\tau_0=\tau_1=\theta=0$, i.e. when the base $T$ is reduced, this map induces a conformal transformation of ramification order $2s+1$. Moreover, $D_\theta=((\tau_0\tau_1+(k+1)\theta\tau_0)z^k+\mc{O}(z^{k+1}))f^*D_\phi$, so that $f$ preserves the superconformal structure, but $f^* \mathcal{D}_{C}=\mathcal{D}_\Sigma(k)\not\cong \mathcal{D}_\Sigma(\floor{\frac{2s+1}{2}})$ for generic $\tau_0,\tau_1$ due to the assumption $k<s$, violating the condition that the ramification index of $f$ is equal to that of its reduction in \cref{def:SCtrans}. Thus, we do not call $f$ of this example a superconformal map. The reduction condition ensures that the leading order of $a(z)$ is invertible, i.e. not nilpotent.
	
\end{example}

\begin{example}
One may wonder what happens if we take any effective divisor instead of a reduced Ramond divisor. Equivalently, what if we consider a limit over $T$ such that several components of a Ramond divisor coincide. Locally, this would mean that $\mc{D}$ is now generated in local coordinates by
\begin{equation}
\frac{\del}{\del \theta} + \theta z^{k} \frac{\del}{\del z}
\end{equation}
for some $k \in \mb{Z}_{\geq2}$. Existence of such a subbundle is referred to as a \emph{parabolic structure}, which is a generalisation of superconformal structures (cf. \cite{DW15}). Thus, by construction, any map to/from such a surface would not be called superconformal. Nonetheless, one may consider a simple holomorphic coordinate transformation $(x|\phi)= (z|z^m\theta)$ for $m\in\mb{Z}_{\geq0}$ which gives
\begin{equation}
\frac{\del}{\del \theta} + \theta z^{k} \frac{\del}{\del z}=x^m\left(\frac{\del}{\del \phi} + \phi x^{k-2m} \frac{\del}{\del x}\right).
\end{equation}
This shows that considering a supermanifold  with a parabolic structure is closely to related to a higher ramified cover of a super Riemann surface.
\end{example}

It is worth mentioning that Ramond coordinates are often called superconformal in the literature. We will try to avoid such a use, as they behave differently. In particular, they are not superconformal anywhere in a punctured neighbourhood of the Ramond divisor. The relation between the Ramond coordinates $(z^{\textup{R}}\, | \, \theta^{\textup{R}})$ and some superconformal coordinates $(z^{\textup{SC}} \, | \, \theta^{\textup{SC}})$ on a simply connected open subset $U$ of that punctured neighbourhood is
\begin{equation}
	\begin{split}
		z^{\textup{R}} & = z^{\textup{SC}}\,;
		\\
		\sqrt{z} \theta^{\textup{R}} &= \theta^{\textup{SC}}\,,
	\end{split}
\end{equation}
as we can see from the fact that $ D^*_{\theta^{\textup{R}}} = \sqrt{z^{\textup{R}}} D_{\theta^\textup{SC}}$. The square root may be chosen on all of $U$ as it is simply connected and $z$ is non-zero on it. The transformation does not meet the condition of \cref{def:SCtrans} at $z=0$\footnote{Another possible approach is taken in \cite{DO24}: in that paper, superconformal coordinates are always taken to be such that $ D_\theta = \frac{\del}{\del \theta} + z \theta \frac{\del}{\del z}$. If the coordinate chart does not intersect the Ramond divisor, this can still be achieved by choosing $z$ non-zero on the entire chart. We will not use this convention, as we want to be able to speak of coordinates centred at any point, be it Ramond or Neveu-Schwarz.}.

\begin{remark}\label{rem:R at infinity}
In \cref{exa:P11|0}, there is another Ramond puncture at which the superconformality breaks down, namely, at $z = \infty $. To see this explicitly, let us consider a coordinate transformation from $\su{z}_0$ to $\su{z}_\infty$ 
\begin{equation}
z_\infty = \frac{1}{z_0} \,,\quad
    	\theta_\infty = i \theta_0
\end{equation}
where we have chosen the sign of $b$ in \cref{RtoR} by $\theta_\infty = i \theta_0 $, not $\theta_\infty = -i \theta_0 $. One can show that this is an R-R superconformal transformation. In particular, the generators of $\mc{D}^{-2}$ in coordinates $\su{z}_0$ and $\su{z}_\infty$ are related by
	\begin{equation}\label{z-t SCR}
        	\frac{\varpi^*_{\theta_0}}{z_0} = - 	\frac{\varpi^*_{\theta_\infty}}{z_\infty}  \,,
	\end{equation}
\end{remark}

Let us close this subsection by introducing the notion of orientation of a Ramond divisor. First, one can obtain the following lemma from \cref{prop:SCtrans}; it also appears in \cite{Wit15}.

\begin{lemma}\label{RamondAut}
	If $ \phi $ is a superconformal automorphism near a  Ramond puncture and $ (z \, | \, \theta ) $ are Ramond coordinates on it, then
	\begin{equation}
	\phi^* \theta = \pm (\theta + \tau ) \pmod{z} \,,
	\end{equation}
	where $ \tau $ is an odd function on the base $ T$.
\end{lemma}
This behaviour is fundamentally different from Neveu--Schwartz points: if $(z \, | \, \theta )$ are superconformal coordinates near an NS point, it is possible to rescale the $ \theta$, provided one also rescales $z$. For Ramond coordinates, this is not possible.

It becomes particularly important to choose the sign at every Ramond puncture in order to define so-called odd periods in \cref{sec:odd periods}.

\begin{definition}\label{DefOrientation}
	An \emph{orientation} of a Ramond divisor $\mathcal{R}$ is a  reduction of the ambiguity in the previous definition to
	\begin{equation}
	\phi^* \theta = \theta + \tau \pmod{z} \,,
	\end{equation}
	at every component of $\mathcal{R}$.
\end{definition}

An oriented Ramond divisor $\mathcal{R}$ has a well-defined form $ d\theta$ associated to every component -- this is the generator of $ \Omega^1_{\mathcal{R}/S}$ and as such is really an orientation.

\subsection{Differential forms}

Let us look into different kinds of differential forms. A good reference for this material is \cite{DO24}.

\subsubsection{Berezinian}

The first important geometric structure on $\Sigma$ is the \emph{Berezinian bundle} $ \omega_{\Sigma/T} = \Ber (\Sigma /T) = \Ber ( \Omega^1_{\Sigma/T})$, which is the super analogue of the canonical bundle -- the Berezinian being the super analogue of the determinant. For readers unfamiliar with this notions, let us recall that for a linear automorphism $ L$ of a super vector space $ V$, we can write $L$ as a block matrix with respect to the grading of $ V$ as
\begin{equation}
    L = 
    \begin{pmatrix}
        A & \Gamma
        \\
        \Delta & B
    \end{pmatrix}
\end{equation}
and
\begin{equation}
    \Ber L = \det A \det (B - \Delta A^{-1} \Gamma)^{-1} \,.
\end{equation}
Given local superconformal coordinates $ \su{z}= (z\, |\, \theta ) $, the Berezinian has a local section $ [ dz | d\theta ] $; given other superconformal coordinates $ \su{w}= (w\, |\, \psi )$, the transition function is
\begin{equation}
	[ dw | d\psi ] = [dz | d\theta ] \Ber \begin{pmatrix} \del_z w & \del_z \psi \\ \del_\theta w & \del_\theta \psi \end{pmatrix} \,=[dz|d\theta](D_\theta\psi).
\end{equation}
The second equality explicitly shows that, for a SRS, there is an isomorphism $ \Ber (\Sigma) \cong \mc{D}^{-1}$, which is given in local coordinates by $ [dz | d\theta ] \mapsto [d\theta] $.

In the presence of a Ramond divisor, however, the Berezinian becomes $ \omega_{\Sigma/T} \cong \mc{D}^{-1} (- \mc{R}) \cong \mc{Q}(-\mc{R})$, with the isomorphism given locally by $ [ dz | d\theta ] \mapsto [d \theta ] z $.  One can check this by explicitly computing the Berezinian for R-R superconformal transformation, or see \cite{Wit15} for more details.

\subsubsection{Closed forms}
One thing to note regarding differential forms on super Riemann surfaces is that there exist non-trivial forms of arbitrary order, unlike the ordinary case. Namely, if $ \theta$ is a local odd coordinate, then $ (d\theta)^k $ is non-vanishing for any $ k \geq 0$. In particular, this means that one-forms are not necessarily closed. For our application, we are interested in closed holomorphic/meromorphic one-forms, as we want to calculate periods and residues.

There are different ways to think about closed one-forms on super Riemann surfaces. One is the extended Berezinian; this is the point of view taken in \cite{DO24}. The isomorphism $ \omega_{\Sigma} \cong \mc{D}^{-1}(-\mc{R}) \cong \mc{Q}(-\mc{R})$ induces a short exact sequence

\begin{equation}
	0 \longrightarrow \omega_{\Sigma} \longhookrightarrow \omega_{\Sigma} (\mc{R}) \stackrel{\big|_\mc{R}}{\longrightarrow} j_*\Omega^1_{\mc{R}} \longrightarrow 0 \,,
\end{equation}
where the first map is an inclusion while the last map is restriction (and $ j $ is the closed immersion of $ \mc{R}$ in $ \Sigma$): it is the composition given in local Ramond coordinates as
\begin{equation}
 	\omega_{\Sigma}(\mc{R} ) \stackrel{\sim}{\longrightarrow} \mc{Q} \to j_*\Omega_{\mc{R}}^1 \colon \frac{[dz|d\theta]}{z} \mapsto [d\theta] \mapsto d\theta\big|_\mc{R} \,.
 \end{equation} 

\begin{definition}
	The \emph{extended Berezinian sheaf} of $ \Sigma$ is
	\begin{equation}
		\omega_{\Sigma}' 
        \coloneqq
        \Ker \big( d \circ \big|_\mc{R} \colon \omega_{\Sigma}(\mc{R}) \to j_* \Omega^1_{\mc{R}} \big) \,.
	\end{equation}
\end{definition}

Concretely, sections of $ \omega_{\Sigma}' $ may have simple poles along $ \mc{R}$, but these poles must be constant along $ \mc{R}$.
That is, they must be of the form $ a \frac{[dz|d\theta]}{z}$ in local Ramond coordinates with $ a $ a local holomorphic function on $\Sigma$ with no term $ z^0 \theta$.

\begin{proposition}[{\cite{RSV88,Wit15}}]\label{ClosedForms}
	The extended Berezinian sheaf is isomorphic to the sheaf of closed holomorphic one-forms:
	\begin{equation}
		\alpha \colon \omega_{\Sigma}' \stackrel{\sim}{\longrightarrow} \mc{Z}^1_{\Sigma} \,.
	\end{equation}
\end{proposition}

It is useful to describe this isomorphism explicitly. Consider a $1$-form in local coordinates $\su{z}=(z \, | \, \theta)$
\begin{equation}
	\sigma = dz \, f(z | \theta ) + d\theta \, g( z | \theta )\,.
\end{equation}
Since $\deg d=1$ and $|d|=0$, we have
\begin{equation}
	- d\sigma = dz \wedge d\theta \frac{\del f}{\del \theta} + d\theta \wedge dz \frac{\del g}{\del z} + d\theta^2 \frac{\del g}{\del \theta} \,.
\end{equation}
This vanishes exactly if $ \frac{\del g}{\del \theta } = 0$ and $ \frac{\del f}{\del \theta} = \frac{\del g}{\del z}$, so $ g $ is independent of $ \theta$, and the $ \theta$-dependent part of $ f$ is determined by $ g$. Writing $ \bar{f} = (1- \theta \del_\theta ) f$ for the $\theta$-independent part of $ f$, we can write
\begin{equation}
	\sigma 
	=
	dz \big( \bar{f} + \theta \frac{\del g}{\del z} \big) + d\theta \, g
	=
	\begin{cases}
		d\theta ( \theta \bar{f} + g) + \varpi D_\theta (\theta \bar{f} + g)
		&
		\text{NS}
		\\
		d\theta (\theta z \bar{f} +g) + \frac{\varpi^*}{z} D_\theta^*(\theta z \bar{f} +g)
		&
		\text{R} 
	\end{cases}\,.
\end{equation}
Keeping in mind that under the isomorphism $ \omega_{\Sigma} \cong \mc{Q} (-\mc{R})$ we have in Ramond coordinates that $ [dz|d\theta] \mapsto [d\theta] z$, we find
\begin{equation}
	\alpha^{-1} (\sigma)
	=
	\begin{cases}
		[dz |d\theta ] (\theta \bar{f} + g) & \text{NS}
		\\
		[dz|d\theta] ( \theta \bar{f} + \frac{g}{z}) & \text{R}
	\end{cases} \,,
\end{equation}
and the inverse is given by
\begin{equation}
	\begin{cases}
		\alpha ( [dz|d\theta] h(\su{z} )) = d\theta \, h + \varpi D_\theta h & \text{NS}
		\\
		\alpha ([dz|d\theta ]h ( \su{z} )) = d\theta z h + \frac{\varpi^*_\theta}{z} D_\theta^* (z h) & \text{R}
	\end{cases} \,.
\end{equation}

Comparing with \cref{Extd,RExtd}, we find that under $ \alpha$, exact forms are strongly related to the image of the superconformal operator:

\begin{equation}
	\alpha^{-1} (df)
	=
	\begin{cases}
		[dz|d\theta] D_\theta f & \text{NS}
		\\
		\frac{[dz|d\theta]}{z} D_\theta^* f & \text{R}
	\end{cases} \,.
\end{equation}

We can also interpret the extended Berezinian sheaf $ \omega_{\Sigma}'$ as a subsheaf of $ \mc{Q}$, of sections of the shape $ [dz]f(z) + [d\theta] g(z)$, with $ f,g$ holomorphic -- note that $ [d\theta ] \theta  = [dz] $ in superconformal coordinates while $ [d\theta ] \theta = \frac{[dz]}{z} $ in Ramond coordinates. So in local coordinates $\su{z}= (z \, | \, \theta )$ in a formal neighbourhood, superconformal or Ramond, the space of meromorphic sections of $ \omega_{\Sigma}'$ has a $\C$-basis $ \{ [d\theta] z^k\}_{k \in \Z} $ and $ \{ [dz] z^{k-1} \}_{k \in \Z^*} $.

\subsubsection{Residues}

Let $ p \in \Sigma_{\textup{red}}$. The residue of a (closed) one-form defined on a punctured neighbourhood $ U$ of $ p$ in $ \Sigma_{\textup{red}}$ is defined by integration along the homology class $ \gamma \in H_1(U;\Z )$  represented by counterclockwise simple closed curves around $p$.\par
On the super Riemann surface $ \Sigma$, there is a canonical lift of the point $p$: given a local coordinate $ z$ at $p \in \Sigma_{\textup{red}}$, this can be completed to superconformal/Ramond coordinates $\su{z}= (z \, | \, \theta)$, and $ P = \{ z = 0\}$ defines an odd divisor lifting $ p$ independent of the choice. Moreover, $ \gamma $ lifts to a unique $ \Gamma \in H_1 ( \Sigma \setminus P ; \Z)$ circling around $ P$, and hence the following definition makes sense.

\begin{definition}\label{def:Residue}
	Let $ p \in \Sigma_\textup{red}$, let $ P$ be the divisor described above, with a neighbourhood $ \mc{U} \supset P$ and let $ q \in \mc{Z}^1 (\mc{U} \setminus P) $. Then, with more notation as above the \emph{even residue} (or simply \emph{residue}) of $ q$ at $p$ is defined as
	\begin{equation}
		\Res_{z=p}q = \frac{1}{2\pi i} \int_\Gamma q \,.
	\end{equation}
\end{definition}

\begin{lemma}[\cite{DO24}]
	Let $ p \in \Sigma_\textup{red}$ with local coordinate $ z$, and extend it to superconformal c.q. Ramond coordinates  $\su{z}= (z|\theta) $ on $ \Sigma$.  Let
	\begin{equation}
		q = \sum_{k = -N}^\infty ([dz] z^k u_k +  [d\theta] z^k \upsilon_k)
	\end{equation}
	be a meromorphic section of $ \mc{Q}$ near $ p$. Then, under the identifications $ \mc{Z}_{\Sigma}^1 \cong \omega_{\Sigma}' \subset \mc{Q}$, its residue is
	\begin{equation}
		\Res_{z = p} q = u_{-1}\,.
	\end{equation}
\end{lemma}

There is a subtlety about residues that appears only in the super setting. If $ P$ happens to be a Ramond puncture, then for a holomorphic
\begin{equation}
	q = \sum_{k = 0}^\infty [d\theta] ( \theta z^k \tilde{u}_k + z^k \tilde{\upsilon}_k)\,,
\end{equation}
its residue is
\begin{equation}
	\Res_{z=p} q = u_{-1} = \tilde{u}_0 \,,
\end{equation}
because $ [d\theta] \theta = \frac{[dz]}{z}$. This may be non-zero unlike for holomorphic one-forms on an ordinary Riemann surface. However, notice that sections of the extended Berezinian sheaf have no residues. In summary, we have:

\begin{corollary}\label{ExtendedBerezinianIsRestrictedQ}
	The extended Berezinian sheaf $\omega_{\Sigma/T}'$ is isomorphic to the subsheaf of $ \mc{Q}$ of residueless sections. That is, the residues at $\mc{R}$ piece together to a map
	\begin{equation}
		\mathop{\textup{RES}} \colon \mc{Q} \to j_* \mc{O}_\mc{R} \,,
	\end{equation}
	and the isomorphism $ \omega_{\Sigma}(\mc{R} ) \cong \mc{Q}$ restricts to an isomorphism $ \omega_{\Sigma}' \cong \mc{Q}' \coloneqq \Ker ( \mathop{\textup{RES}})$.
\end{corollary}
 
 \subsubsection{Meromorphic differentials and super-Givental pairing}
Let us give a local description of the space of residueless meromorphic $1$-forms.

\begin{proposition}\label{SuperGiventalPairing}
	Consider the algebra $ V = \C (\! ( z )\! ) [ \theta ] $ of formal Laurent series, and $ \Omega_V$ its K\"{a}hler differentials. Let $ \varpi \in \Omega_V$ be  $ dz - d\theta z^a \theta $ for some $ a \in \Z $, and $ Q = \Omega_V/( \varpi )$. This gives an exact sequence
	\begin{equation}
		0 \rightarrow \C \to V \xrightarrow{[d]} Q \xrightarrow{\Res} \C \to 0\,.
	\end{equation}
Then, the image of $[d]$, i.e. the space of residueless elements $ Q' \coloneqq \im [d] \subset Q$, admits a super symplectic form, analogous to Givental's pairing, defined by
	\begin{equation}
		\< [df], [dg] \> = \Res_{z=0} [df \cdot g] \,.
	\end{equation}
	We call it the \emph{super Givental pairing}.
	\end{proposition}

\begin{proof}
	For the exact sequence, note that $ z^k d \theta \, \theta \in \Omega_V$ is not in the image of $ d$, because $ \theta^2 = 0$. With this observation, the proof is the analogous to the ordinary case.\par
	
	Let us see what this pairing looks like in practice. We first need that
\begin{equation}
	[d\theta \, z^l] = [d( \theta z^l ) - dz \, l z^{l-1} \theta ] = [d(z^l \theta )]\,,
\end{equation}
because $ [dz \, \theta] = [  d\theta \, z^a \theta^2] = 0$. From this, we find

\begin{equation}
	\< [ dz \, z^{k-1}] , [ dz\, z^{l-1}] \> = \frac{\delta_{k,-l}}{k} \,,
	\qquad
	\< [ dz \, z^{k-1}], [ d\theta \, z^l ]\> = 0 \,,
	\qquad
	\< [  d\theta \, z^k], [ d\theta \, z^l ]\> = \delta_{k+l-a+1} \, .
\end{equation}
	The pairing is well-defined because adding a constant to $g$ will not affect the value since $ \Res_{z=0} [df] = 0$. It is super antisymmetric by the usual integration by parts, and non-degenerate by inspection.
\end{proof}

\subsubsection{Multidifferentials}\label{sec:multi differential}

For our applications, we are interested in symmetric multidifferentials in the super setting.
\begin{definition}
An \emph{$n$-differential} on a super Riemann surface $ \Sigma $ is a meromorphic section of $\mc{Q}'^{\boxtimes n} \to \Sigma^n$. We often call it a multidifferential when $n$ is not specified. A $2$-differential is also called a bidifferential. 
\end{definition}

Let us focus on the case with $n=2$ and clarify the notion of symmetric bidifferentials. For an ordinary Riemann surface $\Sigma_0$, one can naturally consider the involution $\tau \colon \Sigma_0 \to \Sigma_0 \times \Sigma_0$ that exchanges the first and second copy. Strictly speaking, for $(p_1,p_2)\in\Sigma_0\times\Sigma_0$ and for a bidifferential $\omega$, one cannot naively compare $\omega(p_1,p_2)$ with $\omega(\tau(p_1,p_2)) = \omega( p_2, p_1)$ as they live in different vector spaces (fibres of the vector bundle). However, the generator $dz_1\otimes dz_2$ can be canonically identified with $dz_2\otimes dz_1$ as fibres are one-dimensional, hence the notion of symmetric bidifferentials is clear.

Fibres become $1|1$-dimensional in the super setting, so we need to introduce another involution to define symmetric differentials in a consistent way. More concretely, let us define a bundle morphism $\iota \colon \mc{Q}'^{\boxtimes2}\to \tau^* \mc{Q}'^{\boxtimes2}$ whose action on the generators of bidifferentials is given by
\begin{equation}
	\iota \colon [dz_1][dz_2] \mapsto [dz_2][dz_1] \,,
	\quad
	[dz_1][d\theta_2] \mapsto [d\theta_2][dz_1] \,,
	\quad
	[d\theta_1][d\theta_2] \mapsto -[d\theta_2][d\theta_1] \,.
\end{equation}
This morphism $\iota$ is an involution, in the sense that $ \tau^*(\iota) \circ \iota = \Id$.
The first action is a consequence of the last action i.e.
$ \iota ( [dz_1 ] [dz_2] ) = \iota ([ d\theta_1 ] [d\theta_2] \theta_2 \theta_1 ) = - [d\theta_2 ] [d\theta_1 ] \theta_2 \theta_1 = [dz_2] [dz_1]$ in superconformal coordinates. One can similarly check the consistency in Ramond coordinates.

\begin{definition}
	A bidifferential $\omega$ is said to be \emph{symmetric} if it satisfies $\tau^* \omega =\iota (\omega )$. 
\end{definition}

In general, a symmetric bidifferential in local coordinates $\su{z} = (z \, | \, \theta)$ is given by\
\begin{align}
		\omega(\su{z}_1,\su{z}_2)
		&=
		[d\theta_1][d\theta_2]( - \theta_1 \theta_2 s(z_1,z_2) - \theta_1 \alpha(z_1,z_2) + \theta_2 \alpha(z_2,z_1) + a(z_1,z_2) )\nonumber\\
		&=
		\frac{[dz_1][dz_2]}{z_1^kz_2^k}s(z_1,z_2)+\frac{[dz_1][d\theta_2]}{z_1^k}\alpha(z_1,z_2)+\frac{[d\theta_1][dz_2]}{z_2^k}\alpha(z_2,z_1)+[d\theta_1][d\theta_2]a(z_1,z_2)
\end{align}
for $s(z_1,z_2)=s(z_2,z_1)$ and $a(z_1,z_2)=-a(z_2,z_1)$, and $k=0$ if NS and $k=1$ if R. 

We extend the notion to symmetric $n$-differentials for any $n$ in an obvious way.

\subsubsection{Super quadratic Casimirs}

In the ordinary case, a quadratic differential is a section of $\Omega^{\otimes2}$, the square of the canonical sheaf $\Omega$ on a Riemann surface. Quadratic differentials play important roles in algebraic geometry and mathematical physics. In particular in the context of topological recursion, they are related to quadratic loop equations and Virasoro constraints. It turns out that the super analogue of a quadratic differential should be a section of $\mc{Q}^{\otimes3}$. 

In the ordinary case, any quadratic differential $q$ can be (non-uniquely) constructed by a bidifferential $\omega$ of $\Omega^{\boxtimes2}$ and specialising to the diagonal, assuming that $\omega$ has no poles there. Indeed, if $q(z)=f(z)dz^{\otimes2}$ in some local coordinates $z$, then any bidifferential $ \omega=dz_1\otimes dz_2 \, g(z_1,z_2)$ satisfying $g(z,z)=f(z)$ gives $q(z)=\omega(z,z)$. This perspective is crucial to define loop equations in the ordinary topological recursion. In the super setting, however, specialisation is problematic. 

Imagine a natural lift of a bidifferential $\omega$ to a super Riemann surface, i.e., in local coordinates $\su{z}=(z|\theta)$, consider $\omega(\su{z}_1,\su{z}_2)=[dz_1][dz_2]g(z_1,z_2)$. A problem arises upon specialisation, which gives $[dz][dz]=[d\theta]\theta[d\theta]\theta=0$ because $\theta^2=0$. Therefore, specialisation of a bidifferential in the ordinary setting gives a quadratic differential while it vanishes in the super setting. A similar problem also occurs upon specialisation of any multidifferential. 

How can one construct a section of $\mc{Q}^{\otimes3}$ from a bidifferential in the super setting, in particular, in such a way that it is compatible with super Virasoro constraints? We will show that the following definition captures all desired properties:

\begin{definition}\label{SuperQuadraticCasimir}
	We define the \emph{super quadratic Casimir operator} $\mc{C} \colon \mc{Q}'^{\boxtimes2}\to\mc{Q}'^{\otimes3}$ by the action on a symmetric bidifferential $\omega$ in any coordinates $\su{z}=(z|\theta)$ as follows:
	\begin{align}
		\mc{C} \colon \omega(\su{z}_1,\su{z}_2)=[d\theta_1][d\theta_2]g(\su{z}_1,\su{z}_2)\mapsto \mc{C}(\omega)(\su{z})\coloneqq [d\theta]^3
		\begin{cases}
			D_{\theta_1} g(\su{z}_1,\su{z}_2)\big|_{\su{z}_1=\su{z}_2=\su{z}} 
			&
			\text{NS}
			\\
			D^*_{\theta_1} g(\su{z}_1,\su{z}_2)\big|_{\su{z}_1=\su{z}_2=\su{z}} 
			&
			\text{R}
		\end{cases}
	\end{align}
\end{definition}

\begin{proposition}\label{prop:Hamiltonian}
	The super quadratic Casimir operator is well-defined. That is, it is independent of the choice of coordinates.
\end{proposition}

\begin{proof}
	Let us take the same notation as \cref{prop:SCtrans} and consider an NS-NS transformation. Recall that $D_\theta = (D_\theta\phi) D_\phi$ and $[d\phi] = [d\theta] (D_\theta\phi)$, hence we have
	\begin{align}
		\omega(\su{z}_1,\su{z}_2)
		=
		[d\phi_1][d\phi_2](D_{\theta_1}\phi_1)^{-1}(D_{\theta_2}\phi_2)^{-1}g(\su{z}_1,\su{z}_2)
		\eqqcolon
		[d\phi_1][d\phi_2]\tilde{g}(\su{x}_1, \su{x}_2).
	\end{align}
	Then,
	\begin{align}\label{sqCoCalc}
		D_{\phi_1}\tilde{g}(\su{x}_1, \su{x}_2)
		=
		(D_{\theta_1}\phi_1)^{-1} (D_{\theta_2}\phi_2)^{-1} \left((D_{\theta_1} (D_{\theta_1}\phi_1)^{-1}) g(\su{z}_1,\su{z}_2) + ( D_{\theta_1}\phi_1)^{-1} (D_{\theta_1} g(\su{z}_1,\su{z}_2)) \right)
	\end{align}
	Since $g(\su{z},\su{z})=0$ for a symmetric bidifferential, we have:
	\begin{equation}
		[d\phi]^3D_{\phi_1} \tilde{g}( \su{x}_1, \su{x}_2) \Big|_{\su{x}_1 = \su{x}_2 = \su{x}}
		=
		[d\phi]^3 (D_{\theta}\phi)^{-3}D_{\theta} g(\su{z}_1,\su{z}_2) \Big|_{\su{z}_1=\su{z}_2=\su{z}}
		=
		[d\theta]^3D_{\theta}g(\su{z}_1,\su{z}_2)\Big|_{\su{z}_1=\su{z}_2=\su{z}}.
	\end{equation}
	Thus, $\mc{C}(\omega)(\su{z})=\mc{C}(\omega)(\su{x})$ as claimed.

	Superconformal transformation of other types can be similarly shown by replacing $D_\theta$ with $D_\theta^*$ and/or $D_\phi$ with $D_\phi^*$.
\end{proof}

Let us justify why we consider the super quadratic Casimir operator. Given an even differential $\omega=y[dx]+\lambda [d\phi]$ in superconformal coordinates $\su{x}=(x \, | \, \phi)$, one can construct a bidifferential by doubling it, i.e., $\omega ^{\boxtimes2}(\su{z}_1,\su{z}_2) \coloneqq \omega(\su{z}_1)\otimes\omega(\su{z}_2)$. Then, the super quadratic Casimir operator $\mc{C}$ gives:
\begin{equation}\label{SuperQuadraticCasimirInCoordinates}
	\mc{C}(\omega ^{\boxtimes2})(\su{x})
	=
	[d\phi]^3(y\lambda + \phi(y^2+\lambda'\lambda))
	\eqqcolon
	[d\phi]^3(C_1+\phi C_0).
\end{equation}

If one takes $y$ as a free boson and $\lambda$ as a free fermion in the context of $\mc{N}=1$ superconformal field theory in two dimensions, then the $\phi$-dependent part $C_0=y^2+\lambda'\lambda$ and the $\phi$-independent part $C_1=y\lambda$ respectively correspond to the so-called \emph{energy-momentum tensor} and \emph{supercurrent}. Thus, after quantisation, $C_0$ and $C_1$ generate super Virasoro operators, which is the property we would like to have. We also note that under the superconformal rescaling $(x \,| \, \phi) \mapsto (cx \, | \, \sqrt{c}\phi)$, we have $C_0\mapsto c^{-2}C_0 $ and $ C_1 \mapsto c^{-\frac32}C_1$, and this shows that their conformal weights are $2$ and $\frac32$ respectively, which is consistent with the physics expectation.

\subsubsection{Twisted super quadratic Casimir operator}

In the construction of super loop equations, we are primarily interested in an R-NS map with simple ramification $f:\Sigma \to C$ with a (local) involution $\sigma_p:\Sigma\to\Sigma$ for $p\in\text{supp}(\mc{R})$ such that$ f \sigma_p = f$ -- in the ordinary setting, $\sigma_p$ corresponds to a deck transformation at a simple ramification point.\par
In terms of the local coordinates $\su{z}=(z \, | \, \theta) $, $\su{x}=(x \, | \, \phi)$ on $\Sigma$ and $C$ as described in \cref{prop:NS-R SCmap}, the action on the superconformal structure is given as $(\sigma_p)_*  D^*_\theta = -D_\theta^*$. In this setting, it is natural to define a $\sigma_p$-invariant super quadratic operator as below:

\begin{definition}\label{TwistedSuperQuadraticCasimir}
	Given the involution $\sigma_p$ as above, we define the \emph{twisted super quadratic Casimir operator} $\mc{C}^{\sigma_p}:\mc{Q}'^{\boxtimes2}\to\mc{Q}^{\otimes3}$ by the action on a symmetric bidifferential $\omega$ in any (superconformal/Ramond) coordinates $\su{z} = (z \, | \, \theta)$ as follows:
	\begin{multline}
		\mc{C}^{\sigma_p}:\omega(\su{z}_1,\su{z}_2)
		=
		[d\theta_1][d\theta_2]g(\su{z}_1,\su{z}_2)\nonumber\\
		\mapsto \mc{C}^{\sigma_p}(\omega)(\su{z})\coloneqq\frac12 [d\theta]^3
		\begin{cases}
			D_{\theta_1} g(\su{z}_1,\su{z}_2)\big|_{\su{z}_1=\sigma_p(\su{z}_2)=\su{z}} \;-\;D_{\theta_1} g(\su{z}_1,\su{z}_2)\big|_{\sigma_p(\su{z}_1)=\su{z}_2=\su{z}}
			&
			\text{NS}
			\\
			D^*_{\theta_1} g(\su{z}_1,\su{z}_2)\big|_{\su{z}_1=\sigma_p(\su{z}_2)=\su{z}} \;-\;D^*_{\theta_1} g(\su{z}_1,\su{z}_2)\big|_{\sigma_p(\su{z}_1)=\su{z}_2=\su{z}} 			&
			\text{R}
		\end{cases}
	\end{multline}
\end{definition}

\begin{proposition}
	The twisted super quadratic Casimir operator is well-defined. That is, it is independent of the choice of coordinates.
\end{proposition}

\begin{proof}
	Almost the same proof as \cref{prop:Hamiltonian} works. The only difference is that upon specialisation, we find $g(\sigma_p(\su{z}),\su{z})=-g(\su{z},\sigma_p(\su{z}))$, which may not disappear identically in contrast to $g(\su{z},\su{z})=0$. This contribution vanishes only after summing the two terms. 
\end{proof}

The twisted super quadratic Casimir operator is closely related to the notion of the twisted module of the $\mc{N}=1$ super Virasoro algebra of \cite{BCHORS20,BO21}. We also note that the twisted super quadratic Casimir operator is invariant under the involution $\sigma_p$ (the generator $[d\theta]^3$ is anti-invariant), hence it is well-defined not only on $\Sigma$ but also on $C$.

\section{Super fundamental bidifferential and periods}
\label{sec:the_symmetric_bidifferential}

In this section, we focus on the so-called fundamental bidifferential in the super setting. This object encodes geometric information of $\Sigma$. Let us start by reviewing properties of the fundamental bidifferential in the ordinary setting.

\subsection{Fundamental bidifferential}

In the context of Eynard--Orantin topological recursion, one piece of the data of a spectral curve is a symmetric bidifferential with double pole on the diagonal, $B$. On $\P^1$, this has the canonical form

\begin{equation}
	B (z_1, z_2) = \frac{dz_1 \, d z_2}{(z_1 - z_2)^2} = d_1 d_2 \log (z_1 - z_2)\,.
\end{equation}
It is a very important classical object, intimately related to the function theory of Riemann surfaces, and a super analogue is available in the literature. Let us start by recalling some basic function theory of Riemann surfaces, cf. e.g. \cite{Fay73,Mum83,Mum84}, along with its superconformal analogues from \cite{Wit15,DP15,Wit19}.\par

For an ordinary compact Riemann surface $ \Sigma_0$ of genus $g$, $ H_1 (\Sigma_0; \Z ) $ is a symplectic lattice of rank $2g$, and we will pick a symplectic basis (also called Torelli marking) $ \{ A_j, B_j \}_{j\in[g]}$ such that $ A_j \cap A_k = B_j \cap B_k = 0$ and $ A_j \cap B_k = \delta_{jk} $. By the integration pairing, we can consider $ H_1 (\Sigma_0; \Z) \hookrightarrow H^0( \Sigma_0; \Omega^1 )^*$. Then, we call the map
\begin{equation}\label{EvenPeriod}
	P_e 
	\colon
	H^0 (\Sigma_0 ; \Omega^1) \to H^1 (\Sigma_0; \mathbb{C}) 
	\colon
	\omega \mapsto \Big(  \gamma \mapsto \int_\gamma \omega \Big)
\end{equation}
the \emph{(even) period map}. The $g$-dimensional space $ H^0 ( \Sigma_0; \Omega^1 )$ has a unique basis $ \{ \omega_j \}_{j=1}^g $ normalised on $A$-cycles: $ \int_{A_j} \omega_k = \delta_{jk} $. Moreover, the \emph{period matrix} $ \Omega_{jk} \coloneqq \int_{B_j} \omega_k$ is symmetric with positive definite imaginary part. 

The \emph{fundamental bidifferential} (\emph{of the second kind}) $ B$ is a symmetric bidifferential on $ \Sigma_0^2$, uniquely characterised by 
\begin{itemize}
	\item Its only pole is a double one along the diagonal, with principal part $ d_1 d_2 \log (z_1 - z_2)$;
	\item $ \int_{A_j} B = 0$.
\end{itemize} 
It satisfies $ \int_{B_j} B = \omega_j$, and on the universal cover of $ \Sigma_0$, we can define the third kind differential with poles only at $ p$ and $q$ by $\omega^{p-q}(x) = \int_{y = q}^p B(x,y)$. We note that for a non-compact Riemann surface, there is no unique fundamental bidifferential as $H_1(\Sigma_0,\mathbb{Z})$ is not a symplectic lattice any more. \par

In the ordinary setting, $B$ is constructed as follows: to a Riemann surface of genus $g$, we associate its Jacobian $ \Jac(\Sigma_0)$, i.e. the abelian variety of dimension $g$ parametrising the degree zero divisor classes on $\Sigma_0$. Alternatively, $ \Jac (\Sigma_0) \cong H^0 (\Sigma_0; \Omega^1 )^*/ H_1(\Sigma_0; \Z )$, with the map given by $ D \mapsto ( \omega \mapsto \int_D \omega )$. Namely, given $D$, one can take a line integral with $D$ as the signed boundary and integrate one-forms along it -- but this is only well-defined up to integrals along closed cycles, i.e. $ H_1 (\Sigma_0;\Z)$.

Using the Torelli marking, $ \Jac (\Sigma_0 ) \cong \C^g/ (\Z^g \oplus \Omega \Z^g)$, where $ \Omega $ is the period matrix, via $ D \mapsto \{ \int_D \omega_j \}_{j=1}^g$. On this space, we have the theta functions

\begin{equation}
	\vartheta \begin{bmatrix} \delta \\ \epsilon \end{bmatrix} (z; \Omega ) = \sum_{n \in \Z^g} \exp \Big( \pi i ( n+ \delta )^T \Omega (n + \delta ) + 2 \pi i ( n + \delta )^T (z + \epsilon )\Big) \,,
\end{equation}
where $ \begin{bmatrix} \delta \\ \epsilon \end{bmatrix} \in ( \Z^g + \Omega \Z^g ) \otimes_\Z \R \cong \C^g$ and $ z \in \C^g$. This is not quite a function on $ \Jac (\Sigma_0)$, but it is a section of an appropriate line bundle -- accounting for twists along $B$-periods.\par
Given a pair of points $ (x, y) \in \Sigma_0^2$, $ [x] - [y] \in \Jac (\Sigma_0)$, and this gives a variant of the Abel--Jacobi map. For $ g (\Sigma_0) > 0$, we define the prime form in $\Omega^{-\frac12}_{\Sigma_0}\boxtimes \Omega^{-\frac12}_{\Sigma_0}$ by choosing an odd spin structure $ \alpha$, with a section $ h_\alpha$, and setting
\begin{equation}
	E (x,y) \coloneqq \frac{\vartheta [\alpha ]( \int_x^y \omega_1, \dotsc, \int_x^y \omega_g )}{h_\alpha (x) h_\alpha (y)}\,.
\end{equation}
This is independent of the choice of $ \alpha$. For $ g (\Sigma_0) = 0$, there are no odd spin structures, and we define $E (x,y ) = \frac{x-y}{\sqrt{dx \, dy}}$\,. Then
\begin{equation}
	B(x,y) \coloneqq d_x d_y \log E (x,y)\,.
\end{equation}
This is completely independent of any choices, apart from the Torelli marking.

\subsection{Odd periods and super fundamental bidifferential}\label{sec:odd periods}

In the super setting, we should first of all recall that a form being closed is a non-trivial condition. Following \cref{ClosedForms,ExtendedBerezinianIsRestrictedQ}, we identify closed forms with holomorphic sections of $ \mc{Q}'$. 

\begin{lemma}[{\cite{Wit15}}]
	On a compact super Riemann surface $\Sigma$ of genus $g$ with $\deg \mc{R} = 2r$, a Torelli marking $\{ A_j, B_j \}_{j\in[g]}$ of $ \Sigma_\textup{red}$ can be lifted, uniquely in homology, to an even Torelli marking $ \{ A_j, B_j \}_{j\in[g]}$ on $ \Sigma$.
\end{lemma} 

As a consequence, we can write the even period map \cref{EvenPeriod} also in the super setting, and in families:
\begin{equation}
	P_e 
	\colon
	\pi_* \Omega^1_{\Sigma/T} \to R^1\pi_* \Z \otimes \mc{O}_T 
	\colon
	\omega \mapsto \Big(  \gamma \mapsto \int_\gamma \omega \Big)
\end{equation}
In this language, the even periods $ \{ a_j = \int_{A_j}, b_j = \int_{B_j} \}_{j\in[g]}$ are just a trivialisation of the vector bundle $ R^1 \pi_* \Z \otimes \mc{O}_T$.\par
But there are additional, odd periods $ \{ \alpha_\mu, \beta_\mu \}_{\mu\in[r]}$: namely odd residues\footnote{These are not the residues defined in \cref{def:Residue}; those are even.} at the Ramond divisor (cf. \cite{Wit15}). We give the definition of \cite[Section 4]{DO24}.

\begin{definition}\label{def:super periods}
	The \emph{odd period map} is the pushforward of the map
	\begin{equation}
		\omega_{\Sigma/T}' = \Ker (d \circ \big|_\mc{R}) \to \Ker d = \mc{Z}^1_{\mc{R}/T}
	\end{equation}
	along the structure map $ \pi \colon \Sigma \to T$:
	\begin{equation}
		P_o \colon \pi_* \omega_{\Sigma/T}' \to \pi_* \mc{Z}^1_{\mc{R}/T} \,.
	\end{equation}
\end{definition}

\begin{lemma}[{\cite[Lemma 2.2]{DO24}}]
	The codomain of $P_o$ is a $ 0|2r$-dimensional vector bundle associated to a local system of free abelian groups generated by $ \{ d\theta_p\}$, where $ (z_p \, | \, \theta_p )$ are local coordinates near the Ramond punctures $p\in\textup{supp}(\mc{R})$,  with a symmetric bilinear pairing $ \frac{1}{2 \pi i} \sum_{p} d\theta_p^2$. If the Ramond divisor is oriented, cf. \cref{DefOrientation}, the local system (and hence the vector bundle) is trivialised by the well-defined $\frac{d\theta_p}{\sqrt{2 \pi i}}$.
\end{lemma}

\begin{remark}
	Witten's original definition assumes an orientation and is given in local Ramond coordinates, as follows.
	For $ p \in \mathop{\textup{supp}} \mc{R}$, let $ (z_p \, | \, \vartheta_p) $ be local Ramond coordinates at it. Then, for a closed one-form $\tau$, we write $ \tau = \frac{d \vartheta_p}{\sqrt{2 \pi i}} \, w_{\tau,p} \pmod{z_p} $, and the assignment $ w_p \colon \tau \mapsto w_{\tau,p} $ is called an \emph{odd period}.
\end{remark}

We will slightly reformulate and extend this as follows.

\begin{lemma}\label{OddPeriodsAsResidues}
	Let $p$ be a Ramond puncture. Then the odd period at $p$ is given by
	\begin{equation}\label{OddPeriodSuperGivental}
		w_p ( \omega ) 
		=
		\frac{d\theta}{\sqrt{2 \pi i}} \cdot \sqrt{2\pi i} \< [d\theta], \omega \>
		=
		d\theta \< [d\theta], \omega \> \in \mc{Z}^1_{\mathcal{R}/S} \,,
	\end{equation}
	where $ \< \cdot, \cdot \> $ is the super-Givental pairing of \cref{SuperGiventalPairing} and $ (z \, | \, \theta )$ are any Ramond coordinates at $ p$.
\end{lemma}

\begin{proof}
	The super Givental pairing restricted to local holomorphic sections near $p$ is only non-zero in the case $ \< [d\theta ], [d\theta]\> = 1$. 
	As $ d\theta \pmod{z}$ is independent of the choice of Ramond coordinate up to sign by \cref{RamondAut}, and \cref{OddPeriodSuperGivental} is quadratic in $ d\theta $, this sign does not matter, and the right side of \cref{OddPeriodSuperGivental} is well-defined.\par
	For a holomorphic local section $ \omega $ of $ \omega_{\Sigma/T}'$ (without loss of generality of fixed super degree), we see that locally it is of the form
	\begin{equation}
		\omega 
		= 
		[dz | d\theta ] \Big( \frac{h(z)}{z} + \theta f(z)  \Big) \cong [d\theta] \Big( h(z) + \theta z f(z) \Big)\,,
	\end{equation}
	with $f $ and $ h$ holomorphic, and in this notation $ w_p (\omega ) = d\theta h(0)$, while also
	\begin{equation}
		\begin{split}
			d\theta \< [d\theta ], \omega \>
			&=
			(-1)^{|\omega|+1} d\theta \Res_{z = 0} [ d\theta \big( h(z) + \theta z f(z)\big) \cdot \theta ]
			\\
			&=
			(-1)^{|\omega|+1} d\theta \Res_{z = 0} [ d\theta h(z) \theta ]
			\\
			&=
			d\theta \Res_{z = 0} \left[\frac{dz}{z}h(z) \right]
			=
			d\theta h(0)\,,
		\end{split}
	\end{equation}
	using the super skew symmetry of the pairing and the fact that $ |h| = |\omega| +1$.
\end{proof}

\begin{definition}\label{def:Torelli}
	An \emph{odd Torelli marking} on a super Riemann surface is an orientation of all Ramond punctures along with a  labelling of them as $ \{ p_\mu, q_\mu \}_{\mu \in[r]}$. Given such a marking, we use the trivialisation given by the orientation to identify $ \pi_* \mc{Z}_{\mc{R}/S}^1 \cong \C^{0|2r}$ and define $\{\alpha_\mu, \beta_\mu\}_{\mu \in [r]}$ where $\alpha_\mu = \frac{1}{\sqrt{2}} ( w_{p_\mu} + i w_{q_\mu} ) $ and $ \beta_\mu = \frac{1}{\sqrt{2}} ( w_{p_\mu} - i w_{q_\mu} )$.\par
	A \emph{Torelli marking} on a super Riemann surface is the union of an even and an odd Torelli marking.\par
	The \emph{super period space} is the space spanned by $ \{a_j,b_j\}_{j\in[g]}$ and $\{\alpha_\mu,\beta_\mu\}_{\mu \in [r]}$. It is a $ 2g | 2r$-dimensional super vector space with supersymplectic intersection form $ \sum_j a_j \wedge b_j - \sum_\mu \alpha_\mu \wedge \beta_\mu$.
\end{definition}

The normalisation of the intersection form is the reason for the factors $ \frac{1}{\sqrt{2 \pi i}}$ in the trivialisations.

\begin{theorem}[{\cite{Wit15}}]\label{BadLocus}
	A compact super Riemann surface of genus $g$ with $ \deg \mc{R} = 2r$ has a space of holomorphic one-forms of dimension at least $ g | r$. Outside of a certain \emph{bad locus} on the moduli space of super Riemann surfaces, the dimension is exactly $ g|r$. Given a Torelli marking,  there is a unique basis $ \{ \omega_1, \dotsc, \omega_g \, | \, \nu_1, \dotsc, \nu_r \} $ of closed holomorphic one-forms such that
	\begin{equation} \label{Normalisation}
		\begin{alignedat}{3}
			a_j (\omega_k ) &= \delta_{jk} \,, & \quad a_j ( \nu_\mu ) &= 0 \,,
			\\
		\alpha_\rho (\omega_k ) &= 0 \,,  & \quad \alpha_\rho (\nu_\mu ) &= \delta_{\rho \mu} \,.
		\end{alignedat}
\end{equation}
\end{theorem}

\begin{definition}
	Let $ \Sigma$ be a compact super Riemann surface with Torelli marking of genus $g$ and $\deg \mc{R} = 2r$. Its \emph{super genus} is $ \su{g} = (g \, | \, r)$. If $ \Sigma$ is outside the bad locus, its \emph{super period matrix}, which is graded-symmetric,  is defined as
	\begin{equation}\label{BasisFromSuperB}
		\begin{alignedat}{3}
			\hat{\Omega}_{jk} & \coloneqq b_j ( \omega_k) \,, & \quad \hat{\Omega}_{j\mu} & \coloneqq b_j (\nu_\mu) \,,
			\\
			\hat{\Omega}_{\rho k} & \coloneqq \beta_\rho (\omega_k )\,, & \quad \hat{\Omega}_{\rho\mu} & \coloneqq \beta_\rho (\nu_\mu) \,.
		\end{alignedat}
	\end{equation}
\end{definition} 

\begin{theorem}[{\cite[Theorem 8]{DO23}}]
	For $ 1 \leq r < g$, the bad locus has codimension $ \leq r$. For $ r \geq g$, the bad locus is empty.
\end{theorem}

\begin{corollary}\label{SimplyRamifiedNotBad}
	Super Riemann surfaces with simply-ramified maps to $ \P^{1|1}$ are not in the bad locus.
\end{corollary}

\begin{proof}
	For a simply-ramified map $ \Sigma \to \P^{1|1}$ of degree $d$,  the Ramond divisor is exactly the ramification divisor. By the Riemann--Hurwitz formula, $ r = g + d  - 1 \geq g$, so the bad locus is empty.
\end{proof}

The proper construction of a super analogue of $B$ would proceed by defining a super Jacobian as a $ g|r$-dimensional abelian supervariety and defining the prime form via the theta function and the Abel--Jacobi map. One issue with this, is that we may need $ \hat{\Omega} $ to have some positive-definiteness property, analogous to the positive-definite imaginary part for the non-super case, to make the sum in the definition of the theta function converge. To the best of authors' knowledge, it is not known whether this holds.

However, there is another construction available in \cite{DP15}. This construction is in the `supergravity formulation', i.e. constructions are carried out first on the reduced Riemann surface, then lifted to split super Riemann surfaces, and then transported along deformation by a `gravitino', i.e. the odd moduli of super Riemann surfaces. Both the fundamental bidifferential and the right basis of holomorphic one-forms are constructed in \cite{DP15}, which we will briefly review here.

Let us first construct the so-called \emph{Szeg\H{o} kernel}. For a super Riemann surface with Torelli marking of Ramond divisor $\mathcal{R}$, let us denote by $ \Sigma_\textup{red}$ its reduced Riemann surface, by $\mathcal{R}_p = ( p_\mu )_{\mu \in [r]}$, $\mathcal{R}_q = (q_\mu)_{\mu \in [r]}$ the divisor on $ \Sigma_\textup{red}$ associated to the Ramond divisor $\mathcal{R}$ where $\mathcal{R}=\mathcal{R}_p\sqcup\mathcal{R}_q$, and by $\mc{S}_{\Sigma_\textup{red}}$ its twisted spin bundle i.e. $\mc{S}_{\Sigma_\textup{red}}^{\otimes2}=\Omega^1_{\Sigma_\textup{red}}(-\mathcal{R}_p\sqcup\mathcal{R}_q)$. We then define a meromorphic section of $\mc{S}_{\Sigma_\textup{red}}^{\boxtimes2}$
\begin{equation}
	S^{pq}(z_1,z_2)
	\coloneqq
	\frac{\theta[\alpha] (z_1-z_2+\frac12 \mathcal{R}_p-\frac12\mathcal{R}_q)}{\theta[\alpha](\frac12 \mathcal{R}_p-\frac12\mathcal{R}_q) E(z_1,z_2)} \prod_{\mu \in [r]} \left(\frac{E(z_1,p_\mu) E(z_2,q_\mu)}{E(z_1,q_\mu) E(z_2,p_\mu)}\right)^{\frac12}
\end{equation}
which is independent of the choice $\alpha$ similar to the prime for $E(z_1,z_2)$. We note that $S^{pq}(z_1,z_2)\neq -S^{pq}(z_1,z_2)$ in general, but one can construct an antisymmetric form as below:

\begin{proposition}[\cite{DP15}]
	Let $h_{p_\eta},h_{q_\eta}$ be the evaluation of $S^{pq}$ along $\mathcal{R}_{p}, \mathcal{R}_q$ with respect to the second argument. For any anti-symmetric matrix $M$, $S^M$ defined below satisfies $S^{M}(z_1,z_2)= -S^{M}(z_1,z_2)$
	\begin{equation}
		S^M(z_1,z_2)
		=
		S^{pq}(z_1,z_2)+\frac12\sum_{\mu,\rho\in[r]}(M_{\mu\rho}-\delta_{\mu\rho})h_{p_\mu}(z_1)h_{p_\mu}(z_2).
	\end{equation}
	Furthermore, there exists a unique $M$\footnote{$M$ is explicitly given in \cite[Eq. (3.34)]{DP15}. Also see \cite[Eq. (2.21)]{DP15} for the connection with the period matrix.} such that $S^M$ vanishes along $\alpha_\mu$ and $(\nu_{p_\mu} = \frac{1}{2 \sqrt{\pi i}}\beta_\mu(S^M))_{\mu\in[r]}$ forms a basis of closed holomorphic sections on $\mc{S}_{\Sigma_\textup{red}}$, where $(\alpha_\mu,\beta_\mu)_{\mu \in [r]}$ is defined analogously to \cref{def:Torelli}.
\end{proposition}

\begin{definition}
	We call such a unique $S^M$ the \emph{Szeg\H{o} kernel} and denote it simply by  $S$. When a super Riemann surface $\Sigma$ is split, we can canonically lift the fundamental bidifferential $B$ and the Szeg\H{o} kernel $S$ to meromorphic sections of $\mathcal{Q}'^{\boxtimes2}_{\Sigma}$. The \emph{super fundamental bidifferential} $\su{B}$ on $\Sigma$ is then defined by $\su{B}=B+S$ viewed as sections of $\mathcal{Q}'^{\boxtimes2}_{\Sigma}$.
\end{definition}

For a non-split super Riemann surface, the construction is more involved as there is no canonical lift, and it is beyond the scope of our interests. We refer to the readers \cite[Section 3]{DP15} and \cite{Kes19} for more details. For now, let us state a summary of their results:
\begin{proposition}\cite{DP15,Kes19}
Given a super Riemann surface $\Sigma$, one can uniquely determine a metric and a \emph{gravitino} on the underlying Riemann surface $\Sigma_{\textup{red}}$, up to conformal and super Weyl transformations. With these data, there is a way to construct the fundamental bidifferential $\su{B}$ on $\Sigma^2$ whose pole structure is the same as that on $\Sigma_{\textup{red}}$ and whose periods encode the super period matrix in the same manner as on split super Riemann surfaces.
\end{proposition}

%
%
%

\subsection{Local pictures}

The discussion so far describes a global structure of super Riemann surfaces. Although a proposed construction of \cite{DP15} is summarised above, no literature explicitly computes the super fundamental bidifferential $\su{B}$ except in the $g=0$ split case. Fortunately, all we need is local data, which can be thought of as a (formal) deformation from the the split case. Let us thus explicitly give the super fundamental bidifferential on $\mb{P}^{1|1}$, which is a genus-$0$ super Riemann surface, and that on $\mb{P}(1,1|0)$, which is a genus-$0$ super Riemann surface with a degree-$2$ Ramond divisor $\mc{R}$.

\subsubsection{Fundamental bidifferential near NS points}

Consider $ \P^{1|1} $ from \cref{P1|1}. Recall that the skew-symmetric bilinear form is given by
\begin{equation}
	\< (u_1, v_1| \theta_1 ), (u_2, v_2 | \theta_2) \> 
	=
	u_1 v_2 - u_2 v_1 - \theta_1 \theta_2\,.
\end{equation}
This is not well-defined on  $ \P^{1|1}$ due to rescaling, but it does define a section
\begin{equation}
	\mb{E} ( \su{z}_1, \su{z}_1) \coloneqq \frac{\< \su{z}_1, \su{z}_2 \>}{ [d \theta_1] [d\theta_2]}
\end{equation}
of $ \mc{Q}^{-1} \boxtimes \mc{Q}^{-1}$, which is the (super) prime form. In particular, it satisfies
\begin{equation}
	\su{E} (\su{z}_1, \su{z}_2) = [d\theta]D_\theta \su{E} (\su{z}_1, \su{z}_2) = 0 
	\qquad 
	\Leftrightarrow
	\qquad
	\su{z}_1 = \su{z}_2\,,
\end{equation}
in accordance with \cite[(5.38)]{DP89}. Then $\su{B}$ is the mixed second derivative of the logarithm of the prime form, projected to $ \mc{Q} \boxtimes \mc{Q} $. We write it on the open neighbourhood where $ v \neq 0$, so we may take a coordinate $ z = \frac{u}{v}$.
\begin{equation}
	\begin{split}
		\su{B} ( \su{z}_1, \su{z}_2 ) 
		&= 
		[ d_1 d_2 \log ( z_1 - z_2 - \theta_1 \theta_2 ) ]
		\\
		&= 
		[d \theta_1] D_{\theta_1} [d \theta_2] D_{\theta_2} \log ( z_1 - z_2 - \theta_1 \theta_2 )
		\\
		&=
		- [d\theta_1] [d\theta_2] \Big( \frac{1 }{z_1 - z_2 } + \frac{\theta_1 \theta_2}{ (z_1 - z_2)^2 } \Big)
		\\
		&=
		\frac{[dz_1] \, [dz_2]}{(z_1 - z_2)^2} - \frac{[d \theta_1] [d\theta_2]}{z_1 - z_2}\,,
	\end{split}
\end{equation}
where in the last line, we used $ [d \theta] \cdot \theta = [dz] $. Note also that 
\begin{equation}
	d_{\su{z}} \su{E} (\su{z}, \su{w}) \Big|_{\su{w}=\su{z}} = \frac{\varpi_\theta}{[d\theta]^2}\neq0
\end{equation}
gives the isomorphism $ \mc{Q}^2 \cong \mc{P}$.

\subsubsection{Fundamental bidifferential near Ramond punctures}
\label{exa:BinR}

Consider $ \P (1,1|0)$ from \cref{exa:P11|0}. Recall that a skew-symmetric bilinear form is
\begin{equation}
	\< (u_1, v_1| \theta_1 ), (u_2, v_2 | \theta_2) \> 
	=
	u_1 v_2 - u_2 v_1 -\frac12 ( u_1 v_2 + u_2 v_1) \theta_1 \theta_2
\end{equation}
This again induces a prime form by $ \mb{E} ( \su{z}_1, \su{z}_1) \coloneqq \frac{\< \su{z}_1, \su{z}_2 \>}{ \sqrt{\varpi_1^* \varpi_2^*}}$. This is now a meromorphic section of $ \mc{P}^{-1/2} \boxtimes \mc{P}^{-1/2}$, not of $ \mc{Q}^{-1} \boxtimes \mc{Q}^{-1}$ -- recall that the $\mb{C}^*$-action keeps $\theta$ (and $d\theta$) invariant while $\varpi^*$ is of weight $1$. Again, following \cite[(5.38)]{DP89}\footnote{That paper does not deal with Ramond punctures, but the condition we consider generalises directly.}, the prime form should satisfy
	
\begin{equation}
	\su{E} (\su{z}_1, \su{z}_2) =[d\theta]D_\theta^* \su{E} (\su{z}_1, \su{z}_2) = 0 
	\qquad 
	\Leftrightarrow
	\qquad
	\su{z}_1 = \su{z}_2\,
\end{equation}
which one can check explicitly. We now get
\begin{equation}
	d_{\su{z}} \su{E} (\su{z}, \su{w}) \Big|_{\su{w}=\su{z}} = 1\,,
	\qquad
	[d\theta_z]D_{\theta_z}^*[d\theta_w]D_{\theta_w}^* \su{E} ( \su{z},\su{w}) \Big|_{\su{w} = \su{z}}  = \frac{[d\theta]^2 z}{\varpi^*} \,,
\end{equation}
which exhibits the isomorphism $ \mc{P} \cong \mc{Q}^2 (- \mc{F})$. Furthermore, similarly to before, we obtain

\begin{equation}\label{RamondB}
	\begin{split}
		\su{B} ( \su{z}_1 , \su{z}_2 ) 
		&= 
		[d_1 d_2 \log ( z_1 - z_2 - \frac{1}{2} (z_1 + z_2) \theta_1 \theta_2 ) ]
		\\
		&= 
		- [d\theta_1] [d\theta_2] D^*_{\theta_1} D^*_{\theta_2} \log ( z_1 - z_2 - \frac{1}{2}(z_1 + z_2) \theta_1 \theta_2 )
		\\
		&= 
		[d \theta_1] [d\theta_2] \Big( - \frac{(z_1 + z_2)}{2(z_1 - z_2)} - \frac{z_1 z_2 \theta_1 \theta_2}{(z_1 - z_2)^2} \Big)
		\\
		&= 
		\frac{[dz_1] \, [dz_2]}{(z_1 - z_2)^2} -[d\theta_1] [d\theta_2] \frac{(z_1 + z_2)}{2(z_1 - z_2)}.
	\end{split}
\end{equation}


\subsubsection{Odd periods on \texorpdfstring{$ \P(1,1|0)$}{P110}}\label{sec:local odd periods}

Let us continue focusing on $\mb{P}(1,1|0)$. Setting $ p = 0$ and $ q = \infty$, such that $ \alpha = \frac{1}{\sqrt{2}} (w_0 + i w_\infty ) $ and $ \beta = \frac{1}{\sqrt{2}} (w_0 - i w_\infty ) $, we find that

\begin{equation}
	\alpha ([d\theta] ) = 2 \sqrt{\pi i}\,,\quad \beta ([d\theta] ) = 0.
\end{equation}

Accordingly, the unique normalised one-form is $ \nu = \frac{1}{2 \sqrt{\pi i}} [d\theta]$. Therefore, the odd period can be non-zero even when $g=0$ as long as the Ramond divisor is not trivial. However, we also see that the period matrix in this case is $ \hat{\Omega} = \beta(\nu ) = 0$, which is consistent, as it should be a $ (0|1) \times (0|1)$ dimensional graded-symmetric matrix. It is also consistent with the fact that the moduli space of super Riemann surfaces of genus $0$ with one marked Neveu-Schwarz point and two Ramond punctures is of dimension $0|0$, i.e. a point (cf. \cite{Wit19}). Note that by taking the periods of $ \su{B}$ with respect to the first variable, we get 
\begin{equation}
	\begin{split}
		\alpha (\su{B}(\su{z}_1,\su{z}_2))
		&			
		= - \frac{1}{2\sqrt{2}} (w_0 + i w_\infty) [d\theta_1] [d\theta_2]\frac{(z_1 + z_2)}{z_1 - z_2}
		\\
		&
		= \frac{\sqrt{\pi i}}{2} \Big([d\theta_2]  + i [d\psi_2] \Big)=0,
	\end{split}
\end{equation}

Similarly, 
\begin{equation}
	\beta (\su{B}(\su{z}_1,\su{z}_2)))=2 \sqrt{\pi i} \nu (\su{z}_2) 
\end{equation}
which conforms with \cref{Normalisation,BasisFromSuperB}. Thus by analogy, one may think of $\alpha$ as the odd $A$-period and of $\beta$ as the odd $B$-period.

\subsubsection{Ramond coordinate transformation}

	Instead of having Ramond punctures at $0$ and $\infty$, one may wish to consider coordinates $\tilde{z}$ such that $\tilde{z}=c_1,c_2$ are the Ramond punctures by the following M\"{o}bius transformation:
	\begin{equation}
    	\tilde{z}=\frac{c_2z_0-c_1}{z_0-1}\,, 
    	\quad
    	\tilde{\theta}=\theta_0\,
    	\quad 
    	\tilde{\varpi}=d\tilde{z}-\frac{(\tilde{z}-c_1)(\tilde{z}-c_2)}{c_1-c_2}\tilde{\theta} d\tilde{\theta}\,,
    	\quad 
    	\varpi_0=\left(\frac{c_1-c_2}{(\tilde{z}-c_2)^2}\right)\cdot\tilde{\varpi} \,.\label{pi*}
	\end{equation}
	where $(z_0|\theta_0)$ are local coordinates around $(0|0)$.
	
	Be careful that \eqref{RtoR} holds only when $x=z=0$ is a Ramond puncture, which is not the case for the above transformation. \Cref{pi*} is basically the same as \cite[Eq. (4.9)]{Wit19}, except for the normalisation in the denominator. In these coordinates, \eqref{RamondB} can be written as
	\begin{equation}
    	\su{B} = \frac{[d\tilde{z}_1][d\tilde{z}_2]}{(\tilde{z}_1-\tilde{z}_2)^2} - [d\tilde{\theta}_1][d\tilde{\theta}_2]\frac{1}{2}\frac{N(\tilde{z}_1,\tilde{z}_2)}{\tilde{z}_1-\tilde{z}_2} \,, 
    	\quad 
    	N(\tilde{z}_1,\tilde{z}_2) = \frac{(\tilde{z}_1-c_1)(\tilde{z}_2-c_2)+(\tilde{z}_1-c_2)(\tilde{z}_2-c_1)}{c_1-c_2}\, .
	\end{equation}

\subsubsection{Projection property near Ramond punctures}\label{sec:projection}

The form $ \su{B}$ on $ \P (1,1|0)$ has an expansion for $ |z_1| < |z_2|$ as
\begin{equation} \label{Blocal}
	\begin{split}
		\su{B} ( \su{z}_1 , \su{z}_2 ) 
		&= 
		\frac{[dz_1] \, [dz_2]}{(z_1 - z_2)^2} -[d\theta_1] [d\theta_2] \frac{(z_1 + z_2)}{2(z_1 - z_2)}
		\\
		&= 
		\frac{[dz_1] \, [dz_2]}{z_2^2(1 - \frac{z_1}{z_2})^2} + [d\theta_1] [d\theta_2]
		\frac{1}{z_2} \frac{(z_1 + z_2)}{2(1-\frac{z_1}{z_2} )}
		\\
		&= 
		\sum_{k=0}^\infty (k+1)  [dz_1] \, [dz_2]\frac{1}{z_2^2} \left( \frac{z_1}{z_2} \right)^k+ 
		\sum_{k=0}^\infty  [d\theta_1] [d\theta_2]\frac{1}{2z_2} (z_1 + z_2)\left( \frac{z_1}{z_2} \right)^k
		\\
		&= 
		\sum_{l=1}^\infty l  [dz_1][dz_2]z_1^{l-1} \, z_2^{-l-1} + \frac{1}{2} [d\theta_1] \, [d\theta_2 ] +
		\sum_{k=1}^\infty [d\theta_1][d\theta_2]]z_1^k z_2^{-k} \,.
	\end{split}
\end{equation}
Comparing this with the super symplectic form from \cref{SuperGiventalPairing}, we find that for any local meromorphic section $ \omega (\su{z}) = \sum_{k\in \Z} [dz] u_k z^{k-1} +[d\theta] \upsilon_k z^k $ of $ \mc{Q}'$,
\begin{equation}
	\< \su{B} (\su{z}_1,\su{z}_2), \omega (\su{z}_1) \>
	=
	\Res_{z_1 = 0} \int_0^{\su{z}_1} \su{B} (\cdot, \su{z}_2 ) \omega (\su{z}_1) 
	= 
	\frac{1}{2}[d\theta_2] \tau_0  +\sum_{k<0}\left(  [dz_2 ] u_k z_2^{k-1}+ [d\theta_2]\upsilon_k z_2^{k} \right) \,.
\end{equation}
where the pairing is taken with respect to $\su{z}_1$. So we find that pairing with $ \su{B}$ is nearly a projector on the principal part near $ 0$, except that the global superconformal form $ d \theta$ is split over the two principal parts. It turns out that the projection property becomes less simple when one considers polarisation, which we will discuss in \cref{sec:SAStoTR}. In this example, we also have
\begin{equation}
	\Big( \Res_{z_1 = 0} + \Res_{z_1 = \infty} \Big)  \int_0^{\su{z}_1} \su{B} (\cdot, \su{z}_2, ) \omega (\su{z}_1) = \omega (\su{z}_2) \, .
\end{equation}

\begin{remark}
	Now, let us compare this to \cite{BO21}. The formula for the standard $ \omega_{0,0|2}$ is in \cite[(2.24)]{BO21}:

	\begin{equation}\label{BOB}
		\omega_{0,0|2} = - \frac{1}{2} \frac{z_1 + z_2}{z_1 - z_2} \frac{\Theta_1 \Theta_2}{z_1 z_2}\,,
	\end{equation}
	where $ \Theta $ is such that $ \Theta^2 = zdz$. More precisely, we need the pairing \cite[(2.17)]{BO21} to make sense. Comparing to \cref{SuperGiventalPairing} with $ a = 1$, we find that $ \frac{\Theta}{z} $ should correspond to $ d \theta $ in this way. This identifies \cref{BOB} with the second term in \cref{RamondB}. On a heuristic level, this is not surprising, as $ \theta$ and $d\theta $ behave similarly: over the reduced Riemann surface, they generate isomorphic bundles. So the distinction between them is lost easily in a local picture, and hence the heuristic of $ \Theta^2 \simeq z dz \simeq [dz]z = [d\theta] \theta z^2$.
\end{remark}

\subsection{Odd periods for meromorphic differentials}

\Cref{def:super periods,OddPeriodsAsResidues} are well-defined for any holomorphic section of $\mc{Q}'$ because no matter how local coordinates are taken, the difference appears only in mod $z_p$ parts and the leading behaviour remains the same.\par
However, we are interested in meromorphic sections of $\mc{Q}'$. For these, \cref{def:super periods} does not make sense any more, because it would restrict a meromorphic form to a locus where it has a pole, while \cref{OddPeriodSuperGivental} is not well-defined because non-leading orders depend on the choice of coordinates. Let us show this by example.

\begin{example}
	Let $p$ be a Ramond puncture with local Ramond coordinates $ \su{z} = ( z \, | \, \theta )$. Then also
	\begin{equation}
		\hat{\su{z}} = ( \hat{z} \, | \, \hat{\theta} )
		=
		(z e^{2z + \frac{z^2}{2}} \, | \, (1+z) \theta )
	\end{equation}
	are local Ramond coordinates, as can easily be checked with \cref{RtoR}.\par
	Taking $ \omega = [d \hat{\theta} ] \hat{z}^{-1}$, we find that
	\begin{align}
		\< [d\hat{\theta}], \omega \> 
		&=
		0 \,;
		\\
		\< [d\theta ], \omega \>
		&=
		-1 \,,
	\end{align}
	while clearly $ d\theta|_p = d\hat{\theta} |_p$. So \cref{OddPeriodSuperGivental} does not extend to meromorphic forms.
\end{example}

In the context of topological recursion, we have more data to use, though: we also have the fundamental bidifferential $ \su{B}$. We propose a way of defining odd periods for meromorphic differentials as follows. First, we construct a canonical local one-form.

\begin{lemma}\label{lem:nu_p}
	Let $\Sigma$ be a compact super Riemann surface with Ramond divisor $\mc{R}$ with Torelli marking, and $\su{B}$ be the fundamental super bidifferential.
	Let $p$ be an oriented Ramond puncture, and choose local Ramond coordinates $ \su{z}$ at $ p$. Then
	\begin{equation}\label{oddoneform}
		d\eta_{p,0} (\su{z}) \coloneqq \< \su{B} (\su{z}_1, \su{z}), [d\theta_1] \>_{1,p} = \Res_{z_1 = 0} \su{B} (\su{z}_1, \su{z}) \theta_1 
	\end{equation}
	is a globally defined section of $ \mc{Q}'$, and is independent of the choice of $ \su{z}$.
\end{lemma}

The notation $ d\eta_{p,0}$ may look a little strange. We will come back to these sections in \cref{sec:SAStoTR}, where we introduce a sequence of related sections $ d\eta_{p,k}$ for $ k \in \Z$.

\begin{proof}
	The fundamental bidifferential is a globally defined bidifferential with only a double pole on the diagonal. Therefore, it is holomorphic at $ p$ and its Laurent expansion around $p$ has no negative powers in $ \su{z}_1$. Since the same holds for $ [d\theta_1]$ and since superconformal maps do not change the leading order behaviour, $ \nu_p$ only depends on $ d\theta \pmod{z}$, which is well-defined given the orientation.
\end{proof}

\begin{definition}\label{def:periods for meromorphic forms}
	Let $\Sigma$ be a compact super Riemann surface with Ramond divisor $\mc{R}$ with Torelli marking, and $\su{B}$ be the fundamental super bidifferential. Then, for a residueless meromorphic section $\omega$ of $\mc{Q}'$,  its odd periods $\alpha_{\mu}(\omega)$ and $\beta_{\mu}(\omega)$ are defined for all $\mu \in [r]$ as
	\begin{equation}
		\begin{split}
			\alpha_{\mu}(\omega)
			&\coloneqq
			\sqrt{\pi i} \Res_{z=p_\mu} \omega(\su{z})\int_*^\su{z} d\eta_{p_\mu,0} + i\sqrt{\pi i}\Res_{z=q_\mu} \omega(\su{z}) \int_*^\su{z} d\eta_{q_\mu,0},\\
			\beta_{\mu} (\omega)
			&\coloneqq
			\sqrt{\pi i} \Res_{z=p_\mu} \omega(\su{z})\int_*^\su{z} d\eta_{p_\mu,0} - i\sqrt{\pi i}\Res_{z=q_\mu} \omega(\su{z}) \int_*^\su{z} d\eta_{q_\mu,0} \,.
		\end{split}
	\end{equation}
\end{definition}
The above formulae do not depend on the base point $*\in\Sigma$, because $\omega$ is assumed to have no residues. This definition is consistent with \cref{def:super periods} when $\omega$ is holomorphic. One may write the definition in terms of the super Givental paring.

\begin{remark}\label{rem:local zero modes}
	\Cref{lem:nu_p} makes sense for any super Riemann surface with a $ \su{B}$ of the shape described in \cref{def:LSC}, and hence for any local super spectral curve, as defined there. With the same setup, we can define the odd period of a meromorphic section $ \omega$ of $ \mc{Q}'$ at any Ramond puncture $p$ as
	\begin{equation}
		w_p (\omega) 
		\coloneqq
		\sqrt{2\pi i} \Res_{z = p} \omega (\su{z}) \int_*^\su{z} d\eta_{p,0} \,.
	\end{equation}
\end{remark}

%

\section{Super Airy structures} \label{SuperAiry}

In this section, we will recall and generalise the notion of super Airy structures from \cite{BCHORS20}, as well as the relation to super topological recursion of \cite{BO21}. (Super) Airy structures, introduced in \cite{KS18}, are an algebraic encoding of loop equations, the equations that topological recursion solves. They encode the local behaviour of the multidifferentials $ \omega_{g,n}$ near the ramifications of the map $ x$ of the spectral curve.

\subsection{Super Airy structures}

We will consider a more general setup than \cite{BCHORS20}. This is because we need a framework properly capturing odd periods and odd parameters that behave fundamentally differently from even periods and parameters.

We will use grading conventions from \cite{BCJ22}. Let $ R = \C [u_1, \dotsc ,u_k, \upsilon_1, \dotsc, \upsilon_l ]$ be a freely generated supercommutative $\C$-algebra\footnote{The spectrum $\Spec R$ will be a local patch of the base space $T$ of our family of spectral curves.}, and consider a free super module over $R $, $ V = V_0 \oplus V_1 $, where $ V_0 $ is even and $ V_1$ is odd. We will allow for both $ V_0$ and $V_1$ to be countably infinite-dimensional. Let $ A_0 $ be linear coordinates on $ V_0 $ and $ A_1$ linear coordinates on $V_1$. We also write $ A \coloneqq A_0 \cup A_1 $.

\begin{definition}
	The \emph{super Weyl algebra} on $ V$ is the $R$-superalgebra on even generators $ \{ x, \del_x \}_{x \in A_0}$, odd generators $ \{ x, \del_{x} \}_{x \in A_1}$, and with non-trivial supercommutators
	\begin{equation}
		[\del_{x}, x] = 1\,.
	\end{equation}
	The \emph{completed super Weyl algebra} $ \mc{D}_V$ is the  extension of the super Weyl algebra which contains elements of the form
	\begin{equation}
		\sum_{m=0}^M \sum_{a_1, \dotsc, a_m \in A} p^{a_1,\dotsc,a_m} (A) \del_{a_1} \dotsb \del_{a_m}\,,
	\end{equation}
	where $ M$ is finite, the second sum may be infinite, and the $p^{a_1,...,a_m}$ are all polynomial. Furthermore, we introduce a filtration on this algebra by setting the degree of all generators to be $1$ -- this degree is different from cohomological degrees or super degrees. Then, the \emph{completed super Rees Weyl algebra} is the graded superalgebra
	\begin{equation}
		\widehat{\mc{D}}^\hslash_V \coloneq \prod_{n =0}^\infty \hslash^n F_n \mc{D}_V\,.
	\end{equation}
	where $F_n \mc{D}_V$  contains all operators in $\mc{D}_V$ of  degree $\leq n$.
\end{definition}

\begin{definition}[{\cite{BCHORS20,BCJ22,Bou24}}]\label{def:SAS}
	Choose  $ \theta^0 \in A_1$. A \emph{super Airy structure} on $(V, \theta^0)$ or a \emph{super Airy ideal} in $ \widehat{\mc{D}}^\hslash_V$ is a left ideal $ \mc{I}$ of $ \widehat{\mc{D}}^\hslash_V$ which has a collection of homogeneous (with respect to parity) elements 
	\begin{equation}
		H_j = \sum_{n = 0}^\infty \sum_{m = 0}^n \sum_{a_1, \dotsc, a_m \in A } \hslash^n p_{j,n}^{a_1, \dotsc, a_m}(A) \del_{a_1} \dotsb \del_{a_m}
	\end{equation}
	for $ j \in A \setminus \{ \theta^0 \}$ such that
	\begin{enumerate}
		\item The collection $ \{ H_j\}$ is \emph{bounded}: for each $ n \in \N $, $ a_1, \dotsc, a_m \in A$, only finitely many of the polynomials $ p_{j,n}^{a_1, \dotsc, a_m}$ are non-zero;
		\item The ideal is given by
			\begin{equation}
				\mc{I} = \big\{ \sum_{j \in A} c_j H_j \, \big| \, c_j \in \widehat{\mc{D}}^\hslash_V \big\} \,,
			\end{equation}
			where the sum may be infinite. These sums are well-defined by the previous condition;
		\item For $ j \in A \setminus \{ \theta^0 \}$, $ H_j = \hslash \del_{j} + \hslash p_{j,1} (A)+ \mc{O}(\hslash^2) $;
		\item $ [\mc{I}, \mc{I} ] \subseteq \hslash^2 \mc{I}$.
	\end{enumerate}
\end{definition}

\begin{remark}
	By definition of $ \widehat{\mc{D}}^\hslash_V$, the polynomials $ p_{j,n}^{a_1,\dotsc, a_m}$ are of degree at most $ n-m$. Usually, e.g., in \cite{BBCCN18,BCHORS20}, point (3) of the definition requires that $ H_j = \hslash \del_j + \mc{O}(\hslash^2)$. As explained in \cite[Remark 2.12]{Bou24}, we may allow linear polynomials at order $ \hslash^1$, and this turns out to be useful for us. By the boundedness condition, only finitely many of these can be non-zero.
\end{remark}

The main reason these structures are interesting, is the following theorem.

\begin{theorem}[{\cite[Theorem~2.10]{BCHORS20}}]\label{thm:BCHORS20}
	Let $V = V_0 \oplus V_1$ be a freely generated $R$-super module as above, $ \theta^0 \in A_1$, and $ \mc{I} $ a super Airy structure on it. There exists a function $Z $, called the \emph{partition function}, of the shape
	\begin{equation}
		Z = \exp \Big( \sum_{\substack{g \in \frac12\N, n \in \N^* \\ 2g-2+n \geq 0}} \hslash^{2g-2+n} F_{g,n} (A)\Big) \,,
	\end{equation}
	with $ F_{g,n} $ even polynomials, homogeneous of degree $n$, such that
	\begin{equation}
		\mc{I} Z = 0\,.
	\end{equation} 
	This function is unique up to the addition of terms $ \lambda_g \theta^0 $ to $ F_{g,1}$, with $ \lambda_g $ odd elements of $R$. If $\deg p_{j,1}=1$ for all $ j \in A\backslash\{\theta_0\}$, i.e. no constant terms, then $F_{\frac{1}{2},1} (A\backslash\{\theta_0\}) = 0$, and if $ p_{j,1} = 0$ for all $ j \in A\backslash\{\theta_0\}$, then $ F_{0,2}(A\backslash\{\theta_0\})=F_{\frac{1}{2},1} (A\backslash\{\theta_0\})  =0$. Furthermore, $F_{\frac{1}{2},1} (A) = 0$ and $\lambda_g=0$ for all $g\in\mathbb{Z}+\frac12$ imply that $F_{g,1} (A) = 0$ for all $g\in\mathbb{Z}+\frac12$.
\end{theorem}

Note that conditions (1) and (2) in \cref{def:SAS} are used for convergence of $F_{g,n}$, condition (3) is for uniqueness, and condition (4) is for the existence of such a solution. The statement and proof in \cite{BCHORS20} are given for the case $ R = \C$ and $ p_{j,1} = 0$ for all $ j \in A\backslash\{\theta_0\}$. However, the proof extends in this more general setting without any difficulty. In particular, the extension to non-zero $ p_{j,1}$ follows from \cite[Remark 2.12]{Bou24}, also partly discussed in \cite[Theorem 3.4]{Osu21}.

\begin{remark}
	In contrast to the purely bosonic case, there is no $H_0 $ with $\partial_{\theta^0}$ in the definition of super Airy structures. When working over $ R = \C$ (or any other even base ring), as in \cite{BCHORS20}, there is still a unique even partition function. In our setup, uniqueness breaks down due to possible addition of even terms $ \lambda_g \theta^0$ with odd $\lambda_g$.
\end{remark}

Because of the extra variable $ \theta^0$, the direct sum of several super Airy structures, which is necessary for spectral curves with several ramification points, is \emph{not} itself a super Airy structure, unlike the ordinary setting. Therefore, we add the following definition.

\begin{definition}\label{def:PSAS}
	Choose a finite subset $ \hat{A} \subset A$. A \emph{partial Airy structure} on $(V, \hat{A})$ or a \emph{partial Airy ideal} in $ \widehat{\mc{D}}^\hslash_V$ is a left ideal $ \mc{I}$ of $ \widehat{\mc{D}}^\hslash_V$ which has a collection of homogeneous (with respect to parity) elements 
	\begin{equation}
		H_j = \sum_{n = 0}^\infty \sum_{m = 0}^n \sum_{a_1, \dotsc, a_m \in A } \hslash^n p_{j,n}^{a_1, \dotsc, a_m}(A) \del_{a_1} \dotsb \del_{a_m}
	\end{equation}
	for $ j \in A \setminus \hat{A}$ such that
	\begin{enumerate}
		\item The collection $ \{ H_j\}$ is \emph{bounded}: for each $ n \in \N $, $ a_1, \dotsc, a_m \in A$, only finitely many of the polynomials $ p_{j,n}^{a_1, \dotsc, a_m}$ are non-zero;
		\item The ideal is given by
			\begin{equation}
				\mc{I} = \big\{ \sum_{j \in A} c_j H_j \, \big| \, c_j \in \widehat{\mc{D}}^\hslash_V \big\} \,,
			\end{equation}
			where the sum may be infinite. These sums are well-defined by the previous condition;
		\item For $ j \in A \setminus \hat{A}$, $ H_j = \hslash \del_{j} + \hslash p_{j,1} (A)+ \mc{O}(\hslash^2) $;
		\item $ [\mc{I}, \mc{I} ] \subseteq \hslash^2 \mc{I}$.
	\end{enumerate}
\end{definition}

\begin{proposition}\label{prop:partial Airy}
	Let $V = V_0 \oplus V_1$ be a freely generated $R$ super module as above, with linear coordinates $ A$. Let  $ \hat{A} \subset A$ be finite, and let $ \mc{I} $ be a partial Airy structure on $ (V, \hat{A})$. There exists a function $Z $, called a \emph{partition function}, of the shape
	\begin{equation}
		Z = \exp \Big( \sum_{\substack{g \in \frac12\N, n \in \N^* \\ 2g-2+n \geq 0}} \hslash^{2g-2+n} F_{g,n} (A)\Big) \,,
	\end{equation}
	with $ F_{g,n} $ even polynomials, homogeneous of  degree $n$, such that
	\begin{equation}
		\mc{I} Z = 0\,.
	\end{equation} 
	This function is unique up to addition of even polynomials in $ \hat{A} $ to the $ F_{g,n}$. If $\deg_\hslash p_{j,1}=1$ for all $ j \in A \setminus \hat{A}$, i.e. no constant terms, $F_{\frac{1}{2},1} (A \setminus \hat{A}) = 0$, and if $ p_{j,1} = 0$ for all $ j \in A \setminus \hat{A}$, then $ F_{0,2}(A \setminus \hat{A}) = F_{\frac{1}{2},1} (A \setminus \hat{A}) = 0$. Furthermore, if $F_{\frac12,1}(A)=0$ and $F_{g,n}(\hat A)=0$ for all $g\in\mathbb{Z}+\frac12$, then $F_{g,n}( A)=0$ for all $g\in\mathbb{Z}+\frac12$.
\end{proposition}

Before proving the above, let us mention that this proposition is mostly useful for $ \hat{A} \subset A_1$, as polynomials in odd variables are at most linear in each of the individual variables, and this severely limits the freedom. But for the proof, this condition is not necessary, so we omit it from the statement.

\begin{proof}
	Let us first remark that, similar to the ordinary case, existence of such a solution $Z$ is guaranteed by the left ideal condition, i.e. condition (4) in \cref{def:PSAS}, hence we only focus on uniqueness. The usual proof of uniqueness of partition functions for Airy structures builds $  \log Z$ inductively on $ \chi$, the $\hslash$-degree, using that the $ H_j$ commute up to lower-order terms and invoking the Poincar\'{e} lemma to solve the resulting linear ODEs $ \del_j (\log Z)^\chi = (\text{previously determined})$. The integration constant is fixed because $ F $ is not allowed to have a constant term in $ A$. In the case of partial Airy structures, only the $ \del_j (\log Z)^\chi$ for $ j \in A \setminus \hat{A}$ are specified, so the integration constant may still depend on $ \hat{A}$. 
	
	The statement about $p_{j,1}$ is straightforward to check by inspection.
\end{proof} 

The main point of \cref{prop:partial Airy} is that the sum of partial Airy structures is still a partial Airy structure. So we can get a partial Airy structure from a spectral curve with several ramification points.

\subsection{Super Airy structures for superconformal topological recursion}\label{sec:SAS for STR}

We will now construct the explicit super Airy structures that we will use to encode the loop equations near a single component of the ramification (i.e. the Ramond) divisor, and discuss how to extend it to a partial Airy structure for the entire divisor at the end. In this subsection, we closely follow notations of \cite[Section~4.2]{BCHORS20} and \cite[Section~4.2]{BO21}.

Let $ V_0$ and $V_1$ both be infinite-dimensional, with coordinates $ \{ x^k \, | \, k \in \N^* \} $ and $ \{ \theta^k \, | \, k \in \N \}$ respectively. For $B$ as a notation for Bosonic and $F$ for Fermionic, we introduce
\begin{equation}\label{IndexSet}
	A \coloneqq \{ B, F\} \times \N^* \cup \{ (F,0)\} \,, 
    \qquad 
    \hat{A} \coloneqq \{ (F,0)\} \,, 
    \qquad 
    x^{(B,k)} \coloneqq x^k \,, 
    \quad 
    x^{(F,k)} \coloneqq \theta^k \,,
\end{equation}
which will be consistent with the notation of \cref{prop:partial Airy}. We then define the \emph{Heisenberg} and \emph{Clifford operators} $ \{ J_k, \Gamma_k \, | \, k \in \Z \} $ as follows:

\begin{equation}\label{HC modes}
	J_k
	\coloneqq 
	\begin{cases} \hslash \del_{x^k} & k > 0 \\ 0 & k = 0 \\ -k \hslash x^{-k} & k < 0 \end{cases} \,,
	\qquad 
	\Gamma_k
	\coloneqq 
	\begin{cases} \hslash \del_{\theta^k} & k > 0 \\  \hslash (\frac{\theta^0}{2} + \del_{\theta^0} ) & k = 0\\ \hslash \theta^{-k} & k < 0\end{cases}\,.
\end{equation}
They satisfy
\begin{equation}\label{Heisenberg-Clifford}
	[J_k, J_\ell]
	=
	\hslash^2 \ell \delta_{k+\ell} \,,
	\qquad
	[\Gamma_k, \Gamma_\ell]
	=
	\hslash^2 \delta_{k+\ell} \,,
	\qquad
	[J_k, \Gamma_\ell]
	=
	0\,.
\end{equation}
The space generated by these elements is the Weyl algebra of the supersymplectic space described in \cref{SuperGiventalPairing}: it is the central extension of the vector space with basis $ J_k = \hslash k[z^{k-1} dz]$, $ \Gamma_k = \hslash [z^k d\theta ]$ with commutator $ [\hslash X, \hslash Y] = \hslash^2 \< X,Y\>$.

This space has a canonical `vacuum' representation on the space $ \C [ \{ x^k \}_{k>0}, \{ \theta^k\}_{k \geq 0}]$, with a vacuum $ | 0 \> \coloneqq 1$, a covacuum $ \< 0| = \mathop{\textup{ev}}_0 \colon P (\{ x^k, \theta^k\}) \mapsto P(0)$, and this induces a normal ordering, defined for quadratic expressions in the modes by\footnote{Note that $ \normord{Z} = Z - \< Z\> $ does not hold for all $Z$ in the Weyl algebra. E.g. for $ Z$ scalar, $ \normord{Z} = Z \neq Z - \< Z\> = 0$.}
\begin{equation}\label{DefNormOrd}
	\normord{X Y} = X Y -  \< X Y\> = X Y - \< 0 | X Y | 0 \>\,.
\end{equation}
More explicitly, this gives
\begin{equation}
	\normord{X_k Y_\ell} = X_k Y_\ell - [X_k,Y_\ell] \mb{1}_{k> 0} - \frac{1}{2} [X_k, Y_\ell ] \delta_{\ell,0}\,, \qquad  X, Y \in \{ J, \Gamma\} \,.
\end{equation}
The last term may be unfamiliar. It is necessary to deal with the zero-mode: we require that $ \normord{\Gamma_0^2} = 0$ to get supercommutativity inside the normal ordering. It also follows explicitly from the definition, \cref{DefNormOrd}.

Our interests are in constructing a super Airy structure from (a module of) the $\mc{N}=1$ super Virasoro algebra. More concretely, we define for $ n \in \Z$
\begin{align}
	L_{2n}
	&
	\coloneqq \frac{1}{2} \sum_{j \in \Z } (-1)^{j-1} \normord{J_{-j} J_{2n+j} } + 
	\frac{1}{2} \sum_{j \in \Z} (-1)^j \big( j + n \big) \normord{ \Gamma_{-j} \Gamma_{2n+j}} + \frac{\hslash^2}{4} \delta_{n,0}
	\\
	G_{2n+1}
	&
	\coloneqq \sum_{j \in \Z} (-1)^{j-1} J_{-j} \Gamma_{2n+1+j}\,.
\end{align}
By removing $(J_{2k}|\Gamma_{2k+1})$ and relabelling $(L_{2n}|G_{2n+1})\mapsto(L_{n}|G_{n+\frac12})$ for all $n\in\mb{Z}$, they precisely match with the $\rho$-twisted super Virasoro operators of \cite{BCHORS20}. Then, they satisfy the following commutation relations \cite[Equations~(4.12)-(4.16)]{BO21}
\begin{align}
	[L_{2n} , J_{2j}]
	&=
	2j \hslash^2  J_{2n+2j} 
	\\
	[L_{2n}, \Gamma_{2j-1}]
	&=
	(n+2j-1) \hslash^2 \Gamma_{2n + 2j-1} 
	\\
	[G_{2n+1}, J_{2j} ]
	&=
	2j \hslash^2 \Gamma_{2n + 2j +1}
	\\
	[G_{2n+1}, \Gamma_{2j-1} ]
	&=
	\hslash^2 J_{2n+2j}
	\\
	[L_{2n}, L_{2m}]
	&=
	2 \hslash^2 (n-m) \Big( L_{2n+2m} + \sum_{j\in \Z} \normord{ J_{-2j} J_{2n+2m+2j}} 
	\\
	& \hspace{3cm} + (n+m+2j+1) \normord{ \Gamma_{-2j-1} \Gamma_{2n+2m+2j+1}} \Big)
	\\
	[L_{2n}, G_{2m+1}]
	&=
	\hslash^2 (n-2m-1) \Big( G_{2n + 2m +1} + 2 \sum_{j\in \Z } \normord{J_{-2j} \Gamma_{2n+2m+2j+1}} \Big)
	\\
	[G_{2n+1}, G_{2m+1}]
	&=
	2\hslash^2 \Big( L_{2n+2m+2} + \sum_{j\in \Z} \normord{ J_{-2j} J_{2n+2m+2+2j}} + \normord{\Gamma_{-2j-1} \Gamma_{2n+2m+3+2j} } \Big)
\end{align}
The main use of this for us is in the construction of super Airy structures. Indeed, \cite{BCHORS20,BO21} showed that the set of operators $  \{ J_{2j}, \Gamma_{2j-1},  L_{2j-s-1}, G_{2j-s} \, | \, j \geq1 \}$ gives a super Airy structure after taking an appropriate conjugation.

More concretely, let $ \{ t_k \}_{k \in \Z } $, $  \{ B_{kl}, V_{kl} \}_{k,l \geq 0} $ in $ R^0 $ and $ \{ \tau_k \}_{k \in \Z }$ and $ \{ \phi_{kl} \}_{k,l \geq 0}$ in $ R^1$ be constants with $ t_k = \tau_k = 0$ for $ k \ll 0$. Without loss of generality\footnote{Because $ J_0 = 0 $, $ \Gamma_0^2 = \frac{\hslash^2}{2}$, and for the remaining cases $ J_k J_l = J_l J_k$ and $ \Gamma_k \Gamma_l = - \Gamma_l \Gamma_k$ for $k,l>0$. This is an equivalent, but slightly different convention from the one taken in \cite{BO21}.}, we assume that
\begin{equation}\label{conditions on polarisation}
	t_0 = 0 \,, 
    \qquad 
    B_{k0} = B_{0k} = \phi_{0k} = V_{00} = 0 \,, 
    \qquad 
    B_{kl} = B_{lk} \,, 
    \qquad 
    V_{kl} = - V_{lk} \,.
\end{equation}
Define the operators
\begin{equation}\label{dilaton and polarisation}
	T
    \coloneqq
    \exp \Big( \hslash^{-2} \sum_{k} J_k \frac{t_k}{k} - \Gamma_k\tau_k  \Big) \,, 
	\qquad
	\Phi 
    \coloneqq
    \exp \Big( \frac{1}{2\hslash^2} \sum_{k,l \geq 0}\frac{J_kJ_l}{kl}  B_{kl} - 2  \frac{J_k \Gamma_l}{k} \phi_{kl}+  \Gamma_k \Gamma_l V_{kl}\Big)\,,
\end{equation}
called respectively \emph{dilaton shift} and \emph{change of polarisation}. The signs in front of $\tau_k$ and $\phi_{kl}$ are inserted for notational convenience. Note that $t_k$ with $k<0$ were not allowed in  \cite{BCHORS20,BO21} (and  $\tau_k=\phi_{kl}=0$ there). We will show that this causes a significant difference for the construction of super topological recursion.
		
\begin{proposition}\label{BOSuperAiry}
	Let $T$ and $\Phi$ be the dilaton shift and polarisation operator as above, and define 
	\begin{equation}
		\tilde{L}_{2n} \coloneqq \Phi T L_{2n} T^{-1} \Phi^{-1} \,,
		\qquad
		\tilde{G}_{2n+1} \coloneqq \Phi T G_{2n+1} T^{-1} \Phi^{-1} \,.
	\end{equation}
	Set\footnote{Our $s$ should be the same as the $ \epsilon $ of \cite[Proposition 4.3]{BO21}, but that statement has a typo in the definition of $ \epsilon$ --- the definition of $\epsilon$ in the proof is correct, but not  the one in the proposition itself.}
	\begin{equation}\label{Defs}
		s \coloneq \min \{ k \, | \, k \textup{ odd and } t_k \textup{ invertible}\}
	\end{equation}
	and assume that
    \begin{equation}\label{Conditions}
        s \in \{ 1, 3 \} \,.
    \end{equation}
	Then the collection
	\begin{equation}
		\mathcal{I}_{\textup{STR}}=\big\{ J_{2j}, \Gamma_{2j-1}, \tilde{L}_{2j -s -1}, \tilde{G}_{2j-s} \, \big| \, j \geq1\big\}
	\end{equation}
	generates a super Airy structure, if the residueless meromorphic section $ \omega_{0,1} $ of $\mc{Q}'$ on the formal disc $ \Spec R \llbracket z, \theta \rrbracket \to \Spec R$ defined by
	\begin{equation}\label{DefOmega}
		\omega_{0,1} (\su{z})
		\coloneqq
		\sum_{k}  [dz]  t_k z^{k-1}+ [d\theta]  \tau_k z^k
	\end{equation}
	satisfies that for $\sigma:\su{z}\mapsto-\su{z}$,
	\begin{equation}\label{LOLE}
		\omega_{0,1} (\su{z}) + \omega_{0,1} (\sigma(\su{z}))
		=
		[d\theta]\mc{O}(z)
		\qquad 
		\text{and}
		\qquad
		\mc{C}^{\sigma}(\omega_{0,1} \boxtimes \omega_{0,1} )(\su{z}) 
		=
		[d\theta]^3\mc{O}(z^{s}).
	\end{equation}

\end{proposition}

\begin{proof}
	We need to check the four conditions of \cref{def:SAS}. First, we consider the $\hslash$-independent terms and $\hslash$-linear terms of all elements of the generating collection.
	
	The condition that $ \omega_{0,1} (\su{z}) + \omega_{0,1} (-\su{z})$ be holomorphic ensures that
	\begin{equation}
		J_{2j} 
		= 
		\Phi T J_{2j} T^{-1} \Phi^{-1} \,, 
		\qquad
		\Gamma_{2j-1} 
		=
		\Phi T \Gamma_{2j-1} T^{-1} \Phi^{-1}
	\end{equation}
	for $ j \geq1$.
	As for the $ L_{2j-s-1}$ and $ G_{2j-s}$, they are homogeneous of $\hslash$-degree $2$ and $ \Phi$ preserves this. Every adjunction by $ \log T$ decreases this by one. The first adjunction by $\log T$ (without $\Phi$) gives
	\begin{align}
		\ad_{\log T} L_{2n} 
		&= 
		\sum_{k \in \Z} (-1)^{k-1}J_{2n+k} t_k  - (-1)^k (k+n) \Gamma_{2n+k}\tau_k  \,;
		\\
		\ad_{\log T} G_{2n+1}
		&=
		\sum_{k \in \Z} (-1)^{k-1}\Gamma_{2n+1+k} t_k  + (-1)^{k} J_{2n+1+k} \tau_k  \,,
	\end{align}	
	which are homogeneous of $\hslash$-degree one. Further application of $ \log T$ gives
	\begin{align}
		\ad^2_{\log T} L_{2n} 
		&= 
		\ad_{\log T} \Big( \sum_{k \in \Z} (-1)^{k-1} J_{2n+k} t_k -  (-1)^k (k+n) \Gamma_{2n+k} \tau_k \Big)
		\\
		&=
		\sum_{k \in \Z} (-1)^{k-1} t_k t_{-k-2n} + (-1)^k (k+n) \tau_k \tau_{-k-2n} \,;
		\\
		\ad^2_{\log T} G_{2n+1}
		&=
		\ad_{\log T} \Big( \sum_{k \in \Z} (-1)^{k-1} \Gamma_{2n+1+k} t_k + (-1)^{k} J_{2n+1+k}\tau_k \Big)
		\\
		&=
		\sum_{k \in \Z} (-1)^{k-1} \tau_{-k-2n-1}t_k  + (-1)^k  t_{-k-2n-1}\tau_k
		\\
		&=
		2 \sum_{k \in \Z} (-1)^{k-1} t_k \tau_{-k-2n-1} \,,
	\end{align}
	which are of $\hslash$-degree zero. These sums are finite because the $ t_k$ and $ \tau_k$ are zero for $ k \ll 0$. The second condition of \cref{LOLE} ensures that all the $\hslash$-independent terms vanish.

	For condition (1), we note that it holds for the $J$, $\Gamma$, $L$, and $G$. Moreover, by interpreting $T$ and $\Phi$ as a linear map from $\{J,\Gamma,L,G\}$ to $\{\tilde J,\tilde \Gamma,\tilde L,\tilde G\}$,  $ \log T$ and $ \log \Phi$ have strictly upper triangular bosonic reduction by the condition on $s$, and only finitely many non-trivial diagonals below the main diagonal by our condition on $ t_k $ and $ \tau_k$, with only nilpotent coefficients. This shows\footnote{Using that $R$ has only finitely many odd generators.} that also $T$ and $\Phi$ have only finitely many non-trivial diagonals below the main diagonal, and therefore conjugation by them preserves boundedness. Condition (2) holds by definition.

    From this computation, and using that $ s > 0$, one sees that there exists a linear map, whose bosonic reduction is upper triangular with non-zero diagonal, and which is therefore invertible. (Cf. \cite[Appendix A.2]{BO21} for reduced cases.) Such an invertible matrix transforms the collection given in the proposition to one with $\hslash$-linear term as in (3) of \cref{def:SAS}. Finally, condition (4) follows directly from the commutation relations of the $ J$, $ \Gamma$, $ L$, and $ G$, since these are preserved by conjugation.
\end{proof}

\begin{remark}
	Because we allow for $ t_k$ and $ \tau_k$ with $ k $ negative, $T$ and $\Phi$ need not commute. We could have also chosen to conjugate first by $ \Phi $ and then by some 
	\begin{equation}
		\tilde{T} \coloneqq \exp \Big( \hslash^{-2} \sum_{k} \frac{J_k}{k}\tilde{t}_k  - \Gamma_k\tilde{\tau}_k \Big)
	\end{equation}
	Then we would have $ \tilde{T} = \Phi T \Phi^{-1}$, and
	\begin{equation}
		\omega_{0,1} (\su{z})
		=
		\sum_{k}d\xi_k ( \su{z}) \tilde{t}_k  + d\eta_k (\su{z})\tilde{\tau}_k 
	\end{equation}
	with notation introduced later, in \cref{sec:SAStoTR}, for the \emph{same} $ \omega_{0,1}$ as in \cref{DefOmega}, just expanded in a different basis. In fact, the polarisation $\Phi$ is closely related to a choice of a basis of the space of meromorphic section of $\mc{Q}'$. This perspective will become clearer in \cref{sec:SAStoTR}. 
	\end{remark}

\subsection{Partial super Airy structures}\label{sec:partial SAS}
We can naturally extend \cref{BOSuperAiry} into the case of a direct sum of super Airy structures. This will not be a super Airy structure any more, but it will still be a partial Airy structure. Since one only needs to repeat similar arguments and computations, let us simply outline changes:
\begin{itemize}
	\item The index set $A$ in \eqref{IndexSet} is now taken as $A_r\coloneqq [r]\times \big( \{ B, F\} \times \N^* \cup \{ (F,0)\} \big)$, and $\hat{A}_r\coloneqq[r]\times\{(F,0)\}$. Their coordinates now carry an extra index $a\in[r]$ i.e., $x^{a,k},\theta^{a,k}$;
	\item The operators $J_{a,k},\Gamma_{a,k}$ are defined in terms of $x^{a,k},\theta^{a,k}$ in an analogous manner, and the commutation relations \eqref{Heisenberg-Clifford} are equipped with $\delta_{a,b}$ for $a,b\in[r]$. That is, operators from different copies commute;
	\item Similarly, define $L_{a,2n},G_{a,2n+1}$ for every $a\in[r]$, and as a consequence, $\delta_{a,b}$ appears in the commutation relations among $J_{a,k},\Gamma_{a,k},L_{a,2n},G_{a,2n+1}$;
	\item For ${a,b\in[r]}$ and ${k,l\geq0}$, define $B^{ab}_{kl}, \phi^{ab}_{kl}, V^{ab}_{kl}$. And for $ a \in [r]$ and $k \in\mathbb{Z}$, define $(t_{a,k},\tau_{a,k})$ such that $t_{a,k}=\tau_{a,k}=0$ for $k\ll 0$. Similar to \eqref{conditions on polarisation}, we assume that
	\begin{equation}\label{constraints for multi components}
		t_{a,0} = 0 \,, \qquad B_{0l}^{ab} = B_{k0}^{ab} = \phi_{0l}^{ab} = V_{00}^{ab} = 0 \,\quad B_{kl}^{ab} = B_{lk}^{ba} \,, \qquad V_{kl}^{ab} = - V_{lk}^{ba} \,;
	\end{equation}
	\item Define $s_a$ for each $ a \in [r]$ similarly to \ref{Defs}, and assume $s_a\in \{ 1, 3 \}$. Then define the unique dilaton shift operator $T$ and the polarisation operator $\Phi$ in an analogous way to \eqref{dilaton and polarisation}, by also summing over $ [r]$. This gives $\tilde L_{a,2n},\tilde G_{a,2n+1}$ by conjugation;
	\item To extend the condition \eqref{LOLE} in \cref{BOSuperAiry}, let $\Sigma$ be a super Riemann surface with a degree $r$ Ramond divisor with  elements $p_a\in\textup{supp}(\mc{R})$ labelled by $a \in [r]$ , $\su{z}_a$ be local coordinates and $\sigma_a$ the local involution operator at each Ramond puncture $p_a$;
	\item The monomial basis $\{d\xi_{a,k},d\eta_{a,k}\}$ now carries the extra index $a\in[r]$. The Givental pairing is defined at every ramification point, and the pairing between two different ramification points are set to zero.
	\item $\omega_{0,1}$ now is a section of $\mc{Q}' \to \Sigma$ with the following local expansion at each  $p_a$:
	\begin{equation}
		\omega_{0,1}(\su{z})\underset{\text{near }p_a}{=}\sum_k[dz_a]t_{a,k}z_a^{k-1}+[d\theta_a]\tau_{a,k}z^k_a \,.
	\end{equation} 
\end{itemize}

We can then show the following theorem, whose proof is  parallel to \cref{BOSuperAiry}.

\begin{theorem}\label{thm:SASmain}
	The collection
	\begin{equation}
		\mathcal{I}_{\textup{STR}}^{\,r}=\big\{ J_{a,2j}, \Gamma_{a,2j-1}, \tilde{L}_{a,2j -s_a -1}, \tilde{G}_{a,2j-s_a} \, \big| \,a\in[r],\; j \geq1\big\}
	\end{equation}
	generates a partial Airy structure, if the residueless meromorphic section $ \omega $ of $\mc{Q}'$ satisfies at each neighbourhood of $p_a\in \textup{supp}(\mc{R})$ for $a\in[r]$,
	\begin{equation}
		\omega_{0,1} (\su{z}_a) + \omega_{0,1} (\sigma_a(\su{z}_a))=[d\theta_a]\mc{O}(z_a) \qquad \text{and} \qquad \mc{C}^{\sigma_a}(\omega_{0,1}\boxtimes\omega_{0,1})(\su{z}_a) =[d\theta_a]^3\mc{O}(z_a^{s_a}).
	\end{equation}

\end{theorem}

In the next section, we will construct a geometric counterpart of the partial Airy structure $\mathcal{I}_{\textup{STR}}^{\,r}$, i.e., super topological recursion on a disjoint union of super disks. Note that no generator of $\mathcal{I}_{\textup{STR}}^{\,r}$ contains constant terms with respect to $\hslash$. Thus, by setting $F_{g,n}(\hat{A} )=0$ for all $g\in\mathbb{Z}+\frac12$, we have $F_{g,n}=0$ for all $g\in\mathbb{Z}+\frac12$ as shown in \cref{prop:partial Airy} . From now on, we will work on this setting, and in particular, $g$ in the rest of the article is always an integer.

\begin{remark}\label{remark:4.1.1}
When we consider super topological recursion on a compact super Riemann surface of genus $\su{g} = (\tilde{g}|r)$, we will relate odd holomorphic differentials $d\eta_{a,0}$ as defined in equation \ref{oddoneform}
to algebraic zero modes $\theta^{a,0}$. Since $\{d\eta_{a,0}\}_{a\in[2r]}$ spans an $r$-dimensional space and hence the forms are mutually dependent, one has to add $r$ more constraints to the partial Airy structure in order to capture such relations. We will discuss this point in detail in \cref{sec:global constraints}.
\end{remark}

\section{Super Loop Equations}\label{sec:super loop equations}

In this section, we reinterpret Airy structures and the partition function in terms of super Riemann surfaces. We will construct a geometric counterpart $\omega_{g,n}$ which encodes the same information as $F_{g,n}$ of a partial Airy structures, and explore geometric properties equivalent to \cref{thm:SASmain}.

\subsection{Super spectral curves}\label{sec:super spectral curves}
In the non-super setting, a spectral curve is often given as a collection $ ( \Sigma, x, y, B)$ of a Riemann surface $ \Sigma$, two meromorphic functions $ x, y$ on $ \Sigma$, and a meromorphic symmetric bidifferential $B$ with double pole on the diagonal with biresidue $1$ -- a fundamental bidifferential. For the original topological recursion, we require $ x$ to only have simple ramification (we call the ramification locus $R$), and $y$ to be non-ramified at these ramifications of $x$.

More properly (for our setup), we may think of a simply-ramified map $ x \colon \Sigma \to \P^1$, together with a one-form $ \omega_{0,1} = y\, dx$ on $ \Sigma$, and still $B$. These data define an Airy structure  -- the local expansion of $ \omega_{0,1}$ defines the dilaton shift $T$, and the expansion of $ B$ defines the change of polarisation $\Phi$. Kontsevich--Soibelman \cite{KS18} showed that this Airy structure is equivalent to loop equations for $\omega_{g,n}$: near a ramification point $a$ of $x$, let $\sigma$ be the local deck transformation. Then

\begin{align}
	\mc{L}_{g,n} (z; z_{\llbracket n \rrbracket}) 
	&:= 
	\omega_{g,n+1} (z, z_{\llbracket n \rrbracket} ) + \omega_{g,n+1} (\sigma (z), z_{\llbracket n \rrbracket} ) \label{LLE}
	\\
	\mc{Q}_{g,n} (z; z_{\llbracket n \rrbracket})
	&:=
	\omega_{g-1,n+2} (z ,\sigma (z), z_{\llbracket n \rrbracket}) + \sum_{\substack{g_1 + g_2 = g \\ I \sqcup J = \llbracket n \rrbracket}} \omega_{g_1, |I| +1} (z, z_I) \omega_{g_2, |J| +1} (\sigma (z), z_J)\label{QLE}
\end{align}
are both holomorphic as $ z \to a$, with respectively a single and double zero when $s=3$ (the latter is of order $\mc{O}(1)$ when $s=1$).

The Airy structure then determines the principal parts of all $ \omega_{g,n}$ uniquely. To be able to reconstruct the $ \omega_{g,n}$ completely, we still need to specify the purely holomorphic part. This is where $B$ comes in. It fixes the purely holomorphic part via the \emph{projection property}:

\begin{equation}
	\omega_{g,n+1} (z, z_{\llbracket n \rrbracket}) = \sum_{a \in R} \Res_{z' = a} \Big( \int_a^{z'} B (\cdot, z) \Big) \omega_{g,n+1} (z', z_{\llbracket n \rrbracket}) \,.
\end{equation}
In case $ \Sigma $ is compact, this $B$ is the one from \cref{sec:the_symmetric_bidifferential}. Therefore, the projection
 \begin{equation}
	\omega (z) \mapsto \sum_{a \in R} \Res_{z' = a} \Big( \int_a^{z'} B (\cdot, z) \Big) \omega (z') 
\end{equation}
preserves the principal parts of $ \omega$ at the ramification points (which from the Airy structure construction are the only principal parts of the $ \omega_{g,n}$ for $ (g,n) \notin \{ (0,1), (0,2)\}$), and kills the $A$-periods.

Moreover, if $ \omega $ is only given as germs at the ramification points, the projection constructs the unique global form on $ \Sigma$ with the same principal parts and zero $A$-periods. Therefore, the spectral curve data suffice to construct all $ \omega_{g,n}$ globally and uniquely using the loop equations/Airy structures for the principal parts and the projection property for the holomorphic part.

Now, let us give the super analogue to this idea. We define a simply-ramified super spectral curve as follows.

\begin{definition}\label{def:GSC}
	A \emph{compact super spectral curve} is a tuple $ \su{S} = ( \Sigma, \mc{R}, \mc{T}, x, \omega_{0,1} )$, where
	\begin{enumerate}
		\item $ \Sigma $ is a compact super Riemann surface with Ramond divisor $ \mc{R}$;
		\item $ \mc{T} $ is a super Torelli marking on $ \Sigma$, i.e. a symplectic basis $ \{ A_j, B_j \}_{j = 1}^{\tilde{g}} $ of $ H_1 (\Sigma_{\textup{red}} ; \Z )$ along with an ordering $ \mc{R} = \{ p_\mu, q_\mu \}_{\mu = 1}^r $ and a choice of the sign of $ d\vartheta $ at each $ p_\mu $ and $ q_\mu$;
		\item $x \colon \Sigma \to \P^{1|1} $ is a superconformal cover with only simple ramifications;
		\item $ \omega_{0,1} $ is a meromorphic section of  $ \mc{Q}'$ satisfying \cref{Conditions,LOLE} at every ramification point.
	\end{enumerate}
\end{definition}

The \emph{genus} of $ \su{S}$ is $ \su{g}(\su{S}) = \su{g} (\Sigma) = ( \tilde{g} \, | \, r ) $. The definition implies that $ \mc{R}$ is the ramification divisor of $ x$. Given $\su{S}$, one can uniquely determine the \emph{super fundamental bidifferential} $ \su{B} $ of $ \su{S}$ as in \cref{sec:the_symmetric_bidifferential}. An analogue of the condition on the ramification locus of $y$ is taken care of by the condition \eqref{Conditions} -- in the ordinary setting, these are exactly the same conditions.

If we do not require compactness, we should define a bidifferential explicitly instead of Torelli markings as there is no unique choice.

\begin{definition}\label{def:LSC}
	A \emph{local super spectral curve} is a tuple $ \su{S} = ( \Sigma, \mc{R}, x, \omega_{0,1}, \su{B})$, where
	\begin{enumerate}
		\item $ \Sigma $ is a disjoint union of $r$ formal super disks $ \Spec R \llbracket z, \theta \rrbracket$ such that the union of the centres $ \Spec R [\theta]$ of all disks forms a Ramond divisor  $\mc{R}$ of degree $r$;
		\item $x \colon \Sigma \to \P^{1|1} $ is a nowhere-constant superconformal map with only simple ramifications;
		\item $ \omega_{0,1} $ is a meromorphic section of  $ \mc{Q}'$ satisfying \cref{Conditions,LOLE} at ramification points;
		\item $\su{B} \in H^0 (\Sigma^2; \mc{Q}'^{\boxtimes 2}(2\Delta))$ is symmetric with a local behaviour near the diagonal
			\begin{align}
				\su{B}(\su{z}_1,\su{z}_2) 
				&\sim_{z_1 \to z_2  \textup { NS}} [d\theta_1] [d\theta_2]\Big( - \frac{\theta_1 \theta_2}{(z_1-z_2)^2 } + \frac{1}{z_1-z_2} + \mc{O}(z_1-z_2)^0 \Big) \,,
				\\
				\su{B}(\su{z}_1,\su{z}_2) 
				&\sim_{z_1 \to z_2 \textup { R}} [d\theta_1] [d\theta_2]\Big( - \frac{\theta_1 \theta_2 z_1z_2}{(z_1-z_2)^2 } - \frac{z_1+z_2}{2(z_1-z_2)} + \mc{O}(z_1-z_2)^0 \Big).
			\end{align}
	\end{enumerate}
\end{definition}

Because of the last part of this definition, we introduce the notation

\begin{align}
	\su{B}^{\textup{NS}}(\su{z}_1 , \su{z}_2) 
	&\coloneqq
	[d\theta_1] [d\theta_2]\Big( - \frac{\theta_1 \theta_2}{(z_1-z_2)^2 } + \frac{1}{z_1-z_2} \Big)\,,
	\\
	\su{B}^{\textup{R}} (\su{z}_1, \su{z}_2) 
	&\coloneqq
	[d\theta_1] [d\theta_2]
	\Big( - \frac{\theta_1 \theta_2 z_1 z_2}{(z_1-z_2)^2 } - \frac{z_1+z_2}{2(z_1-z_2)} \Big) \,.
\end{align}
for the local behaviour in these two settings. We note that $\su{B}^{\textup{NS}}$ and $\su{B}^{\textup{R}}$ encode the local behaviour of the sum of the Bergman kernel $B$ and the Szeg\H{o} kernel $S$. Thus, a compact super spectral curve can be realised as a local super spectral curve, because the fundamental bidifferential $ \su{B}$ admits such an expansion by construction. Furthermore, the setting of local super spectral curves can also describe non-compact super Riemann surfaces.

\begin{remark}
	We restrict to cases with simple ramifications, because the non-simply ramified case poses significant extra difficulties. Firstly, this case is already a lot more complicated in the non-super context, both in writing down the formulas and in proving that they are actually well-posed. Secondly, the super context provides two more complications: whether a component of the ramification locus is Neveu--Schwarz or Ramond depends on the parity of the ramification index by \cref{prop:SCtrans}, and a super Riemann surface which is not simply-ramified over $ \P^{1|1}$ may lie in the bad locus of \cref{BadLocus} because \cref{SimplyRamifiedNotBad} does not hold. Finally, the expected algebra for higher-ramified cases, often called the parafermion algebra, has a complicated structure. As a consequence, it requires further studies to show e.g. whether some modules give rise to partial Airy structures as a generalisation of \cref{thm:SASmain}. 
\end{remark}

\subsection{From super Airy structures to super loop equations} \label{sec:SAStoTR}

In this section, we will derive super loop equations from the super Airy structures constructed in \cref{thm:SASmain} and the projection property associated to the super fundamental bidifferential. This was already done in \cite[Section~4.2]{BO21} in the reduced setting, but we will give an exposition closer to \cite[Sections~4~\&~5]{BKS23} for a non-reduced setting, i.e. with odd parameters.

We consider a local super spectral curve as in \cref{def:LSC}. Near each ramification, using \cref{prop:NS-R SCmap} we choose a local coordinate $ \su{z} $ such that $ \su{x}( \su{z} ) =(x(\su{z}) \, | \, \phi(\su{z}))= ( z^2 \, | \, \sqrt{2} z\theta )$ -- this implies that $ \{ z = 0 \}$ is a Ramond puncture by \cref{prop:SCtrans}. We will build a super Airy structure at each of these punctures. For ease of notation, we now fix one such puncture and ignore contributions from other ramifications, and we do not indicate it in notation. The space of local meromorphic sections (on a formal disc) of $ \mc{Q}'$ is then isomorphic to $ Q'$ from \cref{SuperGiventalPairing}, and hence has a super symplectic pairing.

We define the constants $ \{ t_k , \tau_k \}_{k \in \Z}$ by the expansion

\begin{equation}
	\omega_{0,1}(\su{z})
	\eqqcolon
	\sum_{k}[dz] t_k z^{k-1}+[d\theta]  \tau_k z^k\,.
\end{equation}
Similarly, we define the constants $ \{ B_{kl}, V_{kl}, \phi_{kl} \}_{k,l \geq 0}$ by the expansion of $\su{B}$ as
\begin{equation} \label{DefSuperB}
		\begin{split}
		\su{B} ( \su{z}_1 , \su{z}_2 ) 
		&=
		\su{B}^{\textup{R}} (\su{z}_1, \su{z}_2) +
		\sum_{k,l > 0}^\infty[dz_1] [dz_2] B_{lk} z_1^{k-1} z_2^{l-1} 
		\\
		&\quad
		+ \sum_{k> 0, l \geq 0}[dz_1] [d\theta_2 ] \left( 1- \frac{\delta_{0l}}{2}\right) \phi_{kl} z_1^{k-1} z_2^l 
		+ 
		\sum_{k \geq 0,l>0} [d\theta_1] [dz_2] \left(1- \frac{\delta_{k0}}{2}\right) \phi_{lk}z_1^k z_2^{l-1}
		\\
		&\quad
		+ \sum_{k,l \geq 0} [d\theta_1] [d\theta_2] \left(1-\frac{\delta_{k,0}}{2}-\frac{\delta_{l,0}}{2}\right) V_{kl}z_1^k z_2^l  \,,
	\end{split}
\end{equation}
The symmetry of $\su{B}$ implies that the constants satisfy the conditions \eqref{conditions on polarisation}. This implies that the dilaton shift $T$ and the polarisation $\Phi$ determines $\omega_{0,1}$ and $\su{B}$ and vice versa. This is why we use the same letters $(t_k,\tau_k,B_{kl},V_{kl},\psi_{kl})$ here.

The space $Q'$ has a natural maximally isotropic subspace, namely the space of holomorphic sections $Q^+$. The fundamental super bidifferential determines a nearly complementary maximally isotropic subspace $Q^-$, together with a one-dimensional rest $Q^0$, together called a \emph{polarisation}, as follows.

In the domain $ |z_1| < |z_2|$, we can rewrite $ \su{B}$ in the following way:
\begin{equation}
\su{B} ( \su{z}_1 , \su{z}_2 ) = 
		\sum_{k>0}^\infty k  d\xi_k (\su{z}_1) d\xi_{-k}(\su{z}_2) + \frac{1}{2} d\eta_0 (\su{z}_1) d\eta_0 (\su{z}_2) +
		\sum_{k>0}^\infty d\eta_k (\su{z}_1) d\eta_{-k}(\su{z}_2)\label{polarised expansion}
\end{equation}
where for $ k > 0$, we define, recalling the assumption \eqref{conditions on polarisation},
\begin{align}
	d\xi_k (\su{z})
	&=
	[dz]z^{k-1} 
	\\
	d\eta_k (\su{z})
	&=
	[d\theta]z^k 
	\\
	d\xi_{-k} (\su{z})
	&=
	\frac{[dz]}{z^{k+1}} + \sum_{l \geq 0} \Big( [dz]\frac{2 B_{kl} + \phi_{l0}\phi_{k0}}{2k} z^{l-1} + [d\theta] \frac{2\phi_{kl} + V_{0l}\phi_{k0}}{2k} z^l \Big)
	\\
	d\eta_0 (\su{z})
	&=
	[d\theta ] + \sum_{l>0} \Big( [dz]\phi_{l0} z^{l-1} +[d\theta] V_{0l} z^l \Big)
	\\
	d\eta_{-k} (\su{z})
	&=
	\frac{[d\theta]}{z^k} + \sum_{l \geq 0} \Big( [dz]\frac{2\phi_{lk} + V_{k0}\phi_{l0}}{2} z^{l-1} +[d\theta] \frac{2V_{kl} + V_{k0}V_{0l}}{2} z^l \Big) \,.
\end{align}

The sets $\{d\xi_k\}_{k\in\mb{Z}^*}$ and $\{d\eta_k\}_{k\in\mb{Z}}$ form a super symplectic basis of $Q'$. The subset $ \{ d\xi_k, d\eta_k \}_{k > 0}$ span $Q^+$, and we define $Q^- = \mathop{\textup{Span}} \{ d\xi_k,d\eta_k\}_{k < 0}$ and $ Q^0 = \mathop{\textup{Span}} \{ d\eta_0 \}$ to get a polarisation. Indeed, one can show by using the conditions \eqref{conditions on polarisation} that these forms satisfy
\begin{equation}
	\< d\xi_k, d\xi_l \> = \frac{\delta_{k,-l}}{k} \,, \quad \< d\xi_k, d\eta_l \> = 0 \,, \quad \< d\eta_k, d\eta_l \> = \delta_{k,-l}
\end{equation}
and therefore the projection operator
\begin{equation}
	\mc{P} \colon Q' \to Q' \colon \omega (\su{z}) \mapsto \< \su{B} (\su{z}, \su{z}'), \omega (\su{z}') \>_{\su{z}'}
\end{equation}
has eigenspaces $ Q^+$, $Q^0$, and $Q^-$, with respective eigenvalues $ 0$, $ \frac{1}{2}$, and $ 1$.

\begin{remark}\label{rem:regularised projection}
	Unlike the case without any polarisation discussed in \cref{sec:projection}, for a residueless local meromorphic section $ \omega (\su{z}) = [dz]\sum_{k\in \Z}  u_k z^{k-1} +[d\theta]\sum_{k\in \Z}  \upsilon_k z^k$ of $Q'$, we find
	\begin{equation}
		\< \su{B} (\su{z}_1,\su{z}_2), \omega (\su{z}_2) \>
		=
		\frac{1}{2} d\eta_0 (\su{z}_1) \left( \upsilon_0 + \sum_{l<0} \frac{\phi_{-l,0}u_{l-1}}{l-1}+V_{0,-l} \upsilon_{l} \right)  +\sum_{k < 0} d\xi_k ( \su{z}_1)  u_k+  d\eta_k (\su{z}_1 ) \upsilon_k\,.
		\label{projection with zero mode}
	\end{equation}
	The first term comes with extra coefficients because in the expansion \eqref{polarised expansion}, $d\eta_0(\su{z}_2)$ contains positive powers of $z_2$ which couples with negative power terms of $\omega(\su{z}_2)$. This motivates us to consider the following regularised projection to $Q^-$:
  	\begin{equation}
  		\begin{split}
    		\su{B}^{\textup{reg}}(\su{z}_1,\su{z}_2)
    		&\coloneqq 
    		\su{B} (\su{z}_1,\su{z}_2)-\frac12d\eta_0(\su{z}_1)d\eta_0(\su{z}_2),
    		\\
  			\<\su{B}^{\textup{reg}}(\su{z}_1,\su{z}_2), \omega (\su{z}_2) \>
			&=
			\sum_{k < 0} d\xi_k ( \su{z}_1)  u_k+  d\eta_k (\su{z}_1 ) \upsilon_k\,
		\end{split}
		\label{projection without zero mode}
	\end{equation}
	We will encounter another role of $ \su{B}^{\textup{reg}}$ in \cref{sec:regularised projection}.
\end{remark}

With these definitions and \cref{BOSuperAiry}, we have a super Airy structure, and therefore a partition function. We will translate this partition function back into principal parts of symmetric multidifferentials $ \omega_{g,n}$. To do this, we first need to translate the Airy structure into more geometric terms. We define differentials $J$ and $\mc{J}$ as
\begin{equation}
	J (\su{z} )\coloneqq \sum_{k \in \Z} [dz] \frac{J_k}{z^{k+1}} + [d\theta] \frac{\Gamma_k}{z^k} = \sum_{k \in \Z} \frac{[d\theta]}{z^k} \Big( \theta J_k  + \Gamma_k \Big),\quad \mc{J} (\su{z}) \coloneqq \sum_{k \in \Z} d\xi_{-k} (\su{z}) J_k + d\eta_{-k} (\su{z})\Gamma_k.
\end{equation}
Note that $J$ corresponds to a pair of free boson and free fermion (in the context of the $\mc{N}=1$ super Virasoro vertex operator algebra) in the monomial basis. On the other hand, one can show that $\mc{J}$ is related to $J$ by conjugation by the polarisation as below:
 \begin{lemma}\label{lem:DilatonPolarisationCurrent}
	The dilaton shift and the change of polarisation act on the differential $J$ by
	\begin{align}
		T J (\su{z}) T^{-1} 
		&= 
		J (\su{z}) + \omega_{0,1} (\su{z}) \,,
		\\
		\Phi T J (\su{z}) T^{-1} \Phi^{-1}
		&=
		\mc{J}(\su{z}) + \omega_{0,1} (\su{z}) \,.
	\end{align}
\end{lemma}
\begin{proof}
	The conjugation by $ T$ shifts $ J_{-k} $ by $ t_k$ and $ \Gamma_{-k}$ by $ \tau_k$, and these shifts combine exactly into $ \omega_{0,1} (\su{z})$.\par
	For the conjugation by $ \Phi$, we see that 
	\begin{align}
		\ad_{\log \Phi} J_{-k} 
		&= 
		\mb{1}_{k>0} \sum_{l \geq 0} \frac{B_{kl}}{l} J_l + \phi_{kl} \Gamma_l \,,
		&
		\ad_{\log \Phi} \Gamma_{-k}
		&=
		\mb{1}_{k\geq 0} \sum_{l \geq 0} \frac{\phi_{lk}}{l} J_l + V_{lk} \Gamma_l \,,
		\\
		\ad_{\log \Phi}^2 J_{-k}
		&=
		\mb{1}_{k>0} \phi_{k0} \sum_{l > 0 } \frac{\phi_{l0}}{l} J_l + V_{l0} \Gamma_l \,,
		&
		\ad_{\log \Phi}^2 \Gamma_{-k}
		&=
		\mb{1}_{k>0} V_{0k} \sum_{l > 0} \frac{\phi_{l0}}{l} J_l + V_{l0} \Gamma_l \,,
		\\
		\ad_{\log \Phi}^n J_{-k} &= \ad_{\log \Phi}^n \Gamma_{-k} = 0 \,, \qquad n > 2 \,.
	\end{align}
	The quadratic adjunction being non-zero is a consequence of the zero mode $\Gamma_0$. So we get
	\begin{equation}
		\begin{split}
			&\Phi J (\su{z}) \Phi^{-1} \\
			&=
			\Big( \Id + \ad_{\log \Phi} + \frac{1}{2} \ad_{\log \Phi}^2 ) J (\su{z})
			\\
			&=
			\sum_{ k \in \Z} [dz] \frac{J_k}{z^{k+1}} + [d\theta] \frac{\Gamma_k}{z^k} 
			+
			[dz]\sum_{k > 0} \Big( \sum_{l \geq 0} \frac{B_{kl}}{l} J_l + \phi_{kl} \Gamma_l \Big) z^{k-1}  + [d\theta] \sum_{k \geq 0} \Big( \sum_{l \geq 0} \frac{\phi_{lk}}{l} J_l + V_{lk} \Gamma_l \Big) z^k 
			\\
			&\qquad 
	       +\frac{1}{2}[dz] \sum_{k,l >0} \phi_{k0}\Big( \frac{\phi_{l0}}{l} J_l + V_{l0} \Gamma_l \Big)  z^{k-1} +\frac{1}{2}[d\theta] \sum_{k,l >0} V_{0k}\Big( \frac{\phi_{l0}}{l} J_l + V_{l0} \Gamma_l \Big)  z^k
			\\
			&=
			\sum_{k < 0}  [dz] \frac{J_k}{z^{k+1}} + [d\theta] \frac{\Gamma_k}{z^k} 
			+
			\Big( [d\theta ] + \sum_{l>0} [dz] \phi_{l0} z^{l-1} + [d\theta]V_{0l} z^l \Big)\Gamma_0 
			\\
			& \qquad +
			\sum_{k > 0}  \Big( \frac{[dz]}{z^{k+1}} + \sum_{l >0} [dz] \frac{B_{lk}}{k} z^{l-1} + \sum_{l \geq 0} \big( [d\theta]\frac{\phi_{kl}}{k} z^l  + [dz]\frac{\phi_{l0} \phi_{k0}}{2k} z^{l-1} + [d\theta] \frac{V_{0l}\phi_{k0}}{2k} z^l \big) \Big)J_k
			\\
			&\qquad +
			\sum_{k > 0} \Big( \frac{[d\theta]}{z^k} + \sum_{l > 0}  [dz] \phi_{lk} z^{l-1} + \sum_{l \geq 0} \big( [d\theta]V_{kl} z^l  +  [dz]\frac{V_{k0}\phi_{l0}}{2}  z^{l-1}  +  [d\theta]\frac{V_{k0}V_{0l}}{2} z^l \big) \Big) \Gamma_k
			\\
			&=
			\sum_{k < 0}  [dz] \frac{J_k}{z^{k+1}} + [d\theta] \frac{\Gamma_k}{z^k} 
			+
			 \Big( [d\theta ] + \sum_{l>0} [dz] \phi_{l0} z^{l-1} + [d\theta]V_{0l} z^l \Big)\Gamma_0
			\\
			& \qquad +
			\sum_{k > 0} \Big( \frac{[dz]}{z^{k+1}} + \sum_{l \geq 0}\big( [dz]\frac{2B_{lk} + \phi_{l0}\phi_{k0}}{2k} z^{l-1}  + [d\theta]\frac{2\phi_{kl} + V_{0l}\phi_{k0}}{2k} z^l \big) \Big)J_k
			\\
			&\qquad +
			\sum_{k > 0}  \Big( \frac{[d\theta]}{z^k} + \sum_{l \geq 0} \big( [dz] \frac{2\phi_{lk} + V_{k0}\phi_{l0}}{2} z^{l-1}  +[d\theta] \frac{2V_{kl} + V_{k0}V_{0l}}{2} z^l \big) \Big)\Gamma_k
			\\
			&=
			\sum_{k > 0} \Big( d\xi_k (\su{z}) J_{-k}  +  d\eta_k (\su{z}) \Gamma_{-k} \Big)+  d\eta_0 (\su{z}) \Gamma_0+
			\sum_{k >0} \Big(  d\xi_{-k} (\su{z})J_k+  d\eta_{-k} (\su{z})\Gamma_k \Big) \,.
		\end{split}
	\end{equation}
	Note that $ \Phi$ does not act directly on forms, and hence does not affect $ \omega_{0,1}$.
\end{proof}

We then define the \emph{Hamiltonians} $\su{H}^i$ as:
\begin{align}
	\su{H}^1 (\su{x} ) 
	&\coloneqq 
	[d\phi]\sum_{j \in \Z} \phi \frac{\Phi T J_{2j} T^{-1} \Phi^{-1}}{x^{j+1}} + \frac{\Phi T \Gamma_{2j+1} T^{-1} \Phi^{-1} }{x^{j+1}}
	\\
	\su{H}^2 (\su{x})
	&\coloneqq
	[d\phi]^3\sum_{j \in \Z} \Big(  \frac{ \phi \tilde{L}_{2j}}{4x^{j+2}} -  \frac{\tilde{G}_{2j+1}}{2\sqrt{2} x^{j+2}} \Big)
\end{align}

The quadratic Hamiltonian $\su{H}^2$ is closely related to the twisted super quadratic Casimir operator $\mc{C}^\sigma$ in \cref{TwistedSuperQuadraticCasimir}. More concretely, we have the following:

\begin{lemma}\label{Hamiltonians}
	Let $\sigma \colon \su{z} = (z \,| \,\theta)\mapsto(-z \,| -\theta)$. Then, the Hamiltonians $\su{H}^i$ can be written as follows:
	\begin{align}
		\begin{split}
			2\su{H}^1 (\su{x})
			&=
			\mc{J} (\su{z} ) + \mc{J} (- \su{z}) + \hslash \omega_{0,1} (\su{z}) + \hslash \omega_{0,1} (-\su{z})\,; \label{LinearHamiltonian}
		\end{split}
		\\
		\begin{split}
			2\su{H}^2 (\su{x})
			&=-\mc{C}^{\sigma}(\hslash^2\su{B}^{\textup{R}}\;+:(\mc{J}+\omega_{0,1})\otimes(\mc{J}+\omega_{0,1}):)(\su{z})\label{quadraticHamiltonian}
		\end{split}
	\end{align}
\end{lemma}

\begin{proof}
	The linear Hamiltonian $ \su{H}^1 $ can be obtained by taking only the $\sigma$-invariant part of $ J$ -- achieved by symmetrising -- and then conjugating by $ \Phi T$. The statement immediately follows from \cref{lem:DilatonPolarisationCurrent}.\par
	For the quadratic Hamiltonian $\su{H}^2$, notice that the conjugation by $ \Phi$ and $T$ commutes with the operation of the twisted super quadratic Casimir operator $\mc{C}^\sigma$. This is because the twisted super quadratic Casimir operator $\mc{C}^\sigma$ is a geometric operation in terms of coordinates $\su{z}$ while the conjugation is algebraic in terms of the modes $J_k,\Gamma_k$, and they are independent of each other. Therefore, it is sufficient to prove \eqref{quadraticHamiltonian} when $T=\Phi=1$.

The term coming from $ \su{B}^\textup{R}$ gives the central term $ -\frac{\hslash^2}{2} \delta_{n,0} $ in $ -2L_{2n}$ as below:
	\begin{equation}
		\begin{split}
			\hslash^2\mc{C}^\sigma (\su{B}^\textup{R})(\su{z})
			&=
			\frac12\hslash^2 [d\theta]^3 D_{\theta_1}^* \Big( \frac{\theta_2 \theta_1 z_1 z_2}{(z_1-z_2)^2 } - \frac{z_1+z_2}{2(z_1-z_2)} \Big) \Big|_{\su{z}_1 = - \su{z}_2=\su{z}}+(\su{z}\leftrightarrow\sigma(\su{z}))
			\\
			&=
			\frac12\hslash^2 [d\theta]^3 \Big(-\frac{\theta_2 z_1 z_2}{(z_1-z)^2 } +\frac{\theta_1 z_1 z_2}{(z_1-z_2)^2} \Big) \Big|_{\su{z}_1 = - \su{z}_2=\su{z}}+(\su{z}\leftrightarrow\sigma(\su{z}))
			\\
			&=- \hslash^2 \frac{[d\theta]^3 \theta}{2}=-\frac12\hslash^2[d\phi]^3\frac{\phi}{4x^2},
		\end{split}
	\end{equation}
	where we used $ [d\theta ] = \frac{[d\phi]}{\sqrt{2} z}$. The computations for other terms closely follow a combination of \cite[Proposition 4.14]{BCHORS20} and \cite[Theorem 4.4]{BO21}:
\begin{equation}
		\begin{split}
			\mc{C}^\sigma( \normord{J\otimes J})(\su{z})
			&=\frac12[d\theta]^3
			\normord{ \Big( D_\theta^*\Big( \sum_{l \in \Z} \theta \frac{J_l}{(z)^l} + \frac{\Gamma_l }{(z)^l} \Big) \sum_{k \in \Z}  \frac{-\theta J_k}{(-z)^k} + \frac{[\Gamma_k]}{(-z)^k} \Big)  }+(\su{z}\leftrightarrow\sigma(\su{z}))
			\\
			&=\frac12[d\theta]^3
			\normord{ \Big( \Big( \sum_{l \in \Z} \frac{J_l}{(z)^l} -\theta \frac{l\Gamma_l }{(z)^l} \Big) \sum_{k \in \Z}  \frac{-\theta J_k}{(-z)^k} + \frac{[\Gamma_k]}{(-z)^k} \Big)  }+(\su{z}\leftrightarrow\sigma(\su{z}))
			\\
			&=
			\frac12[d\theta]^3\sum_{k,l \in \Z}  \frac{\theta (-1)^k}{z^{k+l}} \normord{ -J_l J_k - l \Gamma_l \Gamma_k  }+ \frac{(-1)^k}{z^{k+l}}\normord{ J_l \Gamma_k } +(\su{z}\leftrightarrow\sigma(\su{z}))
			\\
			&=[d\theta]^3\sum_{n,k\in\mb{Z}}  \frac{\theta (-1)^k}{z^{2n}} \normord{ -J_{2n-k} J_k - (2n-k) \Gamma_{2n-k} \Gamma_k }+ \frac{(-1)^{k-1}}{z^{2n+1}}\normord{ J_{k} \Gamma_{2n+1-k} } \,,
		\end{split}
	\end{equation}

 	which after addition of the central term becomes
	\begin{equation}
		-2[d\theta]^3 \sum_{n\in \Z} \frac{ \theta L_{2n} }{z^{2n}} - \frac{G_{2n+1}}{z^{2n+1}} 
		= 
		-2[d\phi]^3 \sum_{n\in \Z} \frac{\phi L_{2n}}{4x^{n+2}} - \frac{G_{2n+1}}{2\sqrt{2}x^{n+2}} \,.
	\end{equation}
	This proves \cref{Hamiltonians}.
\end{proof}

Now we want to see what this means to the partition function of the Airy structure , which we write as
\begin{equation}
	Z = \exp \Big( \sum_{\substack{g,n \geq 0\\ n >0 \\ 2g - 2 + n \geq 0}} \frac{\hslash^{2g-2+n}}{n!} \sum_{i_1, \dotsc i_n \in {A}}  \prod_{k=1}^n x^{i_k}F_{g,n }[i_1, \dotsc, i_n ] \Big)\,,
\end{equation}
where $A= \{ B, F\} \times \N^* \cup \{ (F,0)\}$. Note that the parity of coefficients $F_{g,n}$ can be odd while the parity of $Z$ is even. As a consequence, the order of variables $x^i$ and coefficients $F_{g,n}$ matters due to possible signs, and our results below hold in this specific ordering.

This partition function is annihilated by $ \{J_{2j}, \Gamma_{2j-1}, \tilde{L}_{2j -s -1}, \tilde{G}_{2j-s}\}$ for all $j \geq1$. By \cref{LinearHamiltonian,LOLE}, $ H^1(\su{x} ) = \frac{1}{2} ( \mc{J} (\su{z}) + \mc{J}(-\su{z}) ) + \mc{O}([d \phi])$, and taking into account that $ J_0 = 0$, we get that

\begin{equation}\label{Gvanishing}
	G^i (\su{x}) \coloneqq Z^{-1} \su{H}^i( \su{x} )Z 
	=
	[d\phi]^{2i-1} \mc{O} \big( x^{\frac{(s-3)(i-1)}{2}}\big)
\end{equation}
We decompose $ G^i$ in homogeneous terms with respect to the $\hslash$-degree and the polynomial degree in $ x$ and $ \theta$, writing

\begin{equation}
	G^i (\su{x}) \eqqcolon \sum_{g,n} \frac{\hslash^{2g-2+n}}{n!} G^i_{g,n+1} (\su{x})\,,
\end{equation}
and define

\begin{equation}\label{def of Egn}
	\mc{E}^i_{g,n} (\su{x}; \su{z}_{[n]}) \coloneqq \prod_{j=1}^n \ad_{\hslash^{-1} \mc{J}(\su{z}_j)} G^i_{g,n} (\su{x})
\end{equation}
Furthermore, for $ k \in \Z$, let us define
\begin{equation}
	d\xi_{(B,k)} \coloneqq d\xi_k \,, \qquad d\xi_{(F,k)} \coloneqq d\eta_k \,.
\end{equation}
So $\{d\xi_k\}_{k\in A} = \{d\xi_k \}_{k\in \mb{Z}^*} \cup \{d\eta_k \}_{k\in \mb{Z}}$. For brevity, we denote by $A^{\pm} = \{ (B,\pm k), (F,\pm k) \mid k > 0 \}$ and $\hat A=\{(F,0)\}$, so that $A^+\sqcup A^-\sqcup\hat A=A$.
\begin{definition}
	Let $\su{S}$ be a local super spectral curve, and let  $Z$ be the partition function of the  associated partial Airy structure as in \cref{BOSuperAiry}. Then, we define the \emph{semi-stable correlators} as
	\begin{align}
		\omega_{0,2}(\su{z}_1, \su{z}_2) 
		&\coloneqq
		\su{B} (\su{z}_1, \su{z}_2) + \sum_{i, j \in {A^-}\sqcup \hat A} d\xi_{-i}(\su{z}_1)d\xi_{-j}(\su{z}_2)  F_{0,2}[i,j ]  \,;
		\\
		\omega_{g,n}(\su{z}_{[n]}) 
		&\coloneqq 
		\sum_{i_1, \dotsc, i_n \in {A^-\sqcup \hat A}}  \prod_{k=1}^n d\xi_{-i_k} (\su{z}_k) F_{g,n} [i_1, \dotsc, i_{n}  ]
		\label{def:omegagn}
	\end{align}
\end{definition}

We note that, unlike in the ordinary case, $\omega_{0,2}$ is not only the bidifferential $\su{B}$, but it admits singular terms at ramification points.\par
Then, we arrive at one of the main results:

\begin{proposition}\label{prop:SLE}
	The $\mc{E}^i_{g,n} $ can be expressed as
		\begin{align}
		\mc{E}^1_{g,n} (\su{x}; \su{z}_{[n]})
		&=
		\omega_{g,n+1}(\su{z}, \su{z}_{[n]}) + \omega_{g,n+1} (-\su{z}, \su{z}_{[n]})
		\\
		\mc{E}^2_{g,n} (\su{x}; \su{z}_{[n]})
		&=
		\mc{C}^\sigma( \omega_{g-1,n+2})( \su{z}; \su{z}_{[n] })+\mc{C}^\sigma( \sum_{\substack{g_1 + g_2 = g \\ I \sqcup J = [n]}} \omega_{g_1, |I|+1} \otimes \omega_{g_2, |J|+1})(\su{z}; \su{z}_I; \su{z}_J).\label{E and C}
	\end{align}
	Furthermore, they satisfy \emph{abstract loop equations}
	\begin{equation}
		\mc{E}^i_{g,n} (\su{x}; \su{z}_{[n]})
		=
		[d\phi]^{2i-1} \mc{O} \big( x^{\frac{(s-3)(i-1)}{2}}\big)\,.
	\end{equation}
\end{proposition}

\begin{proof}
    First, we note that $\mc{C}^\sigma$  in \eqref{E and C} is taken with respect to the first entry of $\omega_{g_1, |I|+1}$ and $\omega_{g_2, |J|+1}$ We also note that although each summand $ \omega_{g_1, |I|+1} \otimes \omega_{g_2, |J|+1}$ in $\mc{E}^2_{g,n} $ is not necessarily a symmetric bidifferential, the sum is and thus the operation of $\mc{C}^\sigma$ is well-defined.	

	The proof is standard, see e.g. \cite[Sections 4.3-4]{BKS23} and \cite[Appendix A.3]{BO21} for more details. In short, we consider the shape of the Hamiltonians from \cref{Hamiltonians}, and write down the vanishing from the Airy structure \cref{Gvanishing} for a particular ($\hslash$ and polynomial) degree. Because $ \su{H}^2$ is quadratic in derivatives, the equation for $ \mc{E}^2$ has terms of the kinds $ (D \log Z)^2 $ and $ D^2 \log Z$. As the proof is based on a tedious computation almost parallel to those in \cite{BO21,BKS23}, we leave it to the readers as an exercise.
\end{proof}

\subsection{Multiple components}\label{sec:multiple components}

We will now consider several ramification loci by considering a local spectral curve which consists of $r$ disks with a unique ramification of $ \su{x}$ on each. Let us label $\mc{R}=\sum_{a=1}^r[p_a]$ and recall from \cref{prop:NS-R SCmap} that we can always find local Ramond coordinates $\su{z}_a$ such that $\su{x}=(z_a^2 \, | \, \sqrt{2}z_a\theta_a)$ for each $a \in [r]$. Recall from \cref{sec:partial SAS} the definition of $ A_r$, which admits a decomposition $A_r = A_r^+ \sqcup A_r^- \sqcup \hat A_r$. All the parameters in $\su{B}$ in \eqref{DefSuperB} admit the extra index $a \in [r]$ for all $k_i\in\mb{Z}$ with constraints \eqref{constraints for multi components} which are encoded in the expansion of $\su{B}(\su{z}_1,\su{z}_2)$ near $z_1 = p_{a}$ and $z_2 = p_{b}$ for $a,b \in [r]$ respectively as
\begin{equation} \label{DefSuperBmulti}
		\begin{split}
		\su{B} ( \su{z}_{a1} , \su{z}_{b2} ) 
		&=
		\delta_{ab}\su{B}^{\textup{R}} (\su{z}_{a1}, \su{z}_{b2}) +
		\sum_{k,l > 0}^\infty[dz_{a1}] [dz_{b2}] B_{kl}^{ab} z_{a1}^{k-1} z_{b2}^{l-1} 
		\\
		&\quad
		+ \sum_{k> 0, l \geq 0}[dz_{a1}] [d\theta_{b2} ] \left( 1- \frac{\delta_{0l}}{2}\right) \phi_{kl}^{ab} z_{a1}^{k-1} z_{b2}^l 
		\\
		&\quad+ 
		\sum_{k \geq 0,l>0} [d\theta_{a1}] [dz_{b2}] \left(1- \frac{\delta_{k0}}{2}\right) \phi_{lk}^{ba}z_{a1}^k z_{b2}^{l-1}
		\\
		&\quad
		+ \sum_{k,l \geq 0} [d\theta_{a1}] [d\theta_{b2}] \left(1-\frac{\delta_{k,0}}{2}-\frac{\delta_{l,0}}{2}\right) V_{kl}^{ab}z_{a1}^k z_{b2}^l  \,,
	\end{split}
\end{equation}

The basis $\{d\xi_i\}_{i\in A}$ is also extended to $\{d\xi_i\}_{i\in A_r}=\{d\xi_{a,k},d\eta_{a,k}\}_{a\in[r],k \in \mathbb{Z}}$ (cf. \cite[Section 5]{BBCCN18}), which is given as below
\begin{align}
	d\xi_{a,k} (\su{z})
	&=
	[dz_a]z_a^{k-1} 
	\\
	d\eta_{a,k} (\su{z})
	&=
	[d\theta_a]z_a^k 
	\\
	d\xi_{a,-k} (\su{z})
	&=
	\frac{[dz_a]}{z_a^{k+1}} + \sum_{b\in[r]}  \sum_{l \geq 0} \Big( [dz_b] \frac{2 B_{kl}^{ab} + \phi_{l0}^{ba}\phi_{k0}^{ab}}{2k} z_b^{l-1}+ [d\theta_b] \frac{2\phi_{kl}^{ab} + V_{0l}^{ab}\phi_{k0}^{ab}}{2k} z_b^l\Big)
	\\
	d\eta_{a,0} (\su{z})
	&=
	[d\theta_a ] + \sum_{b\in[r]} \sum_{l>0} \Big( [dz_b] \phi_{l0}^{ba} z^{l-1}_b + [d\theta_b]V_{0l}^{ab} z^l_b\Big)
	\\
	d\eta_{a,-k} (\su{z})
	&=
	\frac{[d\theta_a]}{z_a^k} + \sum_{b\in[r]} \sum_{l \geq 0} \Big( [dz_b]\frac{2\phi_{lk}^{ba} + V_{k0}^{ab}\phi_{l0}^{ba}}{2} z_b^{l-1} + [d\theta_b]\frac{2V_{kl}^{ab} + V_{k0}^{ab}V_{0l}^{ab}}{2} z_b^l \Big)
\end{align}
With this basis, we define
\begin{align}
	\omega_{g,n}(\su{z}_{[n]}) 
	&\coloneqq
	\sum_{i_1, \dotsc, i_n \in {A_r^-}\sqcup \hat A_r} \prod_{k=1}^n d\xi_{-i_k} (\su{z}_k) F_{g,n} [i_1, \dotsc, i_{n}  ] \label{stable wgn}
	\\
	\omega_{0,2}(\su{z}_1, \su{z}_2) 
	&\coloneqq
	\su{B} (\su{z}_1, \su{z}_2) + \sum_{i, j \in {A_r^-}\sqcup \hat A_r}   d\xi_{-i}(\su{z}_1)d\xi_{-j}(\su{z}_2) F_{0,2}[i,j ]  \,.\label{unstable wgn}
\end{align}
Note that for $z_{a1},z_{b2}$ with $a\neq b$, $\su{B} (\su{z}_1, \su{z}_2)$ is written as a linear combination of $\{d\xi_i\}_{i\in A_r^+}$. So when one would like to consider the $d\eta_0$-dependence in $\omega_{0,2}$ it has to be captured by $F_{0,2}$ -- see \cref{sec:Weber} for a concrete example.

As $\sigma_a$ is only locally defined near $p_a$, formally we assume  it acts trivially on other coordinates, i.e., $\sigma_a \colon (z_b \, | \, \theta_b)=(-1)^{\delta_{a,b}}(z_b \, | \, \theta_b)$. Since the modes $\{J_{a,k},\Gamma_{a,k},L_{a,2k},G_{a,2k+1}\}$ for each $a\in[r]$ gives linear and quadratic loop equations respectively, \cref{prop:SLE} is extended as below. The proof is tedious but parallel to the proof of \cref{prop:SLE}, and hence we omit it.

\begin{theorem}\label{thm:SLE}
	For each $a\in[r]$, let us define\footnote{Although we derived \cref{prop:SLE} from the definition \eqref{def of Egn}, we can equivalently define $\mc{E}^{1,a}_{g,n} $ as in \cref{thm:SLE} and then show that they can be written in terms of the hamiltonian as \eqref{def of Egn}. Indeed, this is how one can prove the theorem.} $\mc{E}^{i,a}_{g,n} $ as
	\begin{align}
		\mc{E}^{1,a}_{g,n} (\su{x}; \su{z}_{[n]})
		&\coloneqq
		\omega_{g,n+1}(\su{z}, \su{z}_{[n]}) + \omega_{g,n+1} (\sigma_a(\su{z}), \su{z}_{[n]})
		\\
		\mc{E}^{2,a}_{g,n} (\su{x}; \su{z}_{[n]})
		&\coloneqq
		\mc{C}^{\sigma_a}( \omega_{g-1,n+2})( \su{z}; \su{z}_{[n] })+ \mc{C}^{\sigma_a}(\sum_{\substack{g_1 + g_2 = g \\ I \sqcup J = [n]}} \omega_{g_1, |I|+1} \otimes \omega_{g_2, |J|+1})(\su{z}; \su{z}_I; \su{z}_J)\label{QuaLE}
		\end{align}
	They satisfy \emph{abstract loop equations}
	\begin{equation}
		\mc{E}^{i,a}_{g,n} (\su{x}; \su{z}_{[n]}) 
		=
		[d\phi]^{2i-1} \mc{O}\big( x^{\frac{(s-3)(i-1)}{2}}\big)\,.\label{ALE}
	\end{equation}
\end{theorem}

In the literature, \cref{thm:SLE} is sometimes referred to as existence of a solution of super loop equations. That is, if one assumes that a sequence of symmetric multidifferentials satisfies super loop equations, then it is \emph{a priori} not guaranteed if such a sequence actually exists. Indeed, it is hard to prove existence even in the ordinary case if one only studies loop equations geometrically. However, we have constructed an Airy structure in such a way that its partition function(s) do yield such a solution. This is the power of Airy structures.

One remark is that, because our setup is in terms of partial Airy structures, the partition function is not unique, (cf. \cref{prop:partial Airy}) where in this case $F_{g,n}(\hat A)$ are not determined for $ \hat{A} = \hat{A}_r = [r] \times \{ (F,0) \}$.

\subsubsection{Global constraints}\label{sec:global constraints}
Recall from \cref{prop:partial Airy} that $F_{g,n}(i_1,...,i_n)$ for $i_1,...,i_n\in \hat A$ remain undetermined. When we take a global super spectral curve of $\deg\mathcal{R}=2r$, we know that the space of odd holomorphic one-forms is $r$-dimensional, so $\{[d\theta_a]\}_{a\in[2r]}$, or rather their images under projection $ \{ d\eta_{a,0} \}_{a \in [2r]}$, are mutually dependent. More concretely, let us decompose $\mathcal{R}=\mathcal{R}_p + \mathcal{R}_q$ based on the Torelli marking of a global super spectral curve. Then, since $\alpha_\mu(\su{B})=0$, we find that
\begin{equation}
\forall \mu\in[r]\quad d\eta_{p_\mu,0} + i d\eta_{q_\mu,0} = 0.\label{global constraints for zero modes}
\end{equation}

How can we capture such global constraints in terms of partial Airy structures? For each Ramond puncture $p$, let us define
\begin{equation}
	\tilde\Gamma_{p,0}
	=
	\hslash\left(\frac{\theta^{p,0}}{2}-\partial_{\theta^{p,0}}\right) \,.
\end{equation}
The relative sign is opposite to $\Gamma_0$ in \eqref{HC modes}. As a consequence, they commute not only with each other but also with all Heisenberg and Clifford modes including $\Gamma_{0,p}$. Since they all square to $\frac12\hslash^2$, however, one cannot naively add $\tilde \Gamma_0$ to a partial Airy structure, as it breaks the left-ideal condition in \cref{def:PSAS}. Rather, we consider a pair $(\tilde\Gamma_{p_\mu,0},\tilde\Gamma_{q_\mu,0})$ based on Torelli marking, and define $\tilde\Gamma_{\alpha_\mu,0} \coloneqq \tilde\Gamma_{p_\mu,0} + c \tilde\Gamma_{q_\mu,0}$. One then finds
\begin{equation}
[\tilde\Gamma_{\alpha_\mu,0}, \tilde\Gamma_{\alpha_\mu,0}]=0\quad\Leftrightarrow\quad c=\pm i.
\end{equation}
Therefore, whenever we set $c=\pm i$, we are free to add $\tilde\Gamma_{\alpha_\mu,0}$ to the partial Airy structure of \cref{thm:SASmain}. The new generator $\tilde\Gamma_{\alpha_\mu,0}$ imposes the following condition on $F_{g,n}$, except for $F_{0,2}$,
\begin{equation}
\partial_{\theta^{p_\mu,0}}F_{g,n}(x,\theta)+\pm \partial_{\theta^{q_\mu,0}}F_{g,n}(x,\theta)=0\label{global constraints for Fgn}
\end{equation}
which is equivalent to setting $\theta^{p_\mu,0}+\pm i \theta^{q_\mu,0}=0$ on $F_{g,n}$. To match with the global constraints \eqref{global constraints for zero modes}, we should choose $c=+ i$. The opposite choice corresponds to the case where $\alpha_\mu$ and $\beta_\mu$ are swapped. Notice that \eqref{global constraints for Fgn} with $c=+i$ indicates that all stable $\omega_{g,n}$ have vanishing $\beta$-cycles, not $\alpha$-cycles
\begin{equation}
\beta_\mu(\omega_{g,n})=0.
\end{equation}

For $F_{0,2}$, we find
\begin{equation}\label{blobs for F02}
	\forall\mu\in[r] \,,
	\qquad
	F_{0,2}
	\begin{bmatrix}
		p_\mu & q_\mu  
		\\
		0 & 0 
	\end{bmatrix}
	=
	i \,,
\end{equation}
which is a consequence of contributions from the first term in $\tilde\Gamma_{p_\mu,0},\tilde\Gamma_{q_\mu,0}$. We will show in \cref{sec:Weber} that the constraint on $F_{0,2}$ \eqref{blobs for F02} indeed appears in an example on a global super spectral curve.


\section{Super Topological Recursion}
\label{sec:SuperTR}

Since \cref{thm:SLE} shows that there exists a solution of super loop equations, and \cref{sec:global constraints} shows the solution space is very constrained for global spectral curves, the next question is: how do we compute such solutions? Can we find a similar formula to the one of Eynard--Orantin \cite{EO07}? It turns out that the recursive formula is more complicated and at the same time richer due to supersymmetric effects.

\subsection{Blobbed topological recursion}\label{sec:blobbed TR}

Let us first recall how to solve loop equations in the ordinary setting following \cite{BS17}. The formula, called \emph{blobbed topological recursion}, is more general than the Eynard--Orantin topological recursion \cite{EO07}. 

Let $V^{\pm}=\textup{Span}(d\xi_{k})_{k\in A_r^\pm}$. In the ordinary setting, $A_r^+ $ and $ A^-_r $ partition $A_r$, hence there exist natural projections $\mc{H} \colon V\to V^+$ and $\mc{P}\colon V\to V^-$ which give the holomorphic and principal part of a differential $\omega$ in a chosen basis $\{d\xi_{B,a,k}\}$. One can naturally extend such projections to multidifferentials $\omega_n$ in an obvious way. Recall from \cref{sec:super spectral curves} that the principal projection is realised by local residue computations, whereas the holomorphic part is precisely the kernel of the residue operation. \cite{BS17} shows that the most general solution of abstract loop equations (of rank $2$, as first introduced in \cite{BEO15}) is given as below:

\begin{proposition}[\cite{BS17}]\label{prop:blobbed}
	For each $(g,n)$ with $2g-2+n\geq1$, fix $\varphi_{g,n}\in (V^+)^n$. Then, if a solution $\{\omega_{g,n}\}$ of ordinary abstract linear and quadratic loop with purely holomorphic parts equal to the $ \varphi_{g,n}$ exists, it is constructed as below:
	\begin{equation}
    	\omega_{g,n+1}(z_0,z_{ [n]})
    	=
    	\mc{H}_0\omega_{g,n+1}(z_0,z_{ [n]})+\mc{P}_0\omega_{g,n+1}(z_0,z_{ [n]}),
	\end{equation}
	where
	\begin{align}
    	\mc{H}_0 \omega_{g,n+1}(z_0,z_{ [n]})
    	&=
    	\phi_{g,n+1}(z_0,z_{[n]}) + \sum_{\substack{A'\sqcup B = [n]\\ B\neq\emptyset}} \mc{H}_{{0}\sqcup A'}\mc{P}_B\omega_{g,n}(z_0,z_{ [n]}),\label{H1rec}
    	\\
       	\mc{P}_0 \omega_{g,n+1}(z_0,z_{ [n]})
    	&=
    	\sum_{a \in R}\Res_{z=a} \frac{\int^z_{a}\omega_{0,2}(z_0,\cdot)}{\Delta_a\omega_{0,1}(z)}{\rm Rec}^a_{g,n+1}(z,z_{ [n]}),\label{P1rec}
    \end{align}
	and moreover $\Delta_a\omega_{0,1}(z)=\omega_{0,1}(z)-\omega_{0,1}(\sigma_a(z))$, and for  $\mc{Q}_{g,n}(z;z_{[n]})$ in \eqref{QLE}, we define for each $a\in R$
    \begin{equation}
     	{\rm Rec}^a_{g,n+1}(z;z_{[n]})
     	=
     	\mc{Q}_{g,n}(z;z_{[n]})|_{\sigma=\sigma_a} + \big( \Delta_a\omega_{0,1}(z) \big) \omega_{g,n+1}(z,z_{ [n]}) \,.
    \end{equation}
\end{proposition}

The set of formulae in \cref{prop:blobbed} is called \emph{blobbed topological recursion}. Note that \eqref{H1rec} makes sense because $\omega_{g,n+1}$ is a symmetric differential, and hence as long as $B \neq \emptyset$ one can apply \eqref{P1rec} with respect to one variable (not necessarily $z_0$) in $z_B$. The purely holomorphic parts $\varphi_{g,n+1}$, which are the inaccessible parts by the residue formula \eqref{P1rec}, are called \emph{blobs}. When all blobs vanish, i.e. when $\omega_{g,n}=\mc{P}_{[n]}\omega_{g,n}$, then $\omega_{g,n}$ precisely corresponds to $\omega_{g,n}$ of Eynard--Orantin. Also, we note that existence of a solution with arbitrary blobs is not proved, though some explicit models are known -- see \cite{BS17} and references therein.

Let us briefly review how to derive blobbed topological recursion from abstract loop equations in the ordinary setting. 
 The quadratic loop equations \eqref{QLE} together with the linear loop equations \eqref{LLE} imply that
\begin{equation}
   -\Delta_a\omega_{0,1}(z) \Delta_a \omega_{g,n+1}(z,z_{ [n]})+ {\rm Rec}^a_{g,n+1}(z;z_{[n]})
   =
   (dz)^2 \mc{O}(z^{s-1}).\label{QLE1}
\end{equation}
A crucial point is that the right-hand-side has the same order of zero as $ \Delta_a \omega_{0,1}$ at $z=0$, hence dividing by it, we get:
\begin{equation}
   -\omega_{g,n+1}(z,z_{\llbracket n \rrbracket})+ \frac{{\rm Rec}^a_{g,n+1}(z;z_{[n]})}{\Delta_a \omega_{0,1}(z)}
   =
   dz\, \mc{O}(1) \,.\label{QLE2}
\end{equation}
Therefore, the projection operator kills the unknown right side, and
\begin{equation}\label{QLE3}
	\begin{split}
		0
		&=
		\sum_{a\in\mc{R}}\Res_{z=a}\int^z_{a}\omega_{0,2}(z_0,\cdot)\;\;dz \mc{O}(1)
		\\
		&=
		\sum_{a\in\mc{R}}\Res_{z=a}\int^z_{a}\omega_{0,2}(z_0,\cdot)\;\left(-\omega_{g,n+1}(z,z_{[n]})+\frac{{\rm Rec}^a_{g,n+1}(z,z_{ [n]})}{\Delta_a\omega_{0,1}(z)}\right)
		\\
		&=
		- \mc{P}_0 \omega_{g,n}(z_0,z_{\llbracket n \rrbracket})+\sum_{a\in\mc{R}}\Res_{z=a}\int^z_a\omega_{0,2}(z_0,\cdot)\;\frac{{\rm Rec}^a_{g,n+1}(z,z_{ [n]})}{\Delta_a \omega_{0,1}(z)} \,.
	\end{split}
\end{equation}

The power of the (blobbed) topological recursion formula is that although we do not know \emph{a priori} what the right-hand sides of loop equations are, we can kill the unknown right-hand side contributions by taking residues at every ramification point.

It turns out that super topological recursion admits a similar structure to blobbed topological recursion, even without blobs. This is essentially due to the existence of zero modes inducing $\{d\eta_i\}_{i\in \hat A_r}$ which are holomorphic along the Ramond divisor $\mc{R}$, (see remark \ref{remark:4.1.1})

 This motivates us to introduce two more new notions, \emph{simple super spectral curves}, and \emph{regularised fundamental bidifferential}. Let us consider them one by one.

\subsection{Simple super spectral curves}\label{sec:simple curves}

In order to motivate the notion of simple spectral curve which we shall introduce in this section, we focus on a degree one Ramond divisor $\mc{R}$ for brevity, because it is straightforward to generalise the arguments below to a higher degree. In this example, for brevity of notation, we drop $a$ from $\sigma$ and $\Delta$.

A constructive approach to solve super loop equations (\cref{prop:SLE} and \cref{thm:SLE}) is to follow the classical approach outlined above.  For this purpose, let us isolate terms involving $\omega_{0,1}$ from $\mc{E}^2_{g,n}$, which are
\begin{equation}
    \mc{C}^\sigma(\omega_{0,1}\boxtimes\omega_{g,n+1}+\omega_{g,n+1}\boxtimes\omega_{0,1})
    =
    \frac14 \mc{C}^\sigma(\Delta\omega_{0,1}\boxtimes\Delta\omega_{g,n+1}+\Delta\omega_{g,n+1}\boxtimes\Delta\omega_{0,1})+[d\phi]^3\mc{O}(x^{\frac{s-3}{2}})\label{QLE5}
\end{equation}
where $\Delta\omega(z) = \omega(z) - \omega(\sigma(z))$ as before and we used the linear loop equation at the equality.  One then immediately notices that it does not make sense to divide by $\omega_{0,1}$ in the super setting. This is simply because the resulting object is not a section of $\mc{Q}'$ and taking a residue is not defined. This is a clear contrast to the ordinary case \eqref{QLE2}.

A next reasonable attempt is to consider the inverse of $\mc{C}^\sigma$. In fact, by writing $\Delta\omega_{0,1}=[dx]\Delta y+[d\phi]\Delta \lambda$ let us define a differential operator
\begin{equation}
    \mc{G}
    \coloneqq
    (\Delta y+\phi\Delta\lambda') + (\Delta\lambda+\phi \Delta y)D_\phi
\end{equation}
Then for $\Delta \omega_{g,n+1}=[d\phi](\phi\Delta Y+\Delta \Lambda)$, it can be easily shown that
\begin{equation}
    [d\phi]^3\mc{G}(\phi\Delta Y+\Delta \Lambda)
    =
    \mc{C}^\sigma(\Delta\omega_{0,1}\boxtimes\Delta\omega_{g,n+1}+\Delta\omega_{g,n+1}\boxtimes\Delta\omega_{0,1}).
\end{equation}
It turns out that $\mc{G}$ admits an inverse $\mc{G}^{-1}$ which can be written explicitly as an differential operator with respect to $D_\phi$ whose coefficients are rational functions of $y,\lambda,\lambda',\phi$. Can one use $\mc{G}^{-1}$ to solve loop equations because this is an analogous operation to `dividing by $\omega_{0,1}$' in the ordinary setting? Unfortunately, one finds that for generic $\lambda$, $[d\phi]\mc{G}^{-1}(\mc{O}(x^{\frac{s-3}{2}}))$ becomes singular along the Ramond divisor, and as a consequence, one cannot kill the right-hand side of loop equations by taking residues.

At the moment, we do not know any elegant way of solving super loop equations. Nonetheless, we provide one way which works in a slightly more restricted setting:
\begin{definition}\label{def:admissible}
We say that a super spectral curve is \emph{simple} if, at every Ramond puncture, the expansion of $\Delta_a\tilde y\coloneqq\Delta_a y+\frac{\Delta_a\lambda'\Delta_a\lambda}{2\Delta_a y}$ in any local Ramond coordinate is of order $\mc{O}(z^{s-2})$ with invertible leading coefficient.
\end{definition}

Let us motivate why we consider simple super spectral curves. Using linear loop equations, we can write the quadratic loop equation $\mc{E}^2_{0,0}(\su{x})$ in terms of $\Delta_a \tilde y$ and $\Delta_a \lambda$ as
\begin{equation}
(\Delta_a \tilde y)^2=\mc{O}(x^{\frac{s-3}{2}}),\quad \Delta_a \tilde y\Delta_a \lambda=\mc{O}(x^{\frac{s-3}{2}}).
\end{equation}
By definition of $s$, the condition for simple spectral curves is always satisfied when $s=1$. On the other hand, when $s=3$, for instance, we can have a scenario with odd parameters $\kappa_1,\kappa_2$:
\begin{equation}
\Delta_a \tilde y=\frac{\kappa_1\kappa_2}{2z}+z+\mc{O}(z^3),\quad (\Delta_a \tilde y)^2=\kappa_1\kappa_2+\mc{O}(x)
\end{equation}
where we have chosen the local coordinate $z := x^{\frac{1}{2}}$.

Thus, $\Delta_a \tilde y$ has a pole with nilpotent coefficients, and in general, $\Delta_a \tilde y$ may have arbitrary higher order poles with more nilpotent coefficients as long as those coefficients are fine-tuned such that $(\Delta_a \tilde y)^2=\mc{O}(1)$. The condition of being simple imposes that $\Delta_a \tilde y$ does not admit such poles and it rather has a zero, i.e. the same behaviour as $\Delta_a y$ in the ordinary setting \footnote{$\Delta_a y$ itself may still have a pole because we can inversely write $\Delta_a y = \Delta_a \tilde y - \frac{\Delta_a \lambda'\Delta_a \lambda}{2\Delta_a \tilde y}$.}. Therefore, if one divides quadratic loop equations $\mc{E}^2_{g,n}$ by $[d\phi]^2\Delta_a\tilde y$, we have
\begin{equation}
\frac{\mc{E}^2_{g,n}}{[d\phi]^2\Delta_a \tilde y}=\frac{[d\phi]^3 \mc{O}(x^{\frac{s-3}{2}})}{[d\phi]^2\Delta_a \tilde y}=[d\theta]\mc{O}(1).\label{QLE4}
\end{equation}
which is a similar behaviour to the ordinary setting \eqref{QLE2}. This explains why we introduced the notion of simple super spectral curves. We will shortly present how to solve super loop equations with $\Delta_a \tilde y$.

\subsection{Regularised fundamental bidifferential and projection}\label{sec:regularised projection}

Recall from  \cref{rem:regularised projection} that the projection property behaves simpler with the regularised fundamental bidifferential $\su{B}^{\textup{reg}}$. Let us show another instance where $\su{B}^{\textup{reg}}$ plays an important role. 

Due to the anti-invariance with respect to $\sigma_a$, let us write the right-hand side of \eqref{QLE4} as
\begin{equation}
	\frac{[d\phi]^3 \mc{O}(x^{\frac{s-3}{2}})}{[d\phi]^2\Delta_a \tilde y}=[d\theta](\rho_0+\rho_1\theta z+\cdots)
\end{equation}
where $\rho_i$ are constants in $\su{z}$ but multidifferential in $\su{z}_{[n]}$ which are non-zero in general. Thus, as shown in \eqref{projection with zero mode}, the projection operator with $\su{B}$ gives:
\begin{equation}
	\Res_{z = 0} \int_0^{\su{z}} \su{B} (\su{z}_0 ,\cdot) \frac{[d\phi]^3 \mc{O}(x^{\frac{s-3}{2}})}{[d\phi]^2\Delta_a \tilde y}
	=
	\frac{1}{2} d\eta_0 (\su{z}_0) \rho_0
\end{equation}
This is a significant difference from the ordinary case as reviewed in \eqref{QLE3}, and $\rho_0$ is unknown for every $(g,n)$. Such unwanted contributions can be removed if we instead consider the projection with respect to the regularised fundamental bidifferential defined as below:
\begin{equation}
	\su{B}^{\textup{reg}}(\su{z}_1,\su{z}_2)=\su{B}(\su{z}_1,\su{z}_2)-\frac12d\eta_{0}(\su{z}_1)d\eta_{0}(\su{z}_2).
\end{equation}
That is,
\begin{equation}
	\Res_{\su{z} = 0} \int_0^{\su{z}} \su{B}^{\textup{reg}}(\su{z}_0\cdot)  d\theta\, \mc{O}(1)
	=
	0.
\end{equation}
We note that for a Ramond divisor of degree $r$, one has to consider the regularised fundamental bidifferential for projection at every Ramond puncture.

\begin{remark}
Because $\su{B}$ is globally defined on $\Sigma$ , $d\eta_{p,0}$ for each Ramond puncture $p$ is a globally holomorphic one-form. Thus, the regularised fundamental bidifferential $\su{B}^{\textup{reg}}$ is also globally defined. This is crucial because then the projection by $\su{B}^{\textup{reg}}$ also produces global meromorphic forms. However, note that $\su{B}^{\textup{reg}}$ is not symmetric any more, because the terms subtracted are anti-symmetric.
\end{remark}

Repeating the above discussions for every  Ramond puncture, the projection operator on a differential $\omega$ to $\textup{Span}(d\xi_i)_{i\in A_r^-}$ in the super setting is given by
\begin{equation}
\mc{P}(\omega)(\su{z}_0)=\sum_{a\in\mc{R}}\Res_{z = a} \int_a^{\su{z}} \su{B}^{\textup{reg}}(\su{z}_0\cdot)\; \omega(\su{z})
\end{equation}

On the other hand, recall that $\omega_{g,n}$ has a holomorphic part coming from $\{d\eta_{(a,F,0)}\}_{a\in[r]}$ by construction \eqref{def:omegagn}. Given the partition function of the partial Airy structure of \cref{thm:SASmain}, let us define \emph{zero mode blobs} $\psi_{g,n}$ by
\begin{equation}
	\label{def:psi_{g,n}}
 	\psi_{g,n}(\su{z}_{[n]})
 	\coloneqq
 	\sum_{i_1, \dotsc, i_n \in \hat{A}_r} \prod_{k=1}^n d\xi_{-i_k} (\su{z}_k) F_{g,n} [i_1, \dotsc, i_{n}  ] 
\end{equation}

Then, by construction, we have:

\begin{lemma}\label{lem:blob}
	Let $(\omega_{g,n})_{g,n}$ be the sequence of stable multidifferentials constructed as in \eqref{def:omegagn}. Then,
	\begin{equation}
		\omega_{g,n+1}(\su{z}_0,\su{z}_{[n]})
		=
		\mc{H}_0\omega_{g,n+1}(z_0,z_{[n]})+\mc{P}_0\omega_{g,n+1}(z_0,z_{[n]}) \,,
	\end{equation}
	where
	\begin{equation}
    	\mc{H}_0 \omega_{g,n+1}(z_0,z_{ [n] })
    	=
    	\psi_{g,n+1}(z_0,z_{[n]}) + \sum_{\substack{A'\sqcup B = [n] \\ B\neq \emptyset}} \mc{H}_{{0}\sqcup A'}\mc{P}_B\omega_{g,n+1}(z_0,z_{[n]}) \,.
	\end{equation}
\end{lemma}

\subsection{Recursion formula}

At the cost of having a good behaviour at the Ramond divisor as described in \cref{sec:simple curves}, a drawback of dividing \eqref{QLE5} by $[d\phi]^2\Delta_a \tilde y$ is that $\omega_{g,n+1}$ is not completely isolated from other terms, unlike the ordinary setting \eqref{QLE2}. Thus, we need one more technique in order to obtain a unique solution.

Let us rescale $\lambda$ by a formal even parameter $\gamma$ while keeping $x$, $\phi$, $\Delta_a \tilde y$ and $\su{B}$ unchanged. This means that $\Delta_a y$ is a quadratic polynomial in $\gamma$. Then, if we denote by $\Delta_a \omega_{g,n+1}=[d\phi](\phi\Delta_a Y+\Delta_a \Lambda)$ (the subscript is suppressed), we have:
\begin{multline}
	\frac{1}{2} (2-\phi D_\phi)\mc{C}^{\sigma_a}(\Delta_a \omega_{0,1}\boxtimes\Delta_a \omega_{g,n+1}+\Delta_a \omega_{g,n+1}\boxtimes\Delta_a \omega_{0,1})
	=
	-[d\phi]^2\Delta_a \tilde y\cdot\Delta_a \omega_{g,n+1}
	\\
	-\frac{\gamma}{2} [d\phi]^3\Bigg( \phi\bigg(\Delta_a \lambda'(z) \Delta_a \Lambda + \Delta_a \Lambda' \Delta_a \lambda(z) - \gamma\frac{(\Delta_a \lambda')(\Delta_a \lambda)}{\Delta_a \tilde y(z)} \Delta_a Y\bigg) - \gamma\frac{(\Delta_a \lambda')(\Delta_a \lambda)}{\Delta_a \tilde y(z)}  \Delta_a \Lambda +2 \Delta_a Y  \Delta_a \lambda(z) \Bigg) \,,
		\label{0,1 and g,n+1}
\end{multline}
where we declare that $ D_\phi ([d\phi]^3) = 0$ to make sense of $ D_\phi $ on sections of $ (\mc{Q}')^3$. This is a well-defined operation, as $ (x \, | \, \phi )$ are part of the data of our spectral curve.\par
The strange looking $(2-\phi D_\phi)$ is inserted in order to control the relative coefficients between $\phi$-dependent and $\phi$-independent terms such that the term $[d\phi]^2\Delta_a\tilde y \otimes \Delta_a\omega_{g,n+1}$ can be isolated. A crucial point is that the second line in \eqref{0,1 and g,n+1}, viewed as polynomials in $\gamma$, does not have constant term. Loosely speaking, $\gamma$ measures the degree of odd parameters in our family of super spectral curves. Our strategy is that we will write down a recursive formula not only with respect to $2g-2+n$ but also with respect to the degree of $\gamma$.

\begin{theorem}\label{thm:STR}
	Suppose we have a simple super spectral curve with rescaled $\Delta\lambda$ by a formal parameter $\gamma$, and a sequence of holomorphic multidifferentials $\{\psi_{g,n}\}_{g,n}^{2g-2+n\geq0}$.
    If there exists a sequence $\{\omega_{g,n}\}_{g,n}^{2g-2+n\geq0}$ of symmetric multidifferentials in $(\mc{Q}')^{\boxtimes n}[\gamma]$ which satisfies super loop equations and has the $ \psi_{g,n}$ as holomorphic parts, then it can be uniquely and recursively constructed by the following formulae
	\begin{equation}
    	\omega_{g,n+1}(\su{z}_0,\su{z}_{ [n]})
    	=
    	\mc{H}_0\omega_{g,n+1}(\su{z}_0,\su{z}_{ [n]})+\mc{P}_0\omega_{g,n+1}(\su{z}_0,\su{z}_{ [n]}),
	\end{equation}
	where
	\begin{equation}
    	\mc{H}_0 \omega_{g,n+1}(\su{z}_0,\su{z}_{ [n]})
    	=
    	\psi_{g,n+1}(\su{z}_0,\su{z}_{[n]}) + \sum_{\substack{A'\sqcup B = [n] \\ B\neq \emptyset}} \mc{H}_{{0}\sqcup A'}\mc{P}_B\omega_{g,n+1}(\su{z}_0,\su{z}_{ [n]}),\label{SH1rec}
	\end{equation}
	\begin{equation}
       \mc{P}_0 \omega_{g,n+1}(\su{z}_0,\su{z}_{ [n]})
       =
       \sum_{p_a \in \mc{R}} \Res_{z=p_a} \frac{\int^z_{p_a} \su{B}_{0,2}^{\textup{reg}}(\su{z}_0,\cdot)}{[d\phi]^2\Delta_a\tilde y(z)}\textup{Rec}^a_{g,n+1}(\su{z};\su{z}_{ [n]}),\label{SP1rec}
    \end{equation}
    \begin{equation}
    	\textup{Rec}^a_{g,n+1}(\su{z}; \su{z}_{[n]})
    	=
    	\left(2-\phi D_\phi\right) \mc{E}^{2,a}_{g,n} ( \su{x} ; \su{z}_{[n]}) + \frac{1}{2} [d\phi]^2\Delta_a\tilde y(\su{z}) \Delta_a \omega_{g,n+1} (\su{z},\su{z}_{ [n] }) \,,
    \end{equation}
    and for any multidifferential $\Delta_a\omega(z;z_{[n]}) =\omega (z;z_{[n]})-\omega(\sigma_a(z);z_{[n]})$ for each $a\in[r]$. The formula is recursive in $2g-2+n+1\geq0$ and also in the degree of $\gamma$, with the initial condition $[\gamma^0]\omega_{0,2}=\su{B}$.
\end{theorem}

\begin{corollary}
	The solution exists at least when $\gamma=1$, and it coincides with the sequence of the semi-stable differentials of \eqref{stable wgn} and \eqref{unstable wgn} which satisfies super loop equations of \cref{thm:SLE}.
	\end{corollary}

\begin{proof}
	We will drop the $a$-dependence, i.e. the labelling of the Ramond divisor, for brevity as one has to simply take the sum over $a\in[p]$ to recover the contributions by replacing $\sigma$ with $\sigma_a$.
	
Due to linear loop equations, we have:
	\begin{equation}
		\mc{P}_0\Delta\omega_{g,n+1}(\su{z}_0,\su{z}_{[n]})
		=
		2\mc{P}_0\omega_{g,n+1}(\su{z}_0,\su{z}_{[n]}) .
	\end{equation}
	Therefore, quadratic super loop equations (\cref{thm:SLE}) imply that:
	\begin{equation}
		\mc{P}_0\omega_{g,n+1}(\su{z}_0,\su{z}_{[n]})
		=
		\mc{P}_0\frac{{\rm Rec}_{g,n+1}(\su{z};\su{z}_{[n]})}{[d\phi]^2\Delta\tilde y(\su{z})}.\label{SLE-1}
	\end{equation}
	Then, \cref{lem:blob} implies the formula in the theorem.

    Since $\omega_{g,n+1}$ appears in ${\rm Rec}_{g,n+1}$, 
    our final task is to show that the formula is recursive. Combining \cref{QLE5} with the observation in \eqref{0,1 and g,n+1}, one realises that for any $g$ and $n$ with $2g-2+n+1\geq0$, every term in ${\rm Rec}_{g,n+1}(\su{z};\su{z}_{[n]})$  involves either $\omega_{h,m+1}$ with $2h-2+m<2g-2+n$ or $\omega_{g,n+1}$ with lower degree in $\gamma$. Therefore, the formula is indeed recursive, starting with the case $(g,n+1)=(0,2)$, i.e. for $\omega_{0,2}$. Given a simple super spectral curve, one sets the initial condition of the recursion to $[\gamma^0]\omega_{0,2}=\su{B}$. Note that for fixed $g$ and $n$, the recursion in $\gamma$ stops after finitely many steps to get the full $\omega_{g,n+1}$ because $[\gamma^{k}]\omega_{g,n}$ are nilpotent whenever $k\geq1$ as our base space $T$ is finite dimensional.

    As $\omega_{g,n+1}$ obtained as above are polynomials in $\gamma$, we can evaluate at $ \gamma = 1$, which implies that $(\omega_{g,n})_{g,n}$ solves the same super loop equations as \cref{thm:SLE}. Uniqueness and existence of such a solution  are ensured by the argument of Airy structures (\cref{prop:partial Airy} and \cref{thm:SLE}). This proves the theorem and corollary.
\end{proof}

%
%

\begin{remark}
For a non-simple super spectral curve, a similar approach is still possible by splitting $\Delta y$ into the nilpotent singular part and the part whose leading coefficient is invertible of order $\mc{O}(z^{s-2})$, and also by introducing an appropriate degree analogous to $\gamma$. At the moment of writing, however, the authors do not know how to universally deal with such a general situation.
\end{remark}

\section{Super examples}
\label{sec:Examples}

To conclude, let us present two simple yet non-trivial examples.

\subsection{Super Airy curve}

\begin{definition}\label{def:Super Airy curve}
	The \emph{super Airy curve} $\su{S}_{\textup{Airy}}$ consists of the following data:
	\begin{itemize}
		\item $\Sigma=\mb{P}(1,1|0)$ from \cref{exa:P11|0};
		\item $\su{x}:\Sigma\to\mb{P}^{1|1}$: a degree $2$ superconformal map;
		\item $\su{B}$: the canonical fundamental bidifferential on $\Sigma$ from \cref{exa:BinR};
		\item $\su{y}=(y \, | \, \lambda)$: a pair of even and odd meromorphic functions such that $\omega_{0,1}\coloneqq[dx]y+[d\phi]\lambda$ satisfies the following equation with two odd parameters $\tau_0,\tau_1$\footnote{That is., this is a family of super spectral curves over the base $ \Spec \C [ \tau_0, \tau_1 ]$.}
		\begin{equation}
			\mc{C}^\sigma(\omega_{0,1}^{\boxtimes2})
			=
			[d\phi]^3\left(\frac12\phi x+\tau_1 x +\tau_0\right)\,;
		\end{equation}
		\item All zero mode blobs $\psi_{g,n}=0$.
	\end{itemize}
\end{definition}

In this example, many things can be computed explicitly. First of all, let us take global coordinates $\su{z}$ on $\Sigma$ such that $\mc{R} = \{ z= 0 \} \cup \{ z = \infty \}$. Because the map $ \su{x}$ is of degree $2$, and only ramified at $ \mc{R}$, the involution $\sigma$ is also globally defined by $\sigma \colon \su{z}\mapsto-\su{z}$. Next, we can solve the defining equation using \cref{SuperQuadraticCasimirInCoordinates} to find that $\su{x},\su{y}$ as well as $\tilde y$ can be explicitly given as below
\begin{equation}
	x(\su{z})
	=
	\frac{z^2}{2},
	\quad
	\phi(\su{z})
	=
	\theta z,
	\quad
	y(\su z)
	=
	z+\gamma^2\frac{\tau_0\tau_1}{2z^3},
	\quad
	\lambda(\su z)
	=
	\gamma\frac{\tau_0}{z}+\gamma\frac{\tau_1 z}{2},
	\quad
	\tilde y(\su z)
	=
	z,
\end{equation}
where we have introduced the rescaling parameter $\gamma$ in order to run the recursion formula and we do not necessarily set to $\gamma=1$ in this section in order to demonstrate how $\gamma$ appears in $\omega_{g,n}$. Since there are only two odd parameters $\tau_0,\tau_1$, all $\omega_{g,n}$ are polynomials in $\gamma$ of at most degree 2.

One can then run the super topological recursion\footnote{A heavily model-dependent Mathematica notebook may be shared upon request.} formula, and for example, we obtain:
\begin{multline}
	\omega_{0,2}(\su{z}_1,\su{z}_2)
	=
	\su{B}(\su{z}_1,\su{z}_2)+\gamma[dz_1][d\theta_2]\left(-\frac{\tau_0}{z_1^2z_2^2}+\frac{\tau_1}{4z_1^2}\right)
	\\
	+\gamma[d\theta_1][dz_2]\left(-\frac{\tau_0}{z_1^2z_2^2}+\frac{\tau_1}{4z_2^2}\right)-\gamma^2[dz_1][dz_2]\left(\frac{3\tau_0\tau_1}{2}\left(\frac{1}{z_1^2z_2^4}+\frac{1}{z_1^4z_2^2}\right)\right)
\end{multline}
We note that $\omega_{0,2}$ is symmetric  order by order in $\gamma$, not only when $\gamma=1$. At the next level, $\omega_{0,3}$ is already complicated due to effects by odd parameters $\tau_0,\tau_1$, but $\omega_{1,1}$ is still in a compact  form:
\begin{equation}
	\omega_{1,1}(\su{z}_1)
	=
	-\frac{[dz_1]}{8z^4}+\gamma[d\theta_1]\left(\frac{5\tau_0}{8z_1^6}-\frac{\tau_1}{16z_1^4}\right)+\gamma^2[dz_1]\frac{35\tau_0\tau_1}{16z_1^8}
\end{equation}
We note that if one runs the ordinary topological recursion of Eynard--Orantin \cite{EO07} on the Airy curve reduced from \cref{def:Super Airy curve}, one would obtain $\omega_{1,1}(z)=-\frac{dz}{16z^4}$ which is \emph{half} of the $\gamma$-independent term of $\omega_{1,1}$. This phenomenon was already observed in \cite{BCHORS20,BO21} and indeed, the purely bosonic part $\omega_{g,n}^{\textup{bos}}$ on the super Airy curve and the Eynard--Orantin differentials $\omega^{\textup{EO}}_{g,n}$ on the Airy curve are related by $\omega_{g,n}^{\textup{bos}}=2^g\omega_{g,n}^{\textup{EO}}$ by identification $[dz_i] = dz_i$.

\subsubsection{Towards the variational formula}

In the context of topological recursion, we often consider not just a spectral curve but a family of spectral curves over a base $T$. As long as the family avoids bad loci such as collisions of ramification points (cf. \cite{BBCKS23}), it was shown in \cite{EO07} (cf. also \cite{Osu24-1}) how $\omega_{g,n}$ vary under a deformation of parameters. This is called the \emph{variational formula}. It exhibits a relation between a derivative of $\omega_{g,n}$ and a certain contour integral of $\omega_{g,n+1}$.  Although general structures in the super setting are yet to be investigated, let us show a piece of evidence that a similar story may hold. 

Recall from \cref{sec:odd periods} that the odd periods are given by $ \alpha = \frac{1}{\sqrt{2}} (w_0 + i w_\infty ) $ and $ \beta = \frac{1}{\sqrt{2}} (w_0 - i w_\infty ) $ where $w_0,w_\infty$ are defined in \cref{def:super periods}. With this convention, we have:
\begin{equation}
	\alpha(\omega_{0,1})
	=
	2\sqrt{\pi i}\gamma\tau_0\eqqcolon2\pi i\gamma\alpha_0,\quad \beta(\omega_{0,1})=0.
\end{equation}
Similarly, let us consider a weighted residue:
\begin{equation}
	\Res_{z=\infty}\omega_{0,1}\frac{\phi}{x^{\frac12}}\frac{1}{x}
	=
	\gamma\tau_1
	\eqqcolon
	\gamma\alpha_1,
\end{equation}
Therefore, one can think of $\alpha_0,\alpha_1$ as moduli parameters associated to the odd $A$-periods as well as the  odd cycles at infinity --- this structure resembles to the notion of generalised cycles in \cite{Eyn23}. Furthermore, we have:
\begin{equation}
	\frac{\partial}{\partial\alpha_0}\omega_{0,1}(\su{z})
	=
	-\gamma\sqrt{\pi i}[d\theta]+\gamma^2\sqrt{\pi i}[dz]\frac{\tau_1}{2z^2},
	\quad
	\frac{\partial}{\partial\alpha_1}\omega_{0,1}(\su{z})
	=
	-\frac{\gamma}{2}[d\theta]z^2-\frac{\gamma^2}{2}[dz]\frac{\tau_0}{z^2},
\end{equation}
where the partial derivative indicates that we think of $\alpha_i$ as being independent of $\su{x}=(x|\phi)$ which is equivalent to being independent of $\su{z}=(z \, | \, \theta)$ for the super Airy curve.

On the other hand, we have:
\begin{equation}
	\alpha(\omega_{0,2}(\cdot,\su{z}))
	=
	\sqrt{\pi i}\gamma[d\theta]\frac{\tau_1}{2z^2},
	\quad
	\beta(\omega_{0,2}(\cdot,\su{z}))
	=
	\sqrt{\pi i}[d\theta]
\end{equation}
\begin{equation} 
\Res_{z=\infty} \omega_{0,2}(\cdot,\su{z})\frac{\phi}{x^{\frac12}}x=-\frac{1}{2}[d\theta]z^2-\frac{\gamma}{2}[dz]\frac{\tau_0}{z^2},
\end{equation}
where the dot $\cdot$ as the first argument of $\omega_{0,2}$ indicates that the period is computed with respect to the first variable. Therefore, in summary, we arrive at
\begin{align}
	\frac{\partial}{\partial\alpha_0}\omega_{0,1}(\su{z})
	&=
	\gamma(\alpha-\beta)(\omega_{0,2}(\cdot,\su{z})),\label{01alpha0}
	\\
	\frac{\partial}{\partial\alpha_1}\omega_{0,1}(\su{z})
	&=
	\gamma\Res_{z=\infty}\omega_{0,2}(\cdot,\su{z})\frac{\phi}{x^{\frac12}}x
\end{align}
By explicit computations, we have also checked that the followings hold:
\begin{align}
	\frac{\partial}{\partial\alpha_0}\omega_{0,2}(\su{z}_1,\su{z}_2)
	&=
	\gamma(\alpha-\beta)(\omega_{0,3}(\cdot,\su{z}_1,\su{z}_2),\label{02alpha0}
	\\
	\frac{\partial}{\partial\alpha_1}\omega_{0,2}(\su{z}_1,\su{z}_2)
	&=
	\gamma\Res_{z=\infty}\omega_{0,3}(\cdot,\su{z}_1,\su{z}_2)\frac{\phi}{x^{\frac12}}x.
\end{align}

Let us emphasise that these relation hold for $\omega_{0,2}$, not for $\su{B}$, suggesting that $\omega_{0,2}$ plays possibly a foundational role in deformation of super Riemann surfaces. We also note that all stable correlators $\omega_{g,n}$ have vanishing cycles on $\beta$, so the right-hand side of \eqref{02alpha0} can be written only in terms of $\alpha$-cycle, but \eqref{01alpha0} holds only when we take the difference $\alpha-\beta$. These results motivate us to expect that the \emph{variational formula} \cite{EO07} may hold in the super setting. We hope to return this perspective in the future.

\subsection{Super Weber curve}\label{sec:Weber}

\begin{definition}\label{def:Super Weber curve}
	The \emph{super Weber curve} $\su{S}_{\textup{Web}}$ consists of the following data:
	\begin{itemize}
		\item $\Sigma=\mb{P}(1,1|0)$;
		\item $\su{x}:\Sigma\to\mb{P}^{1|1}$: a superconformal map;
		\item $\su{B}$: the canonical fundamental bidifferential on $\Sigma$;
		\item $\su{y}=(y \, | \, \lambda)$: a pair of even and odd meromorphic functions such that $\omega_{0,1}\coloneqq[dx]y+[d\phi]\lambda$ satisfies the following equation with three odd parameters $\tau_0,\tau_1,\tau_2$ and the parameter $\gamma$
		\begin{equation}
			\mc{C}^\sigma(\omega_{0,1}^{\boxtimes2})
			=
			[d\phi]^3\left(\phi (x^2-1)+\gamma(\tau_2x^2+\tau_1 x +\tau_0)\right)\,;
		\end{equation}
		\item All zero mode blobs are set to zero, except for $F_{0,2}$.
	\end{itemize}
\end{definition}

Let us take global coordinates $\su{z}$ on $\Sigma$ such that $\mc{R} = \{ z= 0 \} \cup \{ z = \infty \}$. Then the involution $\sigma$ is also globally defined by $\sigma:\su{z}\mapsto-\su{z}$. Next, $\su{x}$ and $\tilde y$ can be explicitly given as below
\begin{equation}
	x(\su{z})
	=
	\frac{4+z^2}{4-z^2},
	\quad
	\phi(\su{z})
	=
	\frac{4\theta z}{4-z^2},
	\quad
	\tilde y(\su{z})
	=
	\frac{4z}{4-z^2}.
\end{equation}
With these, $y$ and $\lambda$ can be easily obtained.

Explicit computations for this example are complicated. Instead, we will explain how global constraints can be naturally encoded in local super topological recursion of \cref{thm:STR}.  Let $\su{z}_\pm$ be local coordinates centred at the ramification point $x^{-1}(\pm1)$ such that
\begin{equation}
	\su{x}(\su{z}_\pm)
	=
	( \pm1+\frac{z_\pm^2}{2} \, |\theta_\pm z_\pm ) .
\end{equation}
By solving the above equation for $\su{z}_\pm$ and expanding in local coordinates $\su{z}_\pm$, one can determine all dilaton shift and polarisation parameters. This is a recipe of how to find the partial Airy structure
\begin{equation}
	\mc{I}_{\text{Web}}
	=
	\big\{ J_{a,2j}, \Gamma_{a,2j-1}, \tilde{L}_{a,2j -s -1}, \tilde{G}_{a,2j-s} \, \big| \,a\in\{1,2\},\; j \geq1\big\}
\end{equation}
associated to $\su{S}_{\textup{Web}}$.

Let $(p,q)=(0,\infty)$ be the Torelli marking of $\Sigma$. Following \cref{def:periods for meromorphic forms}, one finds that for a meromorphic form $\omega$ on $\mc{Q}'$, we can write the periods explicitly as
\begin{equation}
	\alpha(\omega)
	=
	\sqrt{\pi i}\left(\Res_{z=0}-\Res_{z=\infty}\right)\omega\theta,
	\quad
	\beta(\omega)
	=
	\sqrt{\pi i}\left(\Res_{z=0}+\Res_{z=\infty}\right)\omega\theta.
\end{equation}
where $\su{z}$ here are global coordinates, not local ones $\su{z}_\pm$. Then, except for $\omega_{0,1}$ and $\omega_{0,2}$, we have
\begin{equation}
	\beta(\omega_{g,n})
	=
	0.
\end{equation}
This is consistent with the constraint \eqref{global constraints for Fgn}, but one can also explicitly see that $ \omega_{g,n}\theta$ for $2g-2+n\geq1$ have poles only at $z=0$ or $z=\infty$, and the sum of residues of meromorphic forms vanish on a compact surface.

Our remaining task is to check whether this is consistent with our choice of all blobs $\psi_{g,n}$ -- recall that principal parts are uniquely determined by super loop equations. If one expands the global $\su{B}^{\text{R}}$ at $z_1=0$ and $z_2=\infty$, we find
\begin{equation}
	\su{B}(z_1,z_2)
	=
	\frac{i}{2}[d\theta_{1}][d\theta_{2}]+\mc{O}(z_{1},z_{2}^{-1}).
\end{equation}
As noted in \cref{sec:global constraints}, $\su{B}^{\text{R}}$ in local $\su{B}$ is non-zero only along the diagonal, hence this implies that we would have to set the blob for $F_{0,2}$ as
\begin{equation}
	F_{0,2}\begin{bmatrix}
		0 & \infty  \\
		0 & 0 
	\end{bmatrix}
	=
	i
\end{equation}
This precisely coincides with the condition on $F_{0,2}$ shown in \eqref{blobs for F02}.

{\setlength\emergencystretch{\hsize}\hbadness=10000
\printbibliography}

\end{document}